\documentclass[a4paper,reqno]{amsart}
\usepackage[left=2cm,right=2cm]{geometry}

\usepackage{amsmath,amssymb,amsfonts,mathtools,mathrsfs}
\usepackage{amsthm}
\usepackage{framed}
\allowdisplaybreaks

\usepackage{circuitikz}
\usepackage{psfrag}
\usepackage{tikz-cd}

\usepackage{xcolor}
\usepackage{enumitem}
\usepackage{outlines}
\usepackage{url}
\usepackage{comment}
\usepackage{anyfontsize}
\usepackage{cite}
\usepackage{algorithmic}

\usepackage[colorlinks=true,linkcolor=blue,citecolor=blue,urlcolor=blue]{hyperref}
\usepackage[nameinlink,capitalise]{cleveref}

\usepackage{thm-restate}

\newtheorem{theorem}{Theorem}[section]
\newtheorem{lemma}[theorem]{Lemma}

\newtheorem{claim}[theorem]{Claim}
\newtheorem{corollary}[theorem]{Corollary}

\newtheorem{example}[theorem]{Example}

\theoremstyle{definition}
\newtheorem{definition}[theorem]{Definition}

\Crefname{lemma}{Lemma}{Lemmata}
\Crefname{prop}{Proposition}{Propositions}

\newcommand{\tuple}{\mathbf}  

\newcommand{\br}[1]{[#1]}
\newcommand{\edge}{\rightarrow}
\newcommand{\ledge}{\leftarrow}

\newcommand{\rplus}{\oplus}

\DeclareMathOperator{\Csp}{CSP}

\DeclareMathOperator{\Pol}{Pol}

\DeclareMathOperator{\Aut}{Aut}

\DeclareMathOperator{\val}{val}
\DeclareMathOperator{\oin}{in}
\DeclareMathOperator{\out}{out}

\DeclareMathOperator{\OR}{OR}
\DeclareMathOperator{\labels}{L}
\DeclareMathOperator{\Betw}{Betw}
\DeclareMathOperator{\comp}{S}
\DeclareMathOperator{\CSP}{CSP}

\newcommand{\sA}{\mathbb A}
\newcommand{\sB}{\mathbb B}

\newcommand{\sH}{\mathbb H}
\newcommand{\sI}{\mathbb I}

\newcommand{\subgraph}{\mathbb H}
\newcommand{\subdomain}{H}
\newcommand{\fingraph}{\mathbb J}
\newcommand{\findomain}{J}
\newcommand{\finsubgraph}{\mathbb K}
\newcommand{\finsubdomain}{K}
\newcommand{\weakcomp}{H}

\newcommand{\urgraph}{\mathbb G}
\newcommand{\urdomain}{G}

\newcommand{\edgea}{\edge}

\newcommand{\edgeo}{\edge}
\newcommand{\ledgeo}{\leftarrow}
\newcommand{\edgek}{\overset{k}{\rightarrow}}

\newcommand{\merge}{\vee}

\newcommand{\inj}{I_{\urdomain}}
\newcommand{\prinj}[1]{I_{#1}}

\newcommand{\fence}[2]{\mathbb{F}[#1,#2]}

\usepackage{orcidlink}
\newcommand{\Addresses}{
  \bigskip
  \noindent
  Johanna Brunar, Institut f\"ur Diskrete Mathematik und Geometrie, Technische Universit\"at Wien, Wien, Austria.\\
  \textit{Email}: \texttt{johanna.brunar@tuwien.ac.at}

  \medskip
  \noindent
  Marcin Kozik, Theoretical Computer Science Department, Jagiellonian University, Krak\'ow, Poland.\\
  \textit{Email}: \texttt{marcin.kozik@uj.edu.pl}

  \medskip
  \noindent
  Tom\'{a}\v{s} Nagy, Theoretical Computer Science Department, Jagiellonian University, Krak\'ow, Poland.\\
  \textit{Email}: \texttt{tomas.nagy@email.com}

  \medskip
  \noindent
  Michael Pinsker, Institut f\"ur Diskrete Mathematik und Geometrie, Technische Universit\"at Wien, Wien, Austria.\\
  \textit{Email}: \texttt{marula@gmx.at}
}

\title[The sorrows of a smooth digraph]{The sorrows of a smooth digraph:\\ the first hardness criterion for infinite directed graph-colouring problems}

\date{}

\begin{document}

\maketitle

\begingroup
\renewcommand\thefootnote{}\footnotetext{This is the author's accepted version of a paper published in the
Proceedings of the 40th Annual ACM/IEEE Symposium on Logic in Computer Science (LICS'25).
© 2025 IEEE. 
The final version is available at \url{https://doi.org/10.1109/LICS65433.2025.00037}.

Funded by the European Union (ERC, POCOCOP, 101071674). Views and opinions expressed are however those of the author(s) only and do not necessarily reflect those of the European Union or the European Research Council Executive Agency. Neither the European Union nor the granting authority can be held responsible for them. 
This research was funded in whole or in part by the Austrian Science Fund (FWF) [I5948]. For the purpose of Open Access, the authors have applied a CC BY public copyright licence to any Author Accepted Manuscript (AAM) version arising from this submission.
Partially funded by the National Science Centre, Poland under the Weave program grant no. 2021/03/Y/ST6/00171. For the purpose of Open Access, the authors have applied a CC BY public copyright licence to any Author Accepted Manuscript (AAM) version arising from this submission. 
The first author is a recipient of a DOC Fellowship of the Austrian Academy of Sciences at the Technische Universit\"at Wien.
}
\addtocounter{footnote}{-1}
\endgroup

\begin{center}
\textsc{Johanna Brunar\orcidlink{0009-0000-7229-0172}, Marcin Kozik\orcidlink{0000-0002-1839-4824}, Tom\'{a}\v{s} Nagy\orcidlink{0000-0003-4307-8556}, and Michael Pinsker\orcidlink{0000-0002-4727-918X}} \\[0.5em]
\end{center}

\vspace{1em}

\begin{abstract}
Two major milestones on the road to the  full complexity dichotomy  for finite-domain constraint satisfaction problems were Bulatov's proof of the dichotomy for conservative templates, and the structural dichotomy for smooth digraphs of algebraic length~1 due to Barto, Kozik, and Niven. 
We lift the combined scenario to the infinite, and  prove that any smooth digraph of algebraic length~1 pp-constructs, together with pairs of orbits of an oligomorphic subgroup of its automorphism group, every finite structure -- and hence its conservative graph-colouring problem is NP-hard -- unless the digraph has a pseudo-loop, i.e.~an edge within  an orbit. We thereby overcome, for the first time,  previous obstacles  to lifting  structural results for digraphs in this context from finite to $\omega$-categorical structures; the strongest  lifting results hitherto not going beyond a ge\-ne\-ra\-li\-sa\-tion of the Hell-Ne\v{s}et\v{r}il theorem for undirected graphs.  As a consequence, we obtain a new algebraic invariant of arbitrary $\omega$-categorical structures enriched by pairs of orbits which fail to pp-construct some finite structure.
\end{abstract}

\section{Introduction}\label{sect:Intro}

\subsection{Complexity and structural dichotomies for graphs}
For a graph $\sH$, the \emph{$\sH$-colouring problem}  is the computational problem where given a finite input graph $\sI$, one has to decide whether there exists a \emph{homomorphism} from $\sI$ to $\sH$, i.e.~a map sending edges to edges; a colouring, so to speak, of $\sI$ by the vertices of $\sH$ which is consistent with edges. This problem is a natural generalisation of the $n$-colouring problem, which is equal to the $\sH$-colouring problem when $\sH$ is the clique on $n$ vertices. Yet more general are (fixed-template) \emph{Constraint Satisfaction Problems (CSPs)}, which are parametrised by arbitrary relational structures $\sA$ (in this context called \emph{templates})  rather than just graphs: here again, one is given a finite input structure $\sI$ in the same signature as $\sA$, and has to decide whether or not there exists a homomorphism from $\sI$ into $\sA$, i.e.~a map preserving all relations. This problem is denoted by $\CSP(\sA)$, and so the $\sH$-colouring problem in particular is $\CSP(\sH)$; the notion allows to model many classical computational problems, such as 3-SAT, within a uniform framework. On the other hand, every problem $\CSP(\sA)$ is logspace-interreducible with the $\sH$-colouring problem for some directed graph $\sH$ (\cite{BulinDelicJacksonNiven}, improving upon~\cite{FederVardi}), and hence in this sense digraphs already contain the entire complexity spectrum   brought along by the study of general Constraint Satisfaction Problems.

A  systematic research programme on the complexity of  CSPs was initiated by Feder and Vardi in the 1990s with their seminal work~\cite{FederVardiSTOC,FederVardi}, although the earliest results stem back to the 1970s, when Schaefer proved a P/NP-complete complexity dichototomy for CSPs over a Boolean domain~\cite{Schaefer}. The continuous results of the thriving  research field culminated,  25 years after the start of the programme, in 2017 with two independent proofs by Bulatov~\cite{BulatovFVConjecture} and Zhuk~\cite{ZhukFVConjecture, Zhuk20}, respectively, of a P/NP-complete dichotomy for CSPs of arbitrary finite structures: even if P$\neq$NP, then
no problem $\CSP(\sA)$ can, for finite $\sA$, be of intermediate complexity, contrasting Ladner's theorem~\cite{Ladner}. Similar complexity dichotomies had been obtained earlier in particular for undirected graphs (\cite{HellNesetril}; see also the later proof in~\cite{Bulatov}),  and for \emph{smooth} digraphs, i.e.~digraphs with no sources and no sinks~\cite{BartoKozikNiven}.

Remarkably, the above   complexity dichotomies coincide with  elegant  structural dichotomies for the classes of relational structures concerned. In fact, over the years  numerous  structural dichotomies inspired by applications to CSPs have been unveiled,   usually taking the following form: either  a template $\sA$ has large  expressive power (as measured by its ability e.g.~to \emph{pp-interpret}, to \emph{pp-interpret with parameters}, or  to \emph{pp-construct} a certain other finite structure, or all finite structures, or an abelian group), or there is an obvious obstacle for this (e.g.~some form of symmetry of the structure in the form of a \emph{polymorphism}, or a constant tuple in a relation). Of course, the former case then corresponds to hardness in the sense of computational complexity  of the CSP (e.g. NP-hardness, or the non-applicability of a certain algorithm such as $k$-consistency) the latter to tractability (e.g. polynomial-time solvability). In the case of finite undirected non-bipartite graphs $\sH$ 
 (note that the bipartite case is just the 2-colouring problem),  Hell and Ne\v{s}et\v{r}il proved, in modern terminology, that either $\sH$ pp-constructs all finite structures -- and hence $\CSP(\sH)$ is NP-complete -- or $\sH$ contains a loop~\cite{HellNesetril}; the latter is an obvious obstacle to the pp-construction of any loopless graph, and implies that $\CSP(\sH)$ is trivial.
 
 The next milestone achievement in this direction concerned \emph{directed} graphs which are smooth and of \emph{algebraic length~1}
(a property which can be viewed as a generalisation of non-bipartiteness for digraphs): again, any such graph either pp-constructs all finite structures, or it has a loop. Subsequently, similar statements for higher-arity relations were obtained, for example  relations  invariant under cyclic shifts~\cite{cyclicterms}, the proofs of which unsurprisingly often  incorporate the  combinatorial principles crucial in the case of digraphs. Finally, the structural dichotomy  
underlying the theorem of Bulatov and Zhuk can be phrased as follows: every finite template $\sA$ either pp-constructs all finite structures, or it is invariant under a 4-ary operation $s$ on its domain (a so-called \emph{polymorphism}) satisfying the \emph{4-ary Siggers identity}  $s(a,r,e,a)=s(r,a,r,e)$ for all values of $a,r,e$~\cite{KearnesMarkovicMcKenzie}. 

\subsection{Infinity}\label{subsect:infinity}
When it comes to (countably) infinite graphs and structures,
the context in which  results of a similar flavour are conceivable at least in principle is that of \emph{$\omega$-categoricity}: structures $\sA$ with an \emph{oligomorphic} automorphism group $\Aut(\sA)$, i.e.~an automorphism group so large that its action on $k$-tuples has only finitely many orbits (called \emph{$k$-orbits}) for all $k$. This  might be paraphrased as $\sA$ being ``close to finite'' in the sense that it has only finitely many distinguishable $k$-tuples for every $k$. The property of $\omega$-categoricity is prominent in model theory,  and is considered  necessary and sufficient (in a certain precise sense, see~\cite{BodirskyNesetrilJLC,Topo-Birk, wonderland}) for the  applicability  of the \emph{algebraic approach} to CSPs, crucial for the success of finite-template CSPs.
\begin{example}
    The automorphism group of the  digraph $(\mathbb Q;<)$ given by the strict order of the rational numbers is oligomorphic: the orbit of any $k$-tuple under  $\Aut(\mathbb Q;<)$ is completely  determined by the order induced on its elements, for which there are only finitely many possibilities. The digraph $(\mathbb Q;<)$ is also a good example of an infinite graph inducing a natural CSP: namely, a finite digraph $\sI$ is $(\mathbb Q;<)$-colourable if and only if it is acyclic, and the latter property cannot be modelled this way by any finite template.
\end{example}

 Note that $\omega$-categoricity of a template $\sA$ can  not only be witnessed by $\Aut(\sA)$ itself, but by any oligomorphic subgroup $\Omega$ of $\Aut(\sA)$ (which might be a more natural permutation group to consider): the orbits of $\Aut(\sA)$ in its action on $k$-tuples are unions of the orbits of $\Omega$. 
 \begin{example}
     Consider the template $\sA:=(\mathbb Q;\Betw)$, where $\Betw$ is the ternary \emph{betweenness relation} which holds for a tuple $(x,y,z)$ precisely if $x<y<z$ or $z<y<x$. An instance $\sI$ of its CSP is a yes-instance if and only if its elements can be assigned values in  $\mathbb Q$ in such a way that all triples $(x,y,z)$ of elements of $\sI$ which belong to its (only) relation end up in the relation $\Betw$ after the assignment. This amounts to finding a (non-strict, since homomorphisms need not be injective) linear ordering of $\sI$ so that $y$ is strictly between $x$ and $z$ with respect to this ordering for all such triples $(x,y,z)$; the problem is thus  the classical (NP-complete)  \emph{betweenness problem} from the 1970s~\cite{Opatrny}. In this case, it might be more natural to consider the proper subgroup  $\Omega:=\Aut(\mathbb Q;<)$ of $\Aut(\mathbb Q;\Betw)$ rather than the latter group itself, since we are thinking of finding an ordering, but also since the group $\Aut(\mathbb Q;<)$ is natural and well-understood from a model-theoretic perspective.
 \end{example}
  In general,  if $\Omega\leq \Aut(\sA)$, we have that 
  the relations of $\sA$, invariant under its automorphisms, are unions of orbits of $\Omega$, and one can thus imagine them as finite objects modulo $\Omega$ (in the example above, $\Betw$ consists of two orbits with respect to $\Omega:=\Aut(\mathbb Q;<)$). Thus, when solving an instance $\sI$ of $\CSP(\sA)$, we have to pick for every tuple in a relation of $\sI$ one of finitely many orbits (with respect to $\Omega$) in the corresponding relation of $\sA$ in a consistent way. This is the   general principle for $\omega$-categorical templates $\sA$: orbits of an oligomorphic  subgroup $\Omega$ of $\Aut(\sA)$ take the role of values for tuples in an instance, and in particular the 1-orbits of $\Omega$ take the role of values for its elements.  If the orbits of $\Omega$ have a locally   verifiable local  representation (see e.g.~\cite{Book, Pinsker22} for a precise discussion of such a context, namely when $\Omega=\Aut(\sB)$ for  a \emph{finitely bounded homogeneous structure} $\sB$), then $\CSP(\sA)$ is in NP, and a P/NP-complete complexity dichotomy has been conjectured by Bodirsky and Pinsker more than 10 years ago (see~\cite{BPP-projective-homomorphisms,BartoPinskerDichotomy,  BKOPP,BKOPP-equations} for formulations of the conjecture).

Barto and Pinsker obtained the first general structural dichotomy for $\omega$-categorical  graphs \cite{BartoPinskerDichotomy, Topo}. If an undirected graph $\urgraph$ contains the three-element clique $K_3$ as a subgraph, and $\Omega$ is an oligomorphic subgroup of its automorphism group, then    one of the following holds: either $\urgraph$ together with the $\Omega$-orbits pp-constructs every finite structure, or $\urgraph/\Omega$ has a loop, i.e.~there is an edge within an $\Omega$-orbit (also called a \emph{pseudoloop} of $\urgraph$ with respect to $\Omega$). This implies a structural dichotomy for all $\omega$-categorical templates  $\sA$: either $\sA$ together with the orbits of its automorphism group pp-constructs all finite structures, or $\sA$ has unary polymorphisms $u,v$ and a 6-ary polymorphism $s$ such that $us(x,y,x,z,y,z)=vs(y,x,z,x,z,y)$ holds for all $x,y,z$. The hope, sparked by this result, for  a full lifting of the above-mentioned Hell-Ne\v{s}et\v{r}il structural dichotomy for finite undirected non-bipartite graphs to $\omega$-categorical ones  was, however, only fulfilled in 2023~\cite{symmetries}. This delay can be attributed to the fact that  pp-constructions  from a template $\sA$ allow, together with the $\Omega$-orbits, for the addition of elements of $\sA$ as singleton unary relations; adding all elements of a finite $\sA$  yields a structure with a trivial automorphism group, thereby  granting access to decades of fruitful research on such structures; the same cannot be achieved for infinite structures.       Given the additional technical complications when passing from undirected to  directed graphs  already in the finite, it is hardly surprising  that the lifting of the dichotomy for finite smooth digraphs of algebraic length~1 had to remain open: in fact, hitherto not a single structural consequence of the inability to pp-construct all finite structures has been unveiled  for $\omega$-categorical digraphs, even under reasonable conditions such as the presence of a certain subgraph.

\subsection{Conservativity \& our contributions}

Following  the complexity dichotomies for 2-element~\cite{Schaefer} and 3-element templates~\cite{Bulatov-3-conf}, the next breakthrough on the path to the theorems of Bulatov and Zhuk  was Bulatov's proof of the  dichotomy for finite \emph{conservative} templates~\cite{Conservative}: that is,  templates $\sA$ enriched by all subsets of their domain as unary relations. An instance $\sI$ of $\Csp(\sA)$ might thus impose for each of its elements an arbitrary list of  allowed  values in $\sA$ which the element may be sent to under a solution   (that is, under a homomorphism to $\sA$).  Connecting this with colouring problems, one might view the elements of $\sA$ as colours, and the lists for elements of the instance $\sI$ as sets of allowed colours from which one colour has to be chosen. 
Adding these unary relations makes the template even more rigid than after adding unary singleton relations (discussed above), and thereby facilitates the analysis of its symmetries (polymorphisms). On the other hand, the conservative setting already exposes some of the most important building blocks required for a structural analysis of arbitrary templates; its importance is witnessed by several works in this setting after Bulatov's original proof~\cite{Barto-Conservative,Kazda2015,Bulatov-Conservative-Revisited}. 
 Moreover, variants of CSPs such as counting CSPs~\cite{BDGJM, CHENconservative} and valued CSPs~\cite{Kolmogorovconservative} have been studied in this setting. Conservative graph-colouring problems  are known as \emph{list homomorphism problems}, and have been examined 
for many classes of finite digraphs \cite{Helllisthom,Helllisthom2, Egrilisthom}. 

For $\omega$-categorical templates $\sA$, where the $1$-orbits of a subgroup $\Omega$ of $\Aut(\sA)$ take the role of elements, a conservative template would thus be enriched by a unary relation for all unions of 1-orbits of $\Omega$. 
In this work, assuming only access to unions of pairs of $1$-orbits, we obtain the first structural dichotomy for smooth digraphs of algebraic length~1. For a template $\sA$,  a subgroup $\Omega$ of $\Aut(\sA)$, and $k\geq 1$ let us call the expansion of $\sA$ by all unary relations $P_1\cup\cdots\cup P_k$, where each $P_i$ is a 1-orbit of $\Omega$, the \emph{$k$-conservative expansion of $\sA$ with respect to $\Omega$.}

\begin{restatable}
{theorem}{thmten}\label{thm:10}
    Let $\urgraph =(\urdomain; \edge)$ be a smooth digraph of algebraic length $1$, and let $\Omega\leq \Aut(\urgraph)$ be oligomorphic. Then one of the following holds:\begin{enumerate}
        \item $\urgraph / \Omega$ has a loop;
        \item the 2-conservative expansion of $\urgraph$ with respect to $\Omega$ pp-constructs EVERYTHING, i.e.~every finite structure.        
    \end{enumerate} 
\end{restatable}

As a corollary, we obtain a hardness criterion for conservative CSPs of smooth digraphs of algebraic length~1;  that is, CSPs of such graphs where unions of $1$-orbits 
can appear as constraints. As in the finite case discussed above, the $1$-orbits can be viewed as colours, and the unions of 1-orbits as lists of colours from which one has to choose. 
\begin{corollary}\label{corollary:listhom}
    Let $\urgraph =(\urdomain; \edge)$ be a smooth digraph of algebraic length $1$, and let $\Omega\leq \Aut(\urgraph)$ be oligomorphic.
    Consider the CSP where each instance $\sI$ is a finite digraph enhanced with a list of  allowed $1$-orbits of $\Omega$ for each of its vertices, and the question is whether the graph on $\sI$ can be $\urgraph$-coloured whilst respecting the lists. Then this problem is NP-hard unless $\urgraph/\Omega$ has a loop.
\end{corollary}

\begin{example}
Let $n\geq 3$, and let $\mathcal F$ be a finite set of $n$-coloured tournaments  (i.e.~tournaments which are expanded by unary predicates $P_1,\ldots,P_n$ which form a partition of its domain). Then there exists an $\omega$-categorical $n$-coloured digraph $\sA$ whose induced substructures are precisely those $n$-coloured digraphs which do not contain a copy of  any member of ${\mathcal F}$ and  which do not have any edge within any colour $P_i$, and such that  the 1-orbits of $\Omega:=\Aut(\sA)$ are precisely the colours $P_i$ (this follows from Fra\"{i}ss\'e's 
theorem~\cite{Fraisse53}; we skip the details). Let $\urgraph$ be the graph reduct of $\sA$, i.e.~the graph obtained by forgetting the colours $P_i$. 
The conservative CSP (or the list homomorphism problem)   
of $\urgraph$ with respect to $\Omega$  gets as an input a finite digraph $\sI$ and for every vertex $v$ a subset $\labels(v)$ of the colours of $\sA$, and asks for the existence of a homomorphism from $\sI$ into $\urgraph$ such that every vertex $v$ is sent into one of the colours specified in $L(v)$; in other words, a colouring of $\sI$ avoiding $\mathcal F$ as well as monochromatic edges.  This problem is NP-complete, as follows in particular from  Corollary~\ref{corollary:listhom}. A similar albeit slightly more involved construction exists for any homomorphism-closed set $\mathcal F$ of connected digraphs~\cite{CherlinShelahShi,Hubicka-Nesetril-All-Those}; these templates fall into the framework of the logic MMSNP~\cite{MMSNP-Journal}.
\end{example}

Using standard methods to derive algebraic consequences from loops in graphs, we then moreover see that any $\omega$-categorical template which is \emph{strongly 2-conservative} in the sense that not only unions of pairs of 1-orbits of  $\Omega$, but also of $k$-orbits are part of it, has a \emph{4-ary pseudo-Siggers polymorphism}  unless its CSP is NP-hard. To put this in context, while certain algebraic conditions corresponding to undirected graphs are known to hold for tractable $\omega$-categorical templates, this   condition is wide open in general, and one of the most intriguing open problems in this context. In particular, it is not even known to apply to tractable \emph{temporal CSPs}~\cite{tcsps-journal}, i.e.~CSPs of templates first-order definable in the order of the rational numbers.

\begin{restatable}{corollary}{Siggers}\label{corollary:siggers}
    Let $\sA$ be a structure, and let $\Omega\leq \Aut(\sA)$ be oligomorphic. Let $\sA'$ be the expansion of $\sA$ by all  unions $P\cup Q$, where $P,Q$ are $\Omega$-orbits 
    of $k$-tuples and $k\geq 1$.  
     Then one of the following holds:
    \begin{itemize}
        \item $\sA'$ has a 4-ary pseudo-Siggers polymorphism, i.e.~there are unary polymorphisms $u,v$ and a 4-ary polymorphism $s$ such that $us(a,r,e,a)=vs(r,a,r,e)$ holds for all values $a,r,e$;
        \item $\sA'$ pp-constructs  EVERYTHING,  and hence $\CSP(\sA')$ is NP-hard.
    \end{itemize}
\end{restatable}

We ought to remark here that expanding a template $\sA$ as in \Cref{corollary:siggers} will yield an NP-hard  CSP for many prominent templates previously studied in the literature: in fact, the expansion $\sA'$ cannot have any binary injective polymorphisms, in contrast to these  templates. One could argue, however, that this is not necessarily the typical case: recent results show the existence of templates within the range of the Bodirsky-Pinsker conjecture  that 
are even solvable by local consistency methods, but do not possess any binary injective polymorphism~\cite[Corollary 8.10]{MarimonPinsker24}. Moreover, if the inequality relation $\neq$ is pp-definable in $\sA$ (frequent in prominent examples, see e.g.~\cite[Proposition~19]{MottetPinskerSmooth}), 
then by the same proof one obtains the following variant of \Cref{corollary:siggers}: $\sA$ is only expanded by  $P\cup Q$ for orbits $P,Q$ of~\emph{injective} tuples, and in the first case, the   polymorphisms $u,v,s$ satisfy said equation only on injective tuples. In this variant, the first case applies 
to several tractable cases of  major complexity classifications, e.g.~for Graph-SAT or Hypergraph-SAT problems~\cite{BodPin-Schaefer-both,MottetNagyPinsker24,SmoothapproxJACM,MottetPinskerSmooth}; moreover, the satisfaction of the equation on \emph{all} tuples can then often be easily derived by different methods (such as pre-composition with binary \emph{canonical injections of type projection},  see e.g.~\cite{SmoothapproxJACM}). 
 Altogether, this hints at the possibility that  for $\omega$-categorical structures, the existence of a 4-ary pseudo-Siggers polymorphism might be equivalent to that of the 6-ary pseudo-Siggers polymorphism from~\cite{BartoPinskerDichotomy,Topo};  the latter has been  conjectured to separate the hard from the tractable CSPs within the Bodirsky-Pinsker conjecture.

\subsection{Outlook: a more liberal future?}\label{sect:liberalfuture}

If one were able to replace 2-conservativity by 1-conservativity 
in \Cref{thm:10}, then one would obtain a significant strengthening of the theorem of Barto and Pinsker for undirected graphs  from~\cite{BartoPinskerDichotomy,Topo}. This would imply the amazingly general statement that any $\omega$-categorical template which does not pp-construct all finite structures has polymorphisms $u,v,s$ as in \Cref{corollary:siggers}. Moreover, even the converse would be true for \emph{model-complete cores} (which we will not need here; any $\omega$-categorical template has the same CSP as a model-complete core~\cite{Cores-journal}). 
 Evidence that our proof lays the foundations for such  endeavour is provided, for example, by  the following proposition, which using concepts from the proof of our main theorem 
  solves the case where $\urgraph/\Omega$ has only three orbits and is isomorphic to the graph corresponding to the equations of \Cref{corollary:siggers} (see its proof). Already for this particular situation, no  proof was hitherto  known, and its undirected version was dedicated an own proof in~\cite[Proposition~5.3]{BartoPinskerDichotomy}.

\begin{restatable}{prop}{Siggerscrap}\label{prop:Siggerscrap}
Let $\urgraph=(\urdomain;\edge)$ be a smooth digraph and let $\Omega\leq \Aut(\urgraph)$ be oligomorphic, and such that $\urdomain / \Omega$ is isomorphic to the following graph.
    
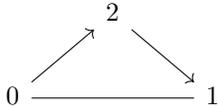
\begin{center}
        
    \begin{tikzcd}[ampersand replacement=\&]
                                 \& 2 \arrow[rd] \&   \\
0 \arrow[ru] \arrow[rr, no head] \&              \& 1
\end{tikzcd}.\end{center}
Then the 1-conservative expansion of  $\urgraph$ with respect to $\Omega$ pp-constructs EVERYTHING.
\end{restatable}
Another, independent,  direction is to replace the assumption of algebraic length~1 by its weaker version modulo $\Omega$, i.e., require only  $\urgraph/\Omega$ to have  algebraic length~1. Resulting statements have been shown in~\cite{symmetries} for finite graphs, and turn out useful when comparing the strength of  algebraic conditions (see~\cite{olsak-loop,Pseudo-loop}). 

\subsection{A couple of thoughts on how to approach the paper}

This paper is organised as follows: We introduce definitions and notation in~\Cref{sect:prelims}. In \Cref{sect:poly}, we discuss the results on $4$-ary pseudo-Siggers polymorphisms derived from~\Cref{thm:10}. \Cref{sect:newideas} motivates the new methods developed for the proof of our main theorem.  We formalise these methods in~\Cref{sect:finitising,sect:defalpha,sect:reductionistic}. Finally, in \Cref{sect:summaryofproof} we  derive \Cref{thm:10} from two auxiliary propositions that will be proved in the subsections therebefore.
A sketch of the key ideas behind the main theorem is provided in \Cref{sect:structureofproof},  once the necessary terminology has been established to ensure comprehension of these concepts. Although the structure of the proof of \Cref{thm:10} is rather complex, the reader might find comfort in \Cref{fig:overview}, designed to provide guidance. We hope the figure, as well as this introductory text, might  appeal to the reader's reason; otherwise, perhaps the following, equivalent  lines might appeal to their heart:

\begin{center} 
    \emph{A graph  set forth, one day, to find}
    
    \emph{A finite thing of any kind}
    
    \emph{Directed 'twas,   yet it was  smooth}
    
    \emph{And at length one still in its youth}
    
    \emph{With pairs of orbits at its hand}
    
    \emph{It strived to build things on its land}
    
    \emph{But when it failed to build $K_3$}
    
    \emph{A pseudoloop it had to see.}
\end{center}

\section{Preliminaries}\label{sect:prelims}

\subsection{Relational structures}

For $n\geq 1$, we write $[n]$ for the set $\{1,\dots,n\}$. A \emph{relation} on a set $G$ is a subset  $ R \subseteq G^n$ for some $n \geq 1$. We refer to the number $n$ as the \emph{arity} of the relation $R$. A \emph{relational structure} (or simply \emph{structure}) is a tuple $\urgraph = (G; \mathcal R)$ consisting of a set $G$ and a family $\mathcal R$ of relations on $G$. A structure $\subgraph=(\urdomain; \mathcal S)$ is an \emph{expansion} of a  
structure $\urgraph=(\urdomain;\mathcal R)$ if $\mathcal S\supseteq \mathcal R$. 
For $\subdomain\subseteq\urdomain$, we write $\urgraph|_\subdomain=(\subdomain;\mathcal R)$ for the \emph{induced substructure} of $\urgraph$ on $\subdomain$. 
 We will understand an $n$-ary function $f$ on a set $\urdomain$ also as a function on $k$-tuples of elements from $\urdomain$ for any $k\geq 1$, where we apply $f$ componentwise.

Let $\urgraph=(\urdomain;\mathcal R)$ and $\subgraph=(\subdomain;\mathcal R')$ be relational structures such that the families $\mathcal R$ and $\mathcal R'$ use the same relational symbols. A mapping $f\colon \urdomain \rightarrow \subdomain$ is a \emph{homomorphism} from $\urgraph$ to $\subgraph$ if for every tuple $\tuple t$ contained in some  relation  of $\urgraph$ it holds that $f(\tuple t)$ belongs to the corresponding relation of $\subgraph$.
 Two structures $\urgraph$ and $\subgraph$ are \emph{homomorphically equivalent} if there exists a homomorphism both from $\urgraph$ to $\subgraph$, and from $\subgraph$ to $\urgraph$. 
A \emph{polymorphism} of a relational structure $\urgraph=(\urdomain;\mathcal R)$ is a function $f\colon \urdomain^n\rightarrow \urdomain$ for some $n\geq 1$ such that for every $R\in\mathcal R$ of some arity $k\geq 1$, and for all tuples $\tuple t_1,\dots,\tuple t_n\in R$, the tuple $f(\tuple t_1,\dots,\tuple t_n)$ is contained in $R$; we then say that $f$ \emph{preserves} $R$, or that $R$ is \emph{invariant} under $f$. An \emph{automorphism}  is a bijective unary polymorphism $f$ so that $f^{-1}$ is a polymorphism as well. 
We write $\Pol(\urgraph)$ and $\Aut(\urgraph)$ for the sets of all polymorphisms and automorphisms of $\urgraph$, respectively. 

\subsection{Formulae}

A first-order formula is a \emph{primitive positive formula} (or simply \emph{pp-formula}) over a family $\mathcal R$ of relational symbols if it is built using only predicates from $\mathcal R$, conjunction, and existential quantification. We abuse notation by using the same name for a relational symbol and its interpretation in a relational structure. A relation is \emph{pp-definable} in a relational structure $\urgraph=(G; \mathcal R)$  if it is first-order definable  by a pp-formula over $\mathcal R$. A structure is pp-definable in $\urgraph$ if all its relations are. We say that $\urgraph$ pp-defines $\subgraph$ \emph{with parameters} if $\subgraph$ is pp-definable in the expansion of $\urgraph$ by the singleton unary relations $\{a\}$ for every element $a$ of its base set. 
 
 The structure $\subgraph$ is a \emph{pp-power} of $\urgraph$ if there exists $n\geq 1$ such that the domain of $\subgraph$ is $\urdomain^n$, and its relations are pp-definable from $\urgraph$ (where a $k$-ary relation on $\urdomain^n$ is understood as a $kn$-ary relation on $\urdomain$).
We say that $\urgraph$ \emph{pp-constructs} $\subgraph$ if $\subgraph$ is homomorphically equivalent to a pp-power of $\urgraph$; $\urgraph$ pp-constructs EVERYTHING if it pp-constructs every finite structure.

A pp-formula $\phi$ is a \emph{tree}  if the following bipartite undirected graph is a tree, i.e.~connected without cycles: the vertices are the variables of $\phi$ (on one side) as well as the conjuncts of $\phi$ (on the other side), and there is an edge between a variable and a conjunct if and only if the variable appears in that conjunct. A relation is \emph{tree pp-definable} in a structure $\urgraph$ if it is pp-definable in $\urgraph$ by a pp-formula which is a tree.

For a relation $R\subseteq \urdomain^n$, the relation $\OR(R,R)$ is the relation defined by the formula $R(\tuple x)\vee R(\tuple y)$ which contains precisely all tuples $(\tuple s,\tuple t)\in \urdomain^{2n}$ for which either $\tuple s\in R$, or $\tuple t\in R$.

\subsection{Permutation groups}
For a permutation group $\Omega$ acting on a set $G$ and $a\in G$, we write $\Omega_a:=\{f\in\Omega \mid f(a)=a\}$ for the stabiliser of $a$ in $G$. For a subset $H$ of $G$, we set $\Omega|_H:=\{f|_H  \mid f \in \Omega, \  f(H)= H\}$ to be the \emph{restriction} of $\Omega$ to $H$. 
For any $k\geq 1$, we will call the orbits of $k$-tuples under the natural (componentwise) action of $\Omega$ on this set \emph{$k$-orbits} of $\Omega$. By an \emph{$\Omega$-orbit} we refer to a $k$-orbit for some $k$. We say that $\Omega$ is \emph{oligomorphic} if it has only finitely many $k$-orbits for every $k\geq 1$.
A structure $\urgraph$ is \emph{$\omega$-categorical} if $\Aut(\urgraph)$ is oligomorphic.
Let $\urgraph$ be a relational structure whose underlying set is $\urdomain$, and let $\Omega$ be a permutation group acting on $\urdomain$. 
For $k\geq 1$, the \emph{$k$-conservative expansion of $\urgraph$ with respect to $\Omega$} is the expansion of $\urgraph$ by all unary relations of the form $O_1\cup\dots\cup O_k$, where $O_i$ are $1$-orbits of $\Omega$. 
We will denote the $\Omega$-orbit equivalence on $\urdomain$
by $\omega$, i.e. the $1$-orbits of $\Omega$ are precisely the $\omega$-classes.   We say that a relation $R$ is \emph{invariant} under $\Omega$ (or simply \emph{$\Omega$-invariant}) if it is invariant under all functions in $\Omega$.
It is easy to see that  a relation first-order definable from a set of $\Omega$-invariant relations is $\Omega$-invariant itself. 
A structure  is a \emph{model-complete core} if it pp-defines all orbits of $k$-tuples under its own automorphism group for every $k\geq 1$ (this definition, though different  from the one commonly found in the literature, is equivalent to it; see~\cite[Theorem 4.5.1]{Book}).  For every $\omega$-categorical structure $\urgraph$, there exists a structure $\urgraph'$ which is a model-complete core and which is homomorphically equivalent to $\urgraph$. Moreover, $\urgraph'$ is again $\omega$-categorical and determined uniquely up to isomorphisms~\cite{Cores-journal}.

If $R$ is a  relation on a set $G$, and $\eta$ an equivalence relation on $G$, then $R$ induces a relation on the factor set $G/\eta$ which we shall again denote by $R$ to avoid cumbersome notation:  $R(A_1,\ldots,A_n)$ holds for $\eta$-classes $A_1,\ldots,A_n$ if there exist $a_i\in A_i$ such that $R(a_1,\ldots,a_n)$ holds for the original $R$. For a relational structure $\urgraph$, the quotient structure modulo an equivalence $\eta$ on its domain is denoted $\urgraph/\eta$. 
We say that a relation $R$ is \emph{$\eta$-stable} if $R(a_1,\ldots,a_n)$ implies $R(b_1,\ldots,b_n)$ whenever $\eta(a_i,b_i)$ holds for all $i\in \br n$. Note that a unary relation is $\eta$-stable if and only if it is a union of $\eta$-classes. The \emph{$\eta$-blow-up} of a relation $R$ on $A/\eta$, denoted by $R^{\eta}$, is the $\eta$-stable relation on $A$ which factors to $R$. 
 If $\Omega$ is a permutation group acting on the domain of $R$, then we write $R/\Omega$ for the factor of $R$ by $1$-orbits of $\Omega$; we use similar notation for structures.

\subsection{Digraphs and paths}\label{sec:relations}

A \emph{directed graph} (or simply \emph{digraph}) is a relational structure $\urgraph=(\urdomain;E)$ with a single binary relation $E$  (which will usually be  denoted by $\edge$). A digraph  is \emph{smooth} if the relation $E$ is \emph{subdirect}, i.e. if the projection to any one of its coordinates equals $G$. The \emph{smooth part} of a digraph $\urgraph$ is the largest subset $\subdomain\subseteq \urdomain$ such that $\urgraph|_\subdomain$ is smooth.
If $E$ is a binary relation on a set $G$, we write $E^{-1}$ for the set $\{(y,x) \mid (x,y)\in E\}$. If $E=\edge$, we prefer the notation $\ledge$ for $E^{-1}$. A \emph{loop} of $E$ is a pair $(a,a)\in E$. 
An \emph{(abstract) $E$-path} is a finite sequence $p=(E_1,\ldots,E_n)$ where each $E_i$ is either $E$ or $E^{-1}$. The \emph{algebraic length} of $p$ is the number of occurrences of $E$ minus the number of occurrences of $E^{-1}$; to make it crystal clear, $\sum_{i}\log_E(E_i)$.  A \emph{realisation} of $p$ is a sequence $(a_1,\ldots,a_{n+1})$ of elements of $G$ such that 
$E_i(a_{i},a_{i+1})$
 holds for all $i$. We also call $(a_1,\ldots,a_{n+1})$ an  \emph{$E$-walk} (along $p$). A path which has a realisation is \emph{realisable}. Note that if $E$ is the edge relation of a smooth digraph, then any $E$-path is realisable. If $E_i=E$ for all $i$, then we call $p$ an \emph{$E$-forward path}, and any realisation of it an \emph{$E$-forward walk}; an \emph{$E$-backward path  (or walk)} is defined accordingly. A \emph{cycle} is an $E$-forward walk $(a_1,\dots,a_{n+1})$ with $a_1=a_{n+1}$.  
A digraph $\urgraph=(\urdomain;E)$ is \emph{weakly connected} if for all $a,b\in\urdomain$ with $a\neq b$, there exists an $E$-walk from $a$ to $b$. 
 A \emph{weakly connected component} of $\urgraph$ is any subset $\subdomain\subseteq\urdomain$ which is maximal with respect to inclusion such that the induced digraph $\urgraph|_{\subdomain}$ is weakly connected. 
 The digraph $\urgraph$ has \emph{algebraic length $1$} if there exists a  
path of algebraic length $1$ that has a realisation of which the first and last element coincide.

Let $\Omega$ be a permutation group on the  domain of the digraph $\urgraph=(\urdomain;E)$. An \emph{$\Omega$-orbit-labelled $E$-}path is a pair $\pi=(p,(P_1,\ldots,P_{n+1}))$ where $p$ is an $E$-path of length $n$, and each $P_i$ is an $\omega$-class; we call $\pi$ an $\Omega$-labelling of $p$, and say that $\pi$ is \emph{from $P_1$ to $P_{n+1}$}. Usually, we shall only  write \emph{$\Omega$-orbit-labelled path when  $E$ is clear from the context.} 
A \emph{realisation} of $\pi$ is a realisation $(a_1,\ldots,a_{n+1})$ of $p$ with the additional property that $a_i\in P_i$ for all $i$; if a realisation exists, then we say that $\pi$ is \emph{realisable}. If $p,q$ are $E$-paths, then we write $p+q$ for the path obtained by concatenation of the two. If $p=(E_1,\ldots,E_n)$, then we set $-p:=(E_n^{-1},\ldots,E_1^{-1})$. 
Finally,  $p-q:=p+(-q)$. We say that $p$ is \emph{symmetric} if it is of the form $q-q$. We use similar notation for $\Omega$-orbit-labelled paths; note, however, that the sum of $\Omega$-orbit-labelled paths $\pi,\rho$ is only well-defined if the last orbit of $\pi$ is equal to the first orbit of $\rho$. An \emph{extension} of a path $p=(E_1,\ldots,E_n)$ is obtained by iteratively inserting either $(E,E^{-1})$ or $(E^{-1},E)$ at arbitrary places of the sequence; an extension of an $\Omega$-orbit-labelled path is obtained via the same process by replacing an orbit label $P_i$ by $P_i,P',P_i$ for some $\omega$-class $P'$ at the appropriate place. A \emph{relabelling} of $\pi=(p,(P_1,\ldots,P_{n+1}))$ is any $\Omega$-orbit-labelled path $\pi'=(p,(P_1',\ldots,P_{n+1}'))$.
  The \emph{merge} of $\pi$ and $\pi'$ then is the pair $\pi\merge \pi':=(p,(P_1\cup P_1',\ldots,P_{n+1}\cup P_{n+1}'))$, with the obvious semantics of a realisation.
To make notation more readable,  
for such a merge $\pi\merge \pi'$ we will sometimes write   $$\binom{P_1}{P'_1} E_1 \binom{P_2}{P'_2} E_2\; \cdots\; E_{n} \binom{P_{n+1}}{P'_{n+1}}. $$ The \emph{relation induced} by a path $p$ is the binary relation $\gamma_p \subseteq \urdomain^2$ containing all pairs  $(a,b)$ for which there exists a realisation of $p$ starting at $a$ and ending at $b$. Similarly, we define the relation  $\gamma_{\pi}$ induced by an $\Omega$-orbit-labelled path $\pi$,  and the relation  $\gamma_{\pi\vee \pi'}$ induced by a merge. Note that $\gamma_p$ is pp-definable from $E$, and  that $\gamma_\pi$ is pp-definable from $E$ together with  the $\omega$-classes; in order to pp-define $\gamma_{\pi\merge\pi'}$, we need $E$ and unions of pairs of $\omega$-classes.

If $R_1,\dots,R_k$ are binary relations, we write $R_1+\cdots + R_k$ for the relation containing all pairs $(a_0,a_k)$ for which there exist $a_1,\dots,a_{k-1}$ with $(a_{i-1},a_i)\in R_i$ for every $i\in[k]$. 
If $R_1=\cdots=R_k=E$, we denote this sum by $E^k$.
If $E=\edge$, we also write $\edgek$ for  $E^k$. 
We set $E^{-k}:=(E^k)^{-1}$. For a unary relation $S$, we write $S+E$ for the relation $\{b \mid \exists a\in S\; (a,b)\in E\}$; note that this is a pp-definition.
 For better readability, we set $a+E:=\{a\}+E$. 
For $k,n\geq 1$, the $(k,n)$-$E$-fence is the $E$-path $q+\cdots+q$, with $n$ occurrences of $q$, where $q=p-p$ for $p=(E,\ldots,E)$,  with $k$ occurrences of $E$ (so the path has length $2kn$).
 If $E$ is clear from the context, we denote the relation induced by the $(k,n)$-$E$-fence by $\fence{k}{n}$. Clearly, $\fence{k}{n}\subseteq \fence{k}{n+1}$. 
 The \emph{k-linkedness relation} is  $\bigcup_n\fence{k}{n}$; it is an equivalence relation, and if  $\urgraph=(\urdomain;E)$ is finite, then it is equal to some $\fence{k}{n}$ and hence tree pp-definable in  $\urgraph$~\cite{cyclicterms}. 
  We say that $E$ is \emph{$k$-linked} if this equivalence relation is full. 

\section{Lifting the Finite to the Infinite}\label{sect:lifting}

\subsection{Polymorphisms}\label{sect:poly}

While some algebraic invariants of finite relational structures have been shown to have \emph{pseudo-}counterparts in the context of $\omega$-categoricity (for example, the existence of a \emph{$6$-ary Siggers polymorphism} as in~\cite{Siggers} corresponding to the existence of a \emph{$6$-ary pseudo-Siggers polymorphism} as presented in~\cite{BartoPinskerDichotomy, Topo}), and for others it is known that these counterparts do not apply (for example, there exist $\omega$-categorical structures that do not pp-construct EVERYTHING, and
 yet do not admit \emph{pseudo-WNU polymorphisms} \cite{symmetries},  contrasting the existence of \emph{WNU polymorphisms} in finite structures as in~\cite{MarotiMcKenzie}), the question of whether or not algebraic invariants can be lifted remains unresolved for many cases. As a consequence of~\Cref{thm:10}, we are now able to derive a first pseudo-version of~\cite[Theorem 2.2]{JMMM} asserting the existence of a \emph{$4$-ary Siggers polymorphism } for any finite structure that fails to pp-construct EVERYTHING.

\Siggers*

\begin{proof}[Proof of~\Cref{corollary:siggers} from~\Cref{thm:10}]
  Let $k \geq 1$, and let  $\tuple a,\tuple e,\tuple r$ be arbitrary $k$-tuples of elements of $A$.
Consider the digraph $\subgraph$ on these three elements which has edges from $\tuple a$ to $\tuple r$, from $\tuple r$ to $\tuple a$, from $\tuple a$ to $\tuple e$, and from $\tuple e$ to $\tuple r$. Let $\urgraph=(\urdomain;\edge)$ be the smallest digraph containing this digraph that is invariant under the action of $\Pol(\sA')$ on $A^k$. More precisely, the domain $\urdomain$ of $\urgraph$ consists of all tuples in $A^k$ of the form $g(\tuple b_1,\dots,\tuple b_n)$, where $g\in\Pol(\sA')$, $n$ is its arity, and  $\tuple b_i\in \{\tuple a,\tuple e,\tuple r\}$ for every $i \in \br n$; the  relation $\edge$ is given by  $g(\tuple b_1,\dots,\tuple b_n)\edge g(\tuple c_1,\dots,\tuple c_n)$ if there is an edge from $\tuple b_i$ to $\tuple c_i$ in $\subgraph$ for every $i\in[n]$. As by construction  $\edge$ is preserved by all polymorphisms,  it is pp-definable in $\sA'$~\cite{BodirskyNesetrilJLC}. Hence, $\urgraph$ is a pp-power of $\sA'$. 
Moreover, the action of $\Omega$ on $k$-tuples is an oligomorphic subgroup $\Omega$ of $\Aut(\urgraph)$ (which contains at least the action of $\Aut(\sA')$ on $k$-tuples). We may thus apply~\Cref{thm:10}.
    
If for any choice of $k,\tuple a,\tuple e,$ and $\tuple r$, the second case of the theorem occurs, then  we get that $\urgraph$ together with the relations $P\cup Q$, where $P,Q$ are orbits of $k$-tuples of the action of $\Omega$ on $k$-tuples, pp-constructs EVERYTHING; since $\urgraph$ is itself pp-constructible from $\sA'$, we are done.
    Hence, assume that  $\urgraph / \Omega$ contains a loop for all $k,\tuple a,\tuple e,\tuple r$. This means that there exist $n\geq 1$, an $n$-ary polymorphism $s' \in \Pol(\sA')$, for every $i \in \br n$ tuples $\tuple b_i,\tuple c_i\in \{\tuple a,\tuple e,\tuple r\}$ with an edge from $\tuple b_i$ to $\tuple c_i$ in $\subgraph$, and $f \in\Omega$ such that $fs'(\tuple b_1,\dots,\tuple b_n)=s'(\tuple c_1,\dots,\tuple c_n)$. Let us set $s(x_1,x_2,x_3,x_4):=s'(y_1,\dots,y_n)$, where $y_i=x_1$ if $(\tuple b_i,\tuple c_i)=(\tuple a,\tuple r)$, $y_i=x_2$ if $(\tuple b_i,\tuple c_i)=(\tuple r,\tuple a)$, $y_i=x_3$ if $(\tuple b_i,\tuple c_i)=(\tuple e,\tuple r)$, and $y_i=x_4$ if $(\tuple b_i,\tuple c_i)=(\tuple a,\tuple e)$. Then $s\in \Pol(\sA')$ satisfies $fs(\tuple a,\tuple r,\tuple e,\tuple a)=s(\tuple r,\tuple a,\tuple r,\tuple e)$,  i.e.~$s$ witnesses the 4-ary pseudo-Siggers identity on the tuples $\tuple a,\tuple e,\tuple r$. Since we obtain such ``local'' witness for arbitrary tuples, a standard compactness argument using the oligomorphicity of $\Omega$  yields the first case of this corollary. Such compactness argument is given, for example, in the  proof of \cite[Lemma 4.2]{Topo} (for the 6-ary pseudo-Siggers identity); the appearance of both $u$ and $v$, which might not be permutations, instead of $f$ happens in the argument and cannot be avoided~\cite{BP-canonical}.
\end{proof}

Our hope that the existence of a $4$-ary pseudo-Siggers polymorphism may in fact characterise \emph{all} $\omega$-categorical structures that do not pp-construct EVERYTHING is sparked, additionally, by~\Cref{prop:Siggerscrap}:  
in assuming that a digraph looks, modulo orbit-equivalence, exactly like the digraph corresponding (as the graph $\subgraph$ in the proof of \Cref{corollary:siggers} above) to  the 4-ary Siggers identity, our result  can be viewed as a base case for such a statement. In other words, the knowledge of the particular shape of the digraph $\urgraph /\Omega$ allows us to drop the $2$-conservativity assumption from \Cref{thm:10} and replace it by $1$-conservativity. In order to obtain the improvement of~\Cref{corollary:siggers} where $\sA$ is not expanded by unions of $\Omega$-orbits,  
it is sufficient to prove the variant  of~\Cref{thm:10} where  $\urgraph$ satisfies the stronger assumption of containing the 3-element graph in~\Cref{prop:Siggerscrap} as a (not necessarily induced) substructure, but assuming only $1$-conservativity instead of $2$-conservativity.
In the proof of~\Cref{prop:Siggerscrap}, we build upon the new method of \emph{finitising digraphs}, as employed in the proof of \Cref{thm:10} (see the outline below and \Cref{sect:finitising}).

\subsection{Revisiting established methods}\label{sect:newideas}

Our endeavours of lifting results from finite to infinite structures align with an active field of research that has thrived through the application of various valuable techniques. Building upon recent results in this direction, we make use of some of the methods presented in~\cite{BartoKozikNiven, symmetries} for finite digraphs (or digraphs with finite weakly connected components), 
and apply them to the realm of $\omega$-categorical digraphs. In particular, in contrast to the setting of the most general result in this direction hitherto, \cite[Theorem 10]{symmetries}, our new approach must address weakly connected components that are not finite.

When working with an $\omega$-categorical structure, it is natural to apply  methods designed for finite structures to  a quotient of it that is either finite 
or at least sufficiently tame (e.g., has finite weakly connected components), and then infer information on the original structure.
However, the choice of such a quotient must be made judiciously to ensure it retains sufficient structural information. While taking the quotient  modulo $\omega$, i.e.~the $1$-orbit equivalence with respect to $\Omega$,  seems to be the natural choice, we might lose, for example, all information about the connectivity of a graph.
To this end, given a digraph $\urgraph=(\urdomain;\edge)$ as in the assumptions of~\Cref{thm:10}, and an oligomorphic subgroup $\Omega$ of its automorphism group, we strive to find a suitable refinement $\alpha$ of $\omega$, which shall in particular have the same weakly connected components as $\urgraph$. Clearly, we want such a refinement to be compatible with $\Omega$, whose orbits are our basic building blocks; that is, 
\begin{enumerate}[label = A1)]
    \item $\alpha$ is $\Omega$-invariant.
\end{enumerate}
This implies, in particular, that $\Omega$ acts on $\alpha$-classes, and that it makes sense to speak of a realisation of an $\Omega$-orbit-labelled path $\pi$ in $\urgraph/\alpha$. Moreover, we would like such realisations to be unique in the following sense:
\begin{enumerate}[label = A2)]
    \item for every realisable $\Omega$-orbit-labelled path $\pi$ whose first label is the $\omega$-class $O$, and for every $A\in O / \alpha$, there exists a unique realisation of $\pi$ in $\urgraph / \alpha$ starting in $A$.
\end{enumerate}
Finally, in order to make use of the techniques developed for digraphs with finite weakly connected components, we require:
\begin{enumerate}[label=A3)]
    \item the $ \alpha$-quotient of every weakly connected component of $\urgraph$ is finite.
\end{enumerate} 

For the precise formulation of  properties A1) -- A3), see \Cref{defi:finitises}. \Cref{fig:windings} provides an example for the interaction between $\omega$ and $\alpha$ for $\Omega=\Aut(\urgraph)$. Here, every $\omega$-class $O \in \{0,1,2\}$ consists of two $\alpha$-classes $A_O, B_O$. The left-hand-side of the figure depicts the digraph $\urgraph/\omega$, the digraph on the right-hand-side $\urgraph/\alpha$.

We remark that our construction shows a connection with the seemingly unrelated area of $2$-Prover-$1$-Round Games: If $\urgraph$ is itself weakly connected, then for any refinement $\alpha$ of $\omega$ satisfying items A2) and A3), $\urgraph / \alpha$ can be seen as an instance of a \emph{unique game} for the digraph $\urgraph / \omega$ (for more details on the definition and broader context, see e.g.~\cite{KhotUGC}). Indeed, to every vertex $O$ of $\urgraph / \Omega$, we associate the set $\labels(O)$ of labels consisting of all $\alpha$-classes contained in $O / \alpha$. Note that the cardinality of this set is finite by A3) and by the weak connectivity of $\urgraph$. 
Moreover, to any edge $O\edgeo P$ of $\urgraph / \omega$, we can associate the bijection between $\labels(O)$ and $\labels(P)$ given by item A2) applied to the path $O\edgeo P$. It is easy to see that this unique game has a satisfying labelling if and only if $\omega=\alpha$.
\begin{figure}[t]
\centering
\begin{tikzcd}[column sep={1cm,between origins}, row sep={1.732050808cm,between origins}]
                                 &              &   &                & A_1 \arrow[rr, no head]            &  & B_0 \arrow[rd] &                \\
                                 & 2 \arrow[rd] &   & A_2 \arrow[ru] &                                    &  &                & B_2 \arrow[ld] \\
0 \arrow[rr, no head] \arrow[ru] &              & 1 &                & A_0 \arrow[lu] \arrow[rr, no head] &  & B_1              &               
\end{tikzcd}
    \caption{$\omega$ vs. $\alpha$}\label{fig:windings}
\end{figure}
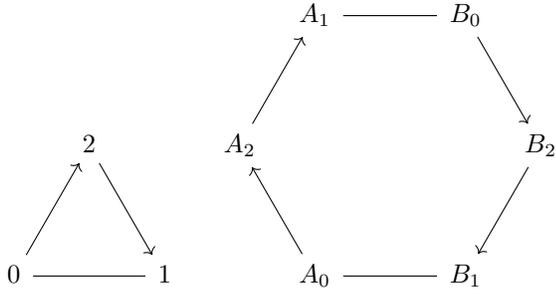
We will construct a suitable refinement $\alpha$ of $\omega$ satisfying, in particular, properties A1) -- A3) outlined above in~\Cref{sect:finitising}. However, further challenges will arise in the attempt to lift results from  the factor graph $\urgraph/\alpha$.
This is due to the fact that $\alpha$ will, in general, not be pp-definable on the whole domain of $\urgraph$. In particular, pp-constructibility from $\urgraph / \alpha$ does not translate to pp-constructibility from $\urgraph$. One can approach this obstacle by using only the restricted form of tree pp-definitions, which are  liftable to pp-definitions in $\urgraph$ in the following sense: if a tree pp-formula defines a relation $R$ in $\urgraph / \alpha$, then it defines a relation $S$ in $\urgraph$ such that its $\alpha$-factor coincides with $R$ (\Cref{lem:liftingTreeDefs}). 
However, there are many relations with the latter property, some of which might have less expressive power than $R$. 
This issue is what motivates the concepts of \emph{$\alpha$-stability} and \emph{$\alpha$-blow-ups}: 
If $S$ is an $\alpha$-stable relation in $\urgraph$, then it is the $\alpha$-blow-up of its $\alpha$-factor, and homomorphically equivalent to it.

Assuming that $\urgraph$ is weakly connected, whence in particular, $\urdomain$ contains only finitely many $\alpha$ classes, we will proceed
by induction on the number of $\alpha$-classes and
successively shrink the domain $\urdomain$ of $\urgraph$ to 
$\alpha$-stable subsets 
$\subdomain\subsetneq\urdomain$. The set $\subdomain$ will always be a class of an $\Omega$-invariant equivalence relation (a \emph{reductionistic set}, see~\Cref{sect:reductionistic}), which will guarantee the following:
\begin{itemize}
    \item[R1)] the orbits of the permutation group $\Omega|_H$ obtained by restricting $\Omega$ to $H$ are the restrictions of the $\Omega$-orbits; and consequently
    \item[R2)] the restriction of $\alpha$ to $H$ retains properties A1)-A3) with respect to $\Omega|_H$.
\end{itemize}

Looking at~\Cref{fig:windings}, it is easy to see that, for instance, the set $A_0 \cup A_1 \cup B_2$ is $\Omega$-reductionistic, whereas $A_0\cup A_1 \cup B_2 \cup B_0$ is not.

\subsection{Structure of the proof of \texorpdfstring{\Cref{thm:10}}{Theorem~\ref{thm:10}}}\label{sect:structureofproof}

We now present a  sketch of the core ideas involved in proving the main result, and of the organisation of the sections pertaining to this proof. While in~\Cref{sect:finitising}, we construct a refinement $\alpha$ of $\omega$ subject to the properties outlined above,  \Cref{sect:thm10} is devoted to proving~\Cref{thm:10}.  The proof is guided by the case distinction presented in~\Cref{fig:overview}, which may help maintaining clarity throughout. Assuming that $\urgraph / \Omega$ does not have a loop, we show that the second item of \Cref{thm:10} applies. Our source of hardness, i.e. a relation which pp-constructs EVERYTHING, will be the $\alpha$-blow-up of a relation $\OR(\sigma,\sigma)$ for a proper equivalence $\sigma$ on a finite subset of $\urdomain / \alpha$. We show that if such a relation $\OR(\sigma,\sigma)$ does not exist, we can inductively reduce to suitable  subgraphs $\subgraph$ of $\urgraph$, and that this reduction process can only be performed finitely many times.

In the course of the proof, we restrict ourselves to a digraph $\urgraph$ such that its $\alpha$-factor is $k$-linked for some $k\geq 1$, and we take  $k$ to be the smallest number with this property. The proof now splits into two cases depending on whether or not the $k$-ary composition $\edgek$ of the edge relation $\edge$ with itself 
gives  a full relation on $\urdomain / \alpha$.  If it does, 
we can replace $\urgraph$ by a proper subgraph still satisfying the prerequisites (\Cref{prop:claimBK12}).
In the other case (\Cref{thm:7}), as outlined before, we seek to construct  $\alpha$-stable relations (\Cref{prop:Marcinsmagic,prop:trickT2}) that allow us to invoke adjusted arguments of~\cite[Theorem 7]{symmetries}. Compared to these arguments, the main difference is that we additionally need to ensure that the relations we pp-define from $\edge$ (which is itself not $\alpha$-stable) together with the previously constructed $\alpha$-stable relations 
remain $\alpha$-stable. Finally, we may again reduce to a suitable proper subgraph, or end up with an $\alpha$-blow-up of a relation $\OR(\sigma, \sigma)$. 

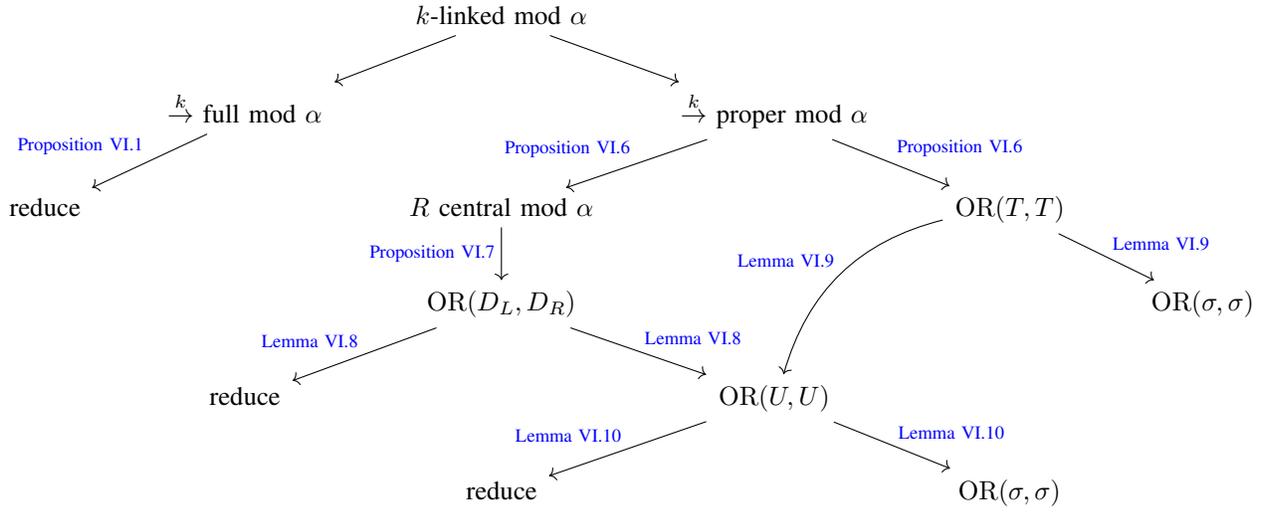
\begin{figure*}[t]
    \centering

\begin{tikzcd}
                    &                                                                                            & \textnormal{$k$-linked mod $\alpha$} \arrow[rd, "\textnormal{ \Cref{thm:7}}"] \arrow[ld]                                   &                                                                                                                                                          &                                                                                                                      &                       \\
                    & \textnormal{$\edgek$ full mod $\alpha$} \arrow[ld, "\textnormal{ \Cref{prop:claimBK12}}"'] &                                                                                              & \textnormal{$\edgek$ proper mod $\alpha$} \arrow[ld, "\textnormal{\Cref{prop:Marcinsmagic}}"'] \arrow[rd, "{\textnormal{\Cref{prop:Marcinsmagic}}}"] &                                                                                                                      &                       \\
\textnormal{reduce} &                                                                                            & \textnormal{$R$ central mod $\alpha$} \arrow[d, "\textnormal{\Cref{prop:trickT2}}"']              &                                                                                                                                                          & {\OR(T,T)} \arrow[rd, "\textnormal{\Cref{lemma:21+22}}"] \arrow[ldd, "\textnormal{\Cref{lemma:21+22}}"', bend right] &                       \\
                    &                                                                                            & {\OR(D_L, D_R)} \arrow[ld, "\textnormal{\Cref{lemma:20}}"'] \arrow[rd, "\textnormal{\Cref{lemma:20}}"] &                                                                                                                                                          &                                                                                                                      & {\OR(\sigma, \sigma)} \\
                    & \textnormal{reduce}                                                                        &                                                                                              & {\OR(U,U)} \arrow[ld, "\textnormal{\Cref{lemma:23}}"'] \arrow[rd, "\textnormal{\Cref{lemma:23}}"]                                                                  &                                                                                                                      &                       \\
                    &                                                                                            & \textnormal{reduce}                                                                          &                                                                                                                                                          & {\OR(\sigma, \sigma)}                                                                                                &                      
\end{tikzcd}
\caption{Overview of the proof}\label{fig:overview}
   
\end{figure*}

\section{Finitising Digraphs}

\subsection{Constructing a finitising equivalence}\label{sect:finitising} 

To construct a refinement $\alpha\subseteq \omega$ adheres to the properties outlined in the previous section, we proceed $\omega$-class by $\omega$-class in building the desired refinement $\alpha$. To this end, fix an 
$\omega$-class~$O$ and consider the relations $\gamma_\pi$ induced by realisable $\Omega$-orbit-labelled paths $\pi$ which are symmetric and start (and hence also end) at $O$.
Clearly, any such relation is a reflexive and symmetric binary relation on $O$. Moreover, it is pp-definable from $\urgraph$ and $\omega$-classes. Using oligomorphicity of~$\Omega$, we will now show that there exists $\pi$ such that $\gamma_\pi$ is maximal among all possible choices of $\pi$. For this $\pi$, the relation $\gamma_\pi$ turns out to be also transitive, and whence an equivalence relation on~$O$. 

\begin{restatable}{lemma}{alphamaximal}\label{lemma:alpha}
    Let $\urgraph =(\urdomain; \edge)$ be a digraph, let $\Omega\leq \Aut(\urgraph)$ be oligomorphic, and let $O$ be an $\omega$-class. Then there exists a symmetric  $\Omega$-orbit-labelled path $\pi$ starting at $O$ such that $\gamma_\pi$ is largest  (with respect to inclusion) amongst all possible choices of $\pi$.
\end{restatable}

Picking $\pi$ as in~\Cref{lemma:alpha}, we indeed obtain an equivalence relation: if $(a,b), (b,c)\in\gamma_\pi$, then $(a,c)\in \gamma_{\pi+\pi}=\gamma_\pi$. This justifies the following definition.

\begin{definition}\label{defi:alpha}
    Let $\urgraph =(\urdomain; \edge)$ be a digraph, and let $\Omega\leq \Aut(\urgraph)$ be oligomorphic. For each 
    $\omega$-class $O$, we set $\alpha_O:=\gamma_\pi$ for $\pi$ a symmetric $\Omega$-orbit-labelled path starting at $O$ such that $\gamma_\pi$ is largest.  
    We set $$\alpha(\urgraph,\Omega):=\bigcup\limits_{O\text{ is an }\Omega\text{-orbit}} \alpha_O.$$
\end{definition}

Let us take note of some immediate consequences of the definition. First of all, $\alpha(\urgraph,\Omega) $ is a refinement of $\omega$ by construction, and every $\alpha$-class $A$ is pp-definable from $\urgraph$ together with $\omega$-classes and an arbitrary element $a\in A$.  Observe that  if $\pi$ is the symmetric $\Omega$-orbit-labelled path with $\alpha_O=\gamma_\pi$ for an $\omega$-class $O$, and $\pi'$ is an extension of $\pi$,  then also $\gamma_{\pi'}=\alpha_O$ by maximality of $\pi$. Moreover, every weakly connected component of $\urgraph$ is $\alpha(\urgraph,\Omega)$-stable as, by definition, any two $\alpha(\urgraph,\Omega)$-related elements are $\edge$-connected by a realisation of $\pi$. Finally, 
we remark that though every restriction of $\alpha(\urgraph,\Omega)$ to an $\omega$-class  is pp-definable from $\urgraph$ and $\omega$-classes, $\alpha(\urgraph,\Omega)$ will in general not be 
pp-definable on the entire  domain. The following definition formalises the properties A1)-A3) discussed in \Cref{sect:newideas}. Afterwards, in \Cref{lem:alphaFinitises}, we show that
$\alpha(\urgraph,\Omega)$ is a valid candidate for a refinement satisfying these properties.

\begin{definition}\label{defi:finitises}
    Let $\urgraph =(\urdomain; \edge)$ be a digraph and let $\Omega\leq \Aut(\urgraph)$ be oligomorphic. We say that a refinement $\alpha$ of $\omega$ \emph{finitises} $(\urgraph,\Omega)$ if it satisfies all of the following.
    \begin{enumerate}[label=A{{\arabic*}})]
        \item The equivalence  $\alpha$ is invariant under $\Omega$. \label{aitm:invariant}
        \item For all $\omega$-classes $O,P$ such that $O\edgeo P$, the restriction of $\edge$ to $O\times P$ induces a bijection between the factor sets  $O/\alpha$ and $P/\alpha$.\label{aitm:bijection}
        \item For every weakly connected component $D$ of $\urgraph$, the set  $D / \alpha$ of $\alpha$-classes is finite.\label{aitm:finite}
    \end{enumerate}
\end{definition}

\begin{restatable}{lemma}{alphafinitises}\label{lem:alphaFinitises}   
    Let $\urgraph =(\urdomain; \edge)$ be a digraph and let $\Omega\leq \Aut(\urgraph)$ be oligomorphic. Then $\alpha(\urgraph,\Omega)$ finitises $(\urgraph,\Omega)$.
\end{restatable}

In the course of the proof of~\Cref{thm:10}, we inductively replace $\urgraph$ with a proper induced subgraph $\subgraph$ whose domain $H$ is an $\alpha$-stable class of an $\Omega$-invariant equivalence. Thus, the orbits of $\Omega|_H$ correspond to the restrictions of $\Omega$-orbits.
However, on every $1$-orbit of $\Omega|_{\subdomain}$, the restriction of $\alpha(\urgraph,\Omega)$ to $\subdomain$ does not necessarily have to be pp-definable from $\subgraph$ any more (in particular, there is no reason for it to coincide with  $\alpha(\subgraph, \Omega|_{H})$). The approach, instead, is to recognise that the restriction of $\alpha(\urgraph,\Omega)$ to $\subdomain$ still provides an equivalence that inherits from $\alpha(\urgraph,\Omega)$ desirable properties, which we encapsulate  in the following definition:  

The strength of the following lemma lies, especially, in the fact that it does not apply only to $\alpha(\urgraph,\Omega)$, but to any equivalence finitising  $(\urgraph, \Omega)$.
\begin{restatable}{lemma}{alphasymmetry} \label{lem:alphasymmetry}
    Let $\urgraph =(\urdomain; \edge)$ be a digraph, let $\Omega\leq \Aut(\urgraph)$, and  let $\alpha$ 
    be an equivalence finitising $(\urgraph, \Omega)$. Assume that $S\subseteq \urdomain^2$ is a binary $\Omega$-invariant relation.
    Then all of the following hold:
    \begin{enumerate}
        \item For all $A,B\in G/\alpha$ contained in the same weakly connected component of $\urgraph / \alpha$, it holds that $\Omega_A=\Omega_B$.\label{aitm:permut}
        \item If  $X\subseteq G$ is $\alpha$-stable and such that all $\alpha$-classes from $X$ and $X+S$ are contained in the same weakly connected component of $\urgraph / \alpha$ as $X$,
        then $X+S$ is  $\alpha$-stable.\label{aitm:walk}
        \item For all $\omega$-classes $P,Q$ and  $\alpha$-classes $A,B\subseteq P$ such that $A,B,A+S$ and $B+S$ are contained in the same weakly connected component of $\urgraph / \alpha$, there exists a bijection between the sets $((A+S)\cap Q)/\alpha$ and $((B+S)\cap Q)/\alpha$.\label{aitm:neighbours}
    \end{enumerate}
\end{restatable}

In general, pp-definability in a factor digraph $\urgraph/\alpha$ of $\urgraph$ does not lift to pp-definability in $\urgraph$. However, if $\alpha $  finitises $(\urgraph, \Omega)$, then substantial consequences  can still be obtained   when utilising only a restricted form of pp-definitions in $\urgraph/\alpha$.  The following auxiliary statement will be used to blow-up (finite) subsets of $\urgraph/\alpha$ with desirable properties to subsets of $\urgraph$ retaining these properties. 

\begin{restatable}{lemma}{liftingtree}\label{lem:liftingTreeDefs}
 Let $\urgraph=(\urdomain;\edge)$ be a digraph, let $\Omega\leq \Aut(\urgraph)$ be oligomorphic, and let $\alpha$ be an equivalence finitising $(\urgraph, \Omega)$.
Let $S\subseteq G/\alpha$ be tree pp-definable in $\urgraph/\alpha$ with parameters. Then the $\alpha$-blow-up of $S$ is (tree) pp-definable from $\urgraph$ and $\alpha$-classes.
\end{restatable}

\subsection{Finitising equivalence on adjacent orbits}\label{sect:defalpha}

Our  goal for this section is to show that the $2$-conservative expansion of $\urgraph$ with respect to $\Omega$ pp-defines $\alpha(\urgraph,\Omega)$ on all pairs of adjacent 
$\omega$-classes (\Cref{l:alphaonpairs}).
Recall that the concept of a merge as introduced in~\Cref{sec:relations} allows for labelling paths by pairs of $\omega$-classes. We now formalise a more restrictive way of the former.

\begin{definition}\label{defi:properSeparation}
Let $p$ be a path, and let $\pi=(p,(P_1,\ldots,P_n))$ and $\pi'=(p,(P_1',\ldots,P_n'))$ be two realisable $\Omega$-orbit labellings of $p$. We say that $(\pi,\pi')$ is \emph{properly separated} if for all $i$, the $\Omega$-orbit labelling $(p,(P_1,\ldots,P_{i},P_{i+1}',\ldots,P_n'))$ is not realisable. Note that this simply means that $P_i$ and $P_{i+1}'$ are not adjacent in the direction   prescribed by $p$. In a different presentation, for $p=(E_1, \dots, E_{n-1}) $  and for every $i \in \br{n-1}$, the merge  $$\pi\merge \pi'=\binom{P_1}{P'_1} E_1 \binom{P_2}{P'_2} E_2\; \cdots\; E_{n-1} \binom{P_{n}}{P'_{n}} $$ forbids the arrows $\searrow $ or $ \nwarrow$, respectively, in the direction set by $E_i$.
\end{definition} 
Observe that if $(\pi,\pi')$ is properly separated, then so is $(-\pi',-\pi)$.
We state a technical lemma that will serve as a crucial ingredient in the proof of~\Cref{l:alphaonpairs}.

\begin{restatable}{lemma}{labellings}\label{lemma:labellings}
    Let $\urgraph=(\urdomain;\edge)$ be a smooth digraph, let $\Omega\leq \Aut(\urgraph)$ be oligomorphic, and suppose that $\urgraph / \Omega$ has no loops. Assume that $O,P$ are $\edge$-adjacent $\omega$-classes, and $\pi$  is a realisable $\Omega$-orbit-labelled path   from $O$ to $O$. Then there exists a properly separated pair $(\pi',\rho')$ such that \begin{itemize}
            \item $\pi'$ is a realisable  $\Omega$-orbit-labelled path from $O$ to $O$ extending $\pi$,
            \item $\rho'$ is a realisable  relabelling  of $\pi'$ from $P$ to $P$;
        \end{itemize}     
\end{restatable}

Finally, we apply~\Cref{lemma:labellings} to prove the main result of this section. A short sketch of the proof is presented below. 

\begin{restatable}{lemma}{alphaonpairs}\label{l:alphaonpairs} 
    Let $\urgraph=(\urdomain;\edge)$ be a smooth digraph, let $\Omega\leq \Aut(\urgraph)$  be oligomorphic,  and suppose that $\urgraph / \Omega$ has no loops. Let $O, P$ be two $\edge$-adjacent $\omega$-classes.
    Then $\urgraph$ together  with all unions of pairs of $\omega$-classes pp-defines the restriction of $\alpha(\urgraph,\Omega)$ to $O\cup P$.  
\end{restatable}

\begin{proof}[Sketch of the proof] 
  Let us set $\alpha:=\alpha(\urgraph,\Omega)$. We apply~\Cref{lemma:labellings} to the symmetric $\Omega$-orbit-labelled path $\pi$ defining $\alpha|_O$ and the orbits $O,P$ to obtain an extension $\pi'$ of $\pi$ from $O$ to $O$, and a relabelling $\rho'$ thereof from $P$ to $P$ such that the pair $(\pi', \rho')$ is  properly separated. Similarly, we proceed with $\mu$ defining $\alpha|_P$ and the orbits $P,O$, obtaining  an extension $\mu'$ and a relabelling $\nu'$. We set $\kappa:=\pi'-\nu'$ (an $\Omega$-orbit-labelled path from $O$ to $O$), and $\lambda:=\rho'-\mu'$ (an $\Omega$-orbit-labelled path from $P$ to $P$ which is a relabelling of $\kappa$), and note that $(\kappa,\lambda)$ is properly separated.

The idea now is to show that the induced relations $\gamma_{\kappa}$ and $\gamma_\lambda$ are both $\alpha$-stable, and induce a bijection on the set of $\alpha$-classes contained in $O$ and $P$, respectively (thereby using~\Cref{aitm:bijection} of~\Cref{defi:finitises}). Thus, by replacing $\kappa$ and $\lambda$ by a sufficiently large powers of themselves, we may assume that they both induce the identity bijection. Proper separateness  now yields that any realisation of $\kappa\merge \lambda$ starting at an element in $O$ must in fact be a realisation of $\kappa$ (and hence end in $O$). Switching the roles of  $O$ and $P$ in the preceding arguments, and intersecting the obtained relations gives the restriction of $\alpha$ to $O\cup P$ as desired.
\end{proof}

\section{Reductionistic sets}\label{sect:reductionistic}

In this section, we introduce the concept of reductionistic sets, which plays a key role in the proof of~\Cref{thm:10}: smaller domains to which we reduce inductively will always be reductionistic.

\begin{definition}\label{def:reductionistic}
Let $\Omega$ be a permutation group acting on a set $G$. An \emph{$\Omega$-reductionistic} set is any class $H\subsetneq G$ of an $\Omega$-invariant equivalence relation on $G$.
\end{definition}

We will use the following characterisation of reductionistic sets. 

\begin{restatable}{lemma}{redchar}\label{lem:reductionistic_characterization}
    $H\subsetneq G$ is $\Omega$-reductionistic if and only if $H$ is preserved by all $f\in \Omega$ with $f(\subdomain) \cap \subdomain \neq \emptyset$.
\end{restatable}

Clearly, in particular  any proper $\omega$-stable subset of $G$ is reductionistic.  Note that for an $\Omega$-reductionistic set  $H$, the $\Omega|_{\subdomain}$-orbits 
are precisely the restrictions of the $\Omega$-orbits to $H$. It follows that if $\Omega$ is oligomorphic, then so is $\Omega|_H$.

In the proof of~\Cref{thm:10}, we will be interested in reductionistic sets which are, additionally, $\alpha(\urgraph,\Omega)$-stable (so that we inductively reduce the number of $\alpha(\urgraph,\Omega)$-classes). The reason the domains to which we inductively reduce must be $\Omega$-reductionistic is that it is the only way to pass down the required properties of $\alpha(\urgraph, \Omega)$ to its restriction, as formalised by the following lemma:

\begin{restatable}{lemma}{redfinitises}\label{lemma:redfinitises}
    Let $\urgraph=(\urdomain;\edge)$ be a digraph, let $\Omega\leq \Aut(\urgraph)$ be oligomorphic, let $\alpha$ be an equivalence finitising $(\urgraph, \Omega)$,  and let $\subdomain\subseteq\urdomain$ be an $\alpha$-stable $\Omega$-reductionistic set. Then $\alpha|_{\subdomain\times\subdomain}$ finitises $(\subgraph , \Omega|_H)$.
\end{restatable}

We finish this section by stating another two  properties of reductionistic sets, namely that they are closed  under taking smooth parts, and under taking weakly connected linked components.

\begin{restatable}{lemma}{smoothpartreduc}\label{l:smoothpartreduc}
     Let $\urgraph=(\urdomain;\edge)$ be a digraph, let $\Omega\leq \Aut(\urgraph)$ be oligomorphic, and let $\alpha$ be an equivalence finitising $(\urgraph, \Omega)$.
    Then the smooth part of an $\alpha$-stable $\Omega$-reductionistic set  is itself $\alpha$-stable and $\Omega$-reductionistic. 
\end{restatable}

\begin{restatable}{lemma}{linkedcompreduc}\label{l:linkedcompreduc}
    Let $\urgraph=(\urdomain;\edge)$ be a digraph, let $\Omega\leq \Aut(\urgraph)$ be oligomorphic, let $\alpha$ be an equivalence finitising $(\urgraph, \Omega)$, and let $k\geq 1$. Assume that $K\subseteq G/\alpha$ is  a $k$-linked weakly connected component of $\urgraph / \alpha$. Then the $\alpha$-blow-up $K^\alpha$ of  $K$  is $\Omega$-reductionistic, and tree pp-definable from $\urgraph$ and $\alpha$-classes. 
\end{restatable}

\section{The Master Proof}\label{sect:thm10}

This section is devoted to the proof of~\Cref{thm:10}.
 In the proof, we want to proceed by induction on the number of $\alpha(\urgraph,\Omega)$-classes. Doing so, we need to ensure that not only we always cut out entire $\alpha(\urgraph,\Omega)$-classes when inductively reducing the domain (hence, the new domain  needs to be $\alpha(\urgraph,\Omega)$-stable), but also to be able to finitise the induced subgraph with the restriction of $\alpha(\urgraph,\Omega)$ (hence, the new domain needs to be reductionistic). Every induction step splits into two cases.

\subsection{The full case}\label{sect:full}

In the first case, we obtain a subgraph that we aim to use as a replacement for $\urgraph$ to continue the induction process. 

\begin{restatable}[Inductive step~\ref{item:inductive1}]{prop}{claimBK}\label{prop:claimBK12}
Let $\urgraph =(\urdomain; \rightarrow)$ be a smooth digraph, let $\Omega\leq\Aut(\urgraph)$ be oligomorphic, and suppose  that $\urgraph / \Omega$ has no loop. Let $\alpha$ be an equivalence finitising $(\urgraph, \Omega)$. If $\edgek$ is full on $\urgraph/\alpha$, and $\urgraph/\alpha$ is not $(k-1)$-linked, then $\urgraph$ together with $\alpha$-classes tree pp-defines an $\alpha$-stable,  $\Omega$-reductionistic set $K \subsetneq \urdomain$ such that $\urgraph|_K$ is smooth, and $(\urgraph|_K) / \alpha$ is $(k-1)$-linked. 
\end{restatable}
\begin{proof}
Let $\fingraph := \urgraph /\alpha$. As the $\alpha$-blow-up of any subset of~$J$ that is tree pp-definable from $\fingraph$ with parameters is pp-definable from $\urgraph$ with $\alpha$-classes by~\Cref{lem:liftingTreeDefs}, and $\alpha$-stable by definition, it suffices to work in the finite digraph $\fingraph$. By abuse of notation, we identify subsets of $J$ with their $\alpha$-blow-ups.

Let $K'$ be an arbitrary class of the $(k-1)$-linkedness equivalence on $J$.
By \Cref{l:linkedcompreduc}, $K'$ is $\Omega$-reductionistic, and tree pp-definable from $\urgraph$ and $\alpha$-classes. We employ the following claim, which is taken from~\cite[Claim 3.11 and Proposition 3.2]{cyclicterms}.

\begin{restatable}{claim}{fullcaseoldargument}\label{claim:fullcaseoldargument}
    Every weakly connected component of the smooth part of $K'$ is non-empty and tree pp-definable from $\urgraph$ and $\finsubdomain'$. Moreover, the relation $\edge$ is 
   $(k-1)$-linked on any such weakly connected component.
\end{restatable}

Take $K$ to be any weakly connected component of the smooth part of $K'$. By~\Cref{l:smoothpartreduc}, the smooth part $S$ of the $\Omega$-reductionistic set $K'$ is itself $\Omega$-reductionistic. Thus, since $K$ is $(k-1)$-linked and weakly connected,~\Cref{l:linkedcompreduc} applied to $\urgraph|_S$ yield that   $K$ is indeed $\Omega$-reductionistic. 
\end{proof}

\subsection{The proper case}\label{sect:thm7}

The second case splits into two subcases: either we again obtain a proper subgraph to which  we reduce, or we terminate with the desired $\alpha$-blow-up of  $\OR(\sigma, \sigma)$ for a proper equivalence $\sigma$ on a finite set;

\begin{restatable}[Inductive step~\ref{item:inductive2}]{prop}{thmseven}\label{thm:7}
    Let $\urgraph =(\urdomain; \rightarrow)$ be a smooth digraph, let $\Omega\leq\Aut(\urgraph)$ be oligomorphic, and suppose that $\urgraph / \Omega$ has no loop. Let $\alpha$ be an equivalence finitising
    $(\urgraph, \Omega)$ such that $\urgraph /\alpha$ is $k$-linked for some $k\geq 1$. If $\edgek $ is not full on $\urdomain / \alpha$,  then $\urgraph$ and $\Omega$-orbits together with the restrictions of $\alpha$ to unions of pairs of $\edge$-adjacent $\omega$-classes 
    pp-define one of the following:
    \begin{enumerate}
        \item an $\omega$-stable set $\subdomain\subsetneq \urdomain$ such that $\urgraph|_\subdomain$ is smooth, and $(\urgraph|_\subdomain)/\alpha$ is $k$-linked; \label{itm:subgraph}
        \item the  $\alpha$-blow-up of  $\OR(\sigma, \sigma)$ for a proper equivalence $\sigma$ on a subset of the finite set $G/\alpha$. \label{itm:blow-up}
    \end{enumerate}
\end{restatable}

Recall that an $\omega$-stable set as in~\Cref{itm:subgraph} is, in particular, $\Omega$-reductionistic. We now provide an overview of the proof of \Cref{thm:7}.
For the remainder of this section, we fix $\urgraph, \Omega, \alpha$, and $k$ as in the formulation of~\Cref{thm:7}. We remark that in many parts, our construction resembles the proof of \cite[Theorem 7]{symmetries} applied to the factor digraph $\urgraph / \alpha$. The  main difference is being restricted to tree-pp-definitions only (to lift definitions to $\urgraph$ with~\Cref{lem:liftingTreeDefs}), along with the additional requirement to ensure $\alpha$-stability of relations pp-defined from $\edge$ and these liftings. 
Let us  define two types of relations which play a prominent role in universal algebra~\cite{Rosenbergclones}, and in its applications to Constraint Satisfaction Problems over finite templates~\cite{cyclicterms, MinimalTaylor,symmetries,Zhuk20}.

\begin{definition}
A proper binary relation $R \subsetneq \urdomain^2$ is \emph{central} if  there exists a non-empty subset $C \subsetneq \urdomain$ such that $C \times G \subseteq R$. If $R$ is central and $C$ is the maximal witnessing subset, then $C$ is called the \emph{centre} of $R$.  Accordingly, $R$ is \emph{central modulo $\alpha$} if $R$ considered as a relation on $\urdomain/\alpha$ is, and its centre will be referred to as \emph{$\alpha$-centre}.
\end{definition}
 
\begin{definition}
A relation $R \subseteq \urdomain^n$ is \emph{totally symmetric} if for every $(a_1, \dots, a_n)\in R$ and every permutation $\tau$ on the index set $\br n$ it also holds that $(a_{\tau(1)}, \dots, a_{\tau (n)})\in R$. We say that $R$ is \emph{totally reflexive} if $a_i = a_j$ for some $i\neq j$ implies $(a_1, \dots, a_n)\in R$. If $R$ is both totally symmetric and totally reflexive, then we call $R$ a \emph{TSR-relation}. 
\end{definition}

We now proceed with the first step of the construction.

\begin{restatable}{prop}{Marcinsmagic}\label{prop:Marcinsmagic}
    $\urgraph$ together with  $\Omega$-orbits and the restrictions of $\alpha$ to $\omega$-classes pp-define one of the following
    \begin{enumerate}
        \item a relation that is central modulo $\alpha$, or \label{itm:marcinsmagiccentralcase}
        \item the relation $\OR(T,T)$ for an $\alpha$-stable TSR-relation $T$
    \end{enumerate}
\end{restatable}
    
In order to start applying suitable adjustments of~\cite[Theorem 7]{symmetries}, we require pp-definable $\alpha$-stable $\OR$-relations in both cases of~\Cref{prop:Marcinsmagic}. The following proposition thus deals with its first case, making  use of the fact that we have access to the restrictions of $\alpha$ to unions of pairs of adjacent $\omega$-classes.

\begin{restatable}{prop}{trickT}\label{prop:trickT2}
    Let $R$ be a relation that is central modulo $\alpha$ and $\Omega$-invariant. Then   $\urgraph, R$,  $\alpha$-classes, $\omega$-classes, and the restrictions of $\alpha$ to unions of pairs of $\edge$-adjacent  $\omega$-classes pp-define $\OR(D_L,D_R)$ for $\omega$-stable sets $D_L,D_R \subseteq \urdomain$.  
\end{restatable}

    Now, with an $\alpha$-stable OR-relation at our hand, we are in the position to replicate the proofs of  \cite[Lemmata 20-23]{symmetries} with small modifications. Thus, our next goal is to get the relation $\OR(U,U)$ for a unary, $\alpha$-stable relation $U$. The constructions splits into~\Cref{lemma:20,lemma:21+22}, depending on the case distinction from~\Cref{prop:Marcinsmagic}.

    \begin{restatable}{lemma}{ltwenty}\label{lemma:20}
    Let $D_L,D_R\subsetneq\urdomain$ be $\omega$-stable. Then $\urgraph$, $\OR(D_L,D_R)$, $\Omega$-orbits, and the restrictions of $\alpha$ to $\omega$-classes pp-define
    \begin{itemize}
        \item a non-empty $\omega$-stable set $\subdomain\subsetneq \urdomain$ such that $\urgraph|_H$ is smooth, and $(\urgraph|_H)/\alpha$ is $k$-linked, or
        \item the relation $\OR(U,U)$ for a unary, $\omega$-stable $U\subsetneq \urdomain$.
    \end{itemize}
    \end{restatable}

\begin{restatable}{lemma}{ltwentyonetwo}\label{lemma:21+22}
    Let $T$ be a proper $\alpha$-stable TSR-relation on~$\urdomain$. Then  $\OR(T,T)$  pp-defines either $\OR(U,U)$, where $U$ is  a unary, $\alpha$-stable relation, or the $\alpha$-blow-up of $\OR(\sigma, \sigma)$ for a proper equivalence $\sigma$ on a subset of  $G/\alpha$.
\end{restatable}

Finally, we conclude the proof by a slight adjustment of \cite[Lemma 23]{symmetries}, which again uses the restrictions of $\alpha$ to unions of adjacent pairs of $\omega$-classes.

\begin{restatable}{lemma}{ltwentythree}\label{lemma:23}
    Let $U$ be a non-empty unary, $\omega$-stable relation. Then $\urgraph$,  $\OR(U,U)$, and the restrictions of $\alpha$ to unions of pairs of $\edge$-adjacent $\omega$-classes pp-define
    \begin{itemize}
         \item a non-empty $\omega$-stable set $\subdomain\subseteq \urdomain$ such that $\urgraph|_H$ is smooth and $(\urgraph|_\subdomain)/\alpha$ is $k$-linked, or
         \item the $\alpha$-blow-up of $\OR(\sigma, \sigma)$ for a proper equivalence $\sigma$ on  a subset of $G/\alpha$. 
    \end{itemize}
\end{restatable}

\subsection{Summary of the proof}\label{sect:summaryofproof}
We are now able to derive~\Cref{thm:10}  from~\Cref{prop:claimBK12,thm:7}. In the course of the proof, 
we will show that the reductionistic sets we obtain from~\Cref{prop:claimBK12} and~\Cref{thm:7} can indeed be used to inductively reduce the domain, i.e. that they consistently satisfy the inductive hypothesis (hence, their name \emph{reductionistic}).
\thmten*
\begin{proof}
Assuming that $\urgraph / \Omega$ does not have a loop, we will prove that the second item holds. In fact, we will prove the slightly weaker statement that the $2$-conservative expansion of $\urgraph$ with respect to $\Omega$ pp-constructs EVERYTHING \emph{with $\Omega$-orbits}, i.e.~we additionally assume access to these orbits. We start by arguing that this is sufficient.
To this end, let $\hat \urgraph$ denote the $2$-conservative expansion of $\urgraph$ with respect to $\Omega$, let $\hat \subgraph$ be its model-complete core, and let $\hat H$ denote the domain of $\hat \subgraph$. We can assume that $\hat \subgraph$ is a substructure of $\hat\urgraph$ by \cite[Remark 4.7.5]{Book}. 
We set $\Upsilon :=\Aut(\hat \subgraph)$. It follows from the proof of \cite[Theorem 5.7]{BKOPP-equations} that  $\Upsilon$-orbits are given by unions of $\Aut(\hat \urgraph)$-orbits restricted to~$\hat H$. In fact, since the $1$-orbits of $\Aut(\hat \urgraph)$ coincide with the $1$-orbits of $\Omega$ by construction of $\hat \urgraph$ and are contained in the signature of $\hat \subgraph$, the $1$-orbits of $\Upsilon$ are precisely the restrictions of  $1$-orbits of $\Omega$ to $\hat{\subdomain}$. Hence, the $2$-conservative expansion of $\urgraph|_{\hat\subdomain}$ with respect to $\Upsilon$ coincides with $\hat\subgraph$, and the digraph $(\urgraph|_{\hat\subdomain}) / \Upsilon$ does not have a loop. Moreover, we now argue that $\urgraph|_{\hat\subdomain}$ is smooth and has algebraic length~1. To this end, pick any homomorphism $h$ from $\hat \urgraph$ into $\hat \subgraph$. Then the image of any  walk in $\urgraph$ witnessing its algebraic length~1 witnesses the same for $\subgraph$. Moreover, if $a\in \hat H$ is arbitrary, then it has incoming and outgoing edges in $\urgraph$, hence $h(a)$ has incoming and outgoing edges in $\subgraph$. We know that  $h$ preserves $\Upsilon$-orbits since they are pp-definable in $\hat \subgraph$; thus, $a$ belongs to the same $\Upsilon$-orbit as $h(a)$ and also has incoming and outgoing edges in $\subgraph$. Whence,   $\subgraph$ is smooth. Now assume that the above-mentioned weakening of the second item of the theorem applies to $\hat\subgraph$ and $\Upsilon$; that is, $\hat\subgraph$ together with  $\Upsilon$-orbits pp-constructs EVERYTHING. We claim that then also $\hat{\urgraph}$ pp-constructs everything. Indeed,  since~$\hat\subgraph$ is a model-complete core, all $\Upsilon$-orbits are pp-definable from~$\hat\subgraph$. Lastly,~$\hat\subgraph$ is homomorphically equivalent to~$\hat{\urgraph}$, whence it is  pp-constructible from~$\hat{\urgraph}$.  
Since pp-constructions are closed under compositions by~\cite[Corollary 3.10]{wonderland}, we conclude that $\hat{\urgraph}$ pp-constructs EVERYTHING, as had to be shown.

Let us now set $\alpha:=\alpha(\urgraph,\Omega)$ -- as defined in~\Cref{defi:alpha}~-- and recall that $\alpha$ finitises $(\urgraph, \Omega)$ by~\Cref{lem:alphaFinitises}. 
Assume that $\urgraph / \Omega$ has no loop; we claim that the $2$-conservative expansion of $\urgraph$ with respect to $\Omega$ pp-defines with $\Omega$-orbits and $\alpha$-classes the $\alpha$-blow-up of $\OR(\sigma,\sigma)$ for a proper equivalence $\sigma$ on a finite set $K=\{K_1,\dots,K_n\}\subseteq \urdomain / \alpha$. Let us illustrate why this indeed finishes the proof: recall that $\alpha$-classes are pp-definable with parameters from $\urgraph$ and $\omega$-classes. Thus, by the proof of~\cite[Lemma~3.9]{wonderland}, the expansion of $\urgraph$ by the $\alpha$-blow-ups of $K_1,\dots,K_n$ is  pp-constructible from $\urgraph$ and $\omega$-classes, and the expansion of $\urgraph$ by the $\alpha$-blow-up of $\OR(\sigma,\sigma)$ is pp-constructible from the $2$-conservative expansion of $\urgraph$ with respect to $\Omega$ together with $\Omega$-orbits (we remark  that the stronger assumption in~\cite[Lemma~3.9]{wonderland} of the structure being a model-complete core is not needed in the proof, except for access to 1-orbits). By a folklore observation, the structure $\finsubgraph:=(\finsubdomain; \OR(\sigma,\sigma),\{K_1\},\dots,\{K_n\})$ pp-constructs  EVERYTHING (see for example the presentation in~\cite[Proposition 5]{symmetries}, where it is formulated for pp-interpretations instead of pp-constructions; 
since the former is just a special form of the latter, we skip the definition). As $\finsubgraph$ is homomorphically equivalent to its $\alpha$-blow up, it then follows that it is pp-constructible from the 2-conservative expansion  of $\urgraph$ with $\Omega$-orbits, and the theorem follows.
    
It is left to show that the $2$-conservative expansion of $\urgraph$ pp-defines with $\alpha$-classes and $\Omega$-orbits the $\alpha$-blow-up of $\OR(\sigma,\sigma)$. Since $\urgraph$ has algebraic length $1$, there exists a closed $\edge$-walk $(a_0,\ldots,a_0)$  of algebraic length $1$.  Let $\weakcomp$ be the weak connected component of $a_0$. In particular, as observed below \Cref{defi:alpha}, $\weakcomp$ is $\alpha$-stable. Moreover, $\weakcomp$ is $\Omega$-reductionistic and pp-definable from $\urgraph$ together with $\alpha$-classes. Indeed, let us set $\subgraph:= \urgraph|_{\weakcomp}$. Evidently, $\subgraph/\alpha $ is a weakly connected smooth digraph of  algebraic length $1$. By the properties of $\alpha$ as listed in~\Cref{defi:finitises}, $\subgraph/\alpha $ is finite.  Since a weakly connected smooth finite digraph has algebraic length $1$ if and only if it is $k$-linked for some $k$ (this follows immediately  for example from~\cite[Claim~3.8]{cyclicterms}), \Cref{l:linkedcompreduc}  yields that the $\alpha$-blow-up $(\weakcomp/\alpha)^{\alpha}$ -- which coincides with $H$ by $\alpha$-stability --  is $\Omega$-reductionistic and pp-definable from $\urgraph$ and $\alpha$-classes.

Let $\alpha_H:=\alpha|_{\subdomain\times\subdomain}$, and let $\omega_H$ denote the $1$-orbit equivalence of $\Omega_H$ (which coincides with $\omega|_{H \times H})$. We claim that~$\subgraph$ pp-defines together with $\alpha_H$-classes, $\Omega|_H$-orbits, and the restrictions of $\alpha_{\subdomain}$ to all pairs of adjacent $\omega_H$-classes the $\alpha$-blow-up of  $\OR(\sigma, \sigma)$ for a proper equivalence $\sigma$ on a subset of the finite set $H/\alpha$. This suffices to prove the statement of the theorem; indeed, as $\subdomain$ is  pp-definable from $\urgraph$   with $\alpha$-classes as shown before, $\alpha_H$-classes and $\Omega|_H$-orbits are pp-definable from $\urgraph$ with $\alpha$-classes and $\Omega$-orbits. Furthermore, recall that by~\Cref{l:alphaonpairs}, the $2$-conservative expansion of $ \urgraph$ with respect to $\Omega$ pp-defines the restrictions of $\alpha$ to all pairs of adjacent $\omega$-classes. Thus the $2$-conservative expansion of $ \urgraph$ with respect to $\Omega$ pp-defines with $\alpha$-classes the restrictions of $\alpha_{\subdomain}$ to all pairs of adjacent $\omega_H$-classes. 
    
Our strategy  is to successively shrink the domain of $\subgraph$ until we obtain the desired statement. When shrinking, we always remove entire $\alpha$-classes, meaning that this process must terminate after a finite number of steps. At every step, we replace $(\subgraph, \alpha_H,\Omega|_H) $  by their restrictions $(\subgraph', \alpha_{H'},\Omega|_{H'}) $ to a proper $\alpha$-stable $\Omega|_H$-reductionistic subset $\subdomain'\subsetneq \subdomain$ while   keeping the properties that $\subgraph|_{\subdomain'}$ is smooth, and that $(\subgraph|_{\subdomain'}) / \alpha_{H'}$ is $\ell$-linked for some $\ell\leq k$. We update $k$ to $\ell$. Observe that after each step, $\subgraph$, $\Omega|_{\subdomain}$, and $\alpha_\subdomain$ satisfy the assumptions of both~\Cref{thm:7} and~\Cref{prop:claimBK12}.  Indeed, $H$ being $\Omega$-reductionistic implies that $\Omega|_H$ is an oligomorphic subgroup of $\Aut(\subgraph)$. Moreover, as $H$ is $\alpha$-stable, $\alpha_{H} $ is a refinement of the $1$-orbit equivalence $\omega_H$ of $\Omega|_H$ (which coincides with $\omega|_{H\times H}$), and $\alpha_H$ finitises $(\subgraph ,\Omega|_H) $ by~\Cref{lemma:redfinitises}. Finally, $\subgraph/\alpha_H $ has no loops as $\subgraph/\omega_H$ does not. 

In every step, we now distinguish the following two cases:
\begin{outline}[enumerate]
        \1\label{item:inductive1}   If  $\edgek$ is full on $\subdomain/\alpha_H$, then \Cref{prop:claimBK12} shows that $\subgraph$ pp-defines together with $\alpha_H$-classes a proper $\alpha_H$-stable, $\Omega|_H$-reductionistic subset $\subdomain'\subsetneq\subdomain$ of~$\urdomain$ such that $\subgraph|_{\subdomain'}$ is smooth, and that $(\subgraph|_{\subdomain'})/\alpha_{H'}$  is  $(k-1)$-linked.
        
        \1 \label{item:inductive2} If  $\edgek$ is not full on $\subdomain/\alpha_H$, then we apply~\Cref{thm:7}.
       Hence, $\subgraph$ together with $\alpha_H$-classes, $\Omega|_H$-orbits, and the restrictions of $\alpha_{\subdomain}$ to all pairs of adjacent $\omega_H$-classes pp-define either
       \2 a proper $\omega_H$-stable subset $\subdomain'$ of $\subdomain$ such that $\subgraph|_{\subdomain'}$ is smooth, and $(\subgraph|_{\subdomain'})/\alpha_{H'}$ is $k$-linked, or
       \2 the $\alpha_H$-blow-up of $\OR(\sigma, \sigma)$ for a proper equivalence $\sigma$ on a subset of $H/\alpha_H$.
    \end{outline}

Observe that there is no loopless digraph on a one-element vertex-set that is $k$-linked for any $k\geq 1$.  In particular, the induction needs to terminate with b) at some step. Now, since  the $\alpha_H$-blow-up of $\OR(\sigma, \sigma)$ coincides with the $\alpha$-blow-up of $\OR(\sigma, \sigma)$, we indeed pp-construct with the\; \ldots   

\begin{quote} 
    \emph{blow-up\; \ldots}\\ 
    \emph{EVERYTHING.}
\end{quote}
\end{proof}

\bibliographystyle{alpha}
\bibliography{references}

\Addresses

\newpage
\section*{Appendix}
\appendix

\setcounter{section}{0}
\refstepcounter{section} 

\setcounter{subsection}{0}

\subsection{Missing proofs from \texorpdfstring{\Cref{sect:finitising}}{Section~\ref{sect:finitising}}}\label{sect:missingfinitising}

\alphamaximal*

\begin{proof}
    Note that for any two realisable $\Omega$-orbit-labelled paths $\pi,\rho$ such that $\rho$ is symmetric and such that $\pi+\rho$ is well-defined, and for all $a,b\in G$, we have that if $b$ is reachable from $a$ via a realisation of $\pi$, then the same can be achieved via a realisation of $\pi+\rho$; a similar statement holds for reachability via $\rho$ if $\pi$ is symmetric.
    Hence, for any symmetric $\Omega$-orbit-labelled paths $\pi,\rho$  starting (and ending) at  $O$  we have $\gamma_{\pi+\rho}\supseteq \gamma_\pi\cup \gamma_\rho$. It follows that the set of relations of this form is upwards directed, and since each such relation is invariant under $\Omega$, the  oligomorphicity of $\Omega$ implies that the number of such relations is finite. Hence, the set has a largest  element.
\end{proof}

\alphafinitises*

\begin{proof}
    For the third item, note  that  $\Omega$ acts on $O$ for every  $\omega$-class $O$. Moreover, since $\alpha_O$ is pp-definable from the $\omega$-classes and $\urgraph$, it is invariant under $\Omega$. Hence $\Omega$ acts on $\alpha_O$-classes, and also on all $\alpha$-classes. 
    
    To see that $\alpha$ satisfies the second item, let $O,P$ be $\omega$-classes with $O\edgeo P$. For any $A\in O/\alpha$ there exists $B\in P/\alpha$ with $A\edge B$ since $\Omega$ acts transitively on $O$. Suppose that $C\in P/\alpha$ satisfies $A\edge C$; we claim that  $B=C$. There exist $a,a'\in A$, $b\in B$, and $c\in C$ such that $a\edge b$ and $a'\edge c$. The fact that $\alpha(a,a')$ holds is witnessed by a symmetric $\Omega$-orbit-labelled $\edge$-path $\pi$ with a realisation from $a$ to $a'$, which we may now extend on both sides by $\edge$ and $\ledge$ and the label $P$ to see that $\alpha(b,c)$ holds as well. Hence, $B=C$.  Switching the roles of $O$ and $P$, we arrive at the desired conclusion.
    
    To prove the third item, pick an arbitrary $a\in D$. Then any two elements $b,c\in D$ that  belong to the same $\Omega_a$-orbit  are $\alpha$-equivalent: since $a,b\in D$, there exists an $\Omega$-labelled $\edge$-path $\pi$  with a realisation from $a$ to $b$;  $\pi$ then  also has a realisation from $a$ to $c$, obtained by applying any $f\in\Omega_a$ sending $b$ to $c$ to the original realisation. Since $\Omega_a$ is oligomorphic, the statement follows.
\end{proof}

\alphasymmetry*

\begin{proof}
    To prove the first item, let $f\in \Omega_A$. Since $A,B$ are contained in the same weakly connected component of $\urgraph / \alpha$, there exists an $\Omega$-orbit-labelled path which has a realisation from $A$ to $B$. By \Cref{aitm:bijection} in \Cref{lem:alphaFinitises}, this path yields a bijection between the $\omega$-classes of $A$ and $B$, respectively which is preserved by $\Omega$ by \Cref{aitm:bijection} in \Cref{lem:alphaFinitises}. As $f\in\Omega_A$, it follows that $f\in\Omega_B$, and since $f\in\Omega_A$ was arbitrary, it holds that $\Omega_A\subseteq \Omega_B$. The other inclusion follows by a dual argument.

    For the second item, let $a\in X+S$, and let $b$ satisfy $\alpha(a,b)$.  Then $a,b$ belong to the same $\omega$-class $O$, and thus there exists $f\in\Omega$ with $f(a)=b$. Observe  that $f\in \Omega_A$ for the $\alpha$-class $A$ of $a,b$, and hence $f$ fixes all $\alpha$-classes in the weakly connected component of $A$ by \Cref{aitm:permut} in \Cref{lem:alphaFinitises}. Since $a\in X+S$, there exists $a'\in X$ with $S(a',a)$.
    Since $S$ is pp-definable from $\Omega$-invariant relations, it is $\Omega$-invariant itself, and we have $S(f(a'),b)$  and also  $\alpha(a',f(a'))$. Since $X$ is $\alpha$-stable, $f(a')\in X$, whence $b\in X+S$.

    To see the third item, pick any $f\in\Omega$ with $f(A)=B$ in the (transitive) action of $\Omega$ on $\alpha_O$-classes. Since $f$ preserves $S$, it sends $A+S$ into $B+S$; since it preserves $\alpha_Q$, it maps the $\alpha_Q$-classes in  $(A+S)\cap Q$ injectively into those in $(B+S)\cap Q$. The same argument with $f^{-1}$ then implies the statement.
\end{proof}

\liftingtree*

\begin{proof} 
Say that $\phi(x,A_1,\ldots,A_n)$, where $A_1,\ldots,A_n\in G/\alpha$, is the tree pp-formula defining $S$. Denote the existentially quantified variables of $\phi$ by $z_1,\ldots,z_m$. We claim that the formula $\phi'(x)\equiv \exists y_1,\ldots,y_n\; \phi(x,y_1,\ldots,y_n)\wedge A_1(y_1)\wedge\cdots\wedge A_n(y_n)$ defines the $\alpha$-blow-up of $S$ in $\urgraph$ (together with the predicates $A_i$). 
Indeed, it is clear that if $\phi'(a)$ holds for some $a\in G$, then the classes of the witnesses of this fact (i.e. the values for $z_1,\ldots,z_m$ as well as $y_1,\ldots,y_m$) witness that $\phi(a/\alpha)$ holds as well (here we do not make use of the fact that $\phi$ is a tree pp-formula). For the converse, suppose that $\phi(a/\alpha)$ holds, as witnessed by values $B_i$ for $z_i$. Then, starting at $a$ as a representative of its $\alpha$-class, we can successively pick  representatives $b_i$ in each class $B_i$ and $a_i$ in $A_i$ following the tree given by the shape of the formula $\phi'$, and so that all conjuncts for the values chosen so far  are satisfied in $\urgraph$. Here we use in every step the fact that if $A,B$ are $\alpha$-classes with $A\edge B$, then for any $a\in A$ there exists $b\in B$ with $a\edge b$; this follows from ~\Cref{lem:alphasymmetry} (setting $X:=B$ and $S:=\ledge$).
\end{proof}

\subsection{Missing proofs from \texorpdfstring{\Cref{sect:defalpha}}{Section~\ref{sect:defalpha}}}\label{sect:missingdefalpha}

The following folklore result will be needed in the proof of~\Cref{lemma:labellings}. 
\begin{restatable}{lemma}{Euklid}[Euclid's hammer]\label{lem:euklid} Let $k, \ell$ be coprime positive integers. Then every $n > k\ell $ coprime with $k$ and $\ell$  can be written in the form $n=sk+t\ell$ with $s,t>0$.
\end{restatable}

\begin{proof} 
Every $n\in \mathbb N$ can be written uniquely in the form $n=sk+t\ell$, with $s \in \mathbb Z$ and $0 \leq t <k$, or $0 \leq s < \ell$ and $t \in \mathbb Z$. Indeed, take $a,b \in \mathbb Z$ with $ak+b\ell=n$. Without loss of generality, assume $b>0$. Then $b$ can be written uniquely as $b=t+mk$ for some $m \in \mathbb N, 0 \leq t <k $. Now,  $t$ and $s:=a+m\ell$  have the required properties. Observe that the unique composition for any $n>k\ell$ not divisible by $k$ satisfies  $s,t >0$.
\end{proof}

\labellings*

\begin{proof}
    Assume without loss of generality that $O\edgeo P$ holds. We first show how to reduce the lemma to the claim below.
\begin{restatable}{claim}{cyclecombinatorics}\label{cyclecombinatorics}
        Whenever $U,V,U',V'$ are arbitrary $\omega$-classes with $U\edgeo U'$ and $V\edgeo V'$, and such that one of the following holds.
        \begin{enumerate}
            \item $U\edgeo V$, or\label{itm:edge}
            \item $U\ledgeo V$\label{itm:ledge}
        \end{enumerate}
        Then there exists a realisable $\Omega$-orbit-labelled path $\mu$ from $U$ to $V$ and a realisable relabelling $\mu'$ thereof from $U'$ to $V'$ such that $(\mu,\mu')$ is  properly separated, and such that $\mu$ extends $U\edgeo V$ or $U\ledgeo V$, respectively.
\end{restatable}

 The  lemma then follows by induction. Indeed, let  $\pi=(p,(O_1,\ldots,O_n))$, where $O_1=O_n=O$. In the induction, we prove that  the restriction $\pi_j$ of $\pi$ to the first $j$ orbits (and hence $j-1$ edges), has a realisable  extension $\mu_j$ from $O$ to $O_j$ with a realisable properly separated relabelling $\mu_j'$ thereof from $P$ to an arbitrary $P_j$ such that $O_j\edgeo P_j$ holds (such $P_j$ exists by smoothness). This is trivial for $j=1$, and thus follows by induction using~\Cref{cyclecombinatorics} in the induction step. Finally, choosing $P_n:=P$, and setting $\pi':=\mu_n$ and $\rho':=\mu_n'$ then finishes the proof.

    \begin{proof}[Proof of~\Cref{cyclecombinatorics}] 
    Let us first observe that walks in $\urgraph / \Omega$ always give rise to realisable $\Omega$-orbit-labelled paths. We will therefore sometimes abuse the notation and will not distinguish between these two objects.
    We first show how to reduce~\Cref{itm:ledge} where $U\ledge V$ to~\Cref{itm:edge} where $U\edge V$. Assume the former holds, and pick an arbitrary $\omega$-class $Q$ with $Q\edge V$; it exists by smoothness of $\urgraph$. Consider the path $q:=(\ledge,\ledge)$ and the $\Omega$-orbit labellings  $\nu:=(q,(U,V,Q))$ and $\nu':=(q,(U',U,V))$ thereof. Then $$\nu \merge \nu' = \binom{U}{U'} \ledgeo \binom{V}{U} \ledgeo \binom{Q}{V}, $$ i.e. $(\nu, \nu')$ is properly separated. Recall that $Q\edgeo V$ and $V\edge V'$. Assuming that the conclusion of the claim holds true in~\Cref{itm:edge}, we may apply it after substituting $U$ by $Q$ and $U'$ by $V$. Hence,  we obtain  a properly separated pair $(\mu, \mu')$ where $\mu $ is from $Q $ to $V $ extending $Q \edgeo V$, and $\mu'$ is from $V $ to $V' $. Then, $\nu +\mu$ is 
    extends $U \ledgeo V$, and  the properly separated pair $(\nu +\mu, \nu' + \mu')$  shows the claim in \Cref{itm:ledge}. 

    We thus henceforth assume \Cref{itm:edge}, i.e. $U\edge V$. For $W\in G / \Omega$, we will denote by $\comp(W)$ the \emph{strongly connected component} of $\urgraph / \omega$ containing $W$, i.e. the maximal (with respect to inclusion) subset $\comp(W) \subseteq G / \Omega$ which contains $W$, and has the property that for all $W'\neq W''\in \comp(W)$, there exist both a $\edgeo$-forward walk from $W'$ to $W''$, and a $\edgeo$-forward walk from $W''$ to $W'$. Note that  every strongly connected component containing at least two distinct $\omega$-classes contains a cycle. In particular, by smoothness of $\urgraph/\omega$, every strongly connected component of size at least two  is itself smooth.  We will distinguish cases based on the strongly connected components 
    $\comp(U), \comp(V), \comp(V')$ of $U,V,V'$ in $\urgraph / \Omega$. \bigskip

    \textbf{Case 1.} $\comp(U)= \comp(V)=\comp(V')$: Let $C=(C_0,\ldots,C_{k-1})$ be a cycle (i.e. $C_i\edge C_{i+1}$ for all $i$, computing modulo $k$) of smallest length $k$ in $\comp(U)$. Such a cycle exists since $\comp(U)$ is a strongly connected component containing at least two distinct $\omega$-classes $U$ and $V$ and $\urgraph / \Omega$ is finite by the oligomorphicity of $\Omega$. Let us fix a $\edge$-backward walk from $U$ to $C$, i.e.~a sequence $(U=U_{n-1},\ldots,U_0)$ such that $U_{i+1}\ledge U_{i}$ for all $i$, and such that $U_0$ appears in $C$; by relabelling $C$, we may assume  $U_0=C_0$. We set $U_{n}:=V$ and $U_{n+1}:=V'$, and let $(U_{n+1},\ldots, U_{q-1}=C_0)$ be a $\edge$-forward walk from $U_{n+1}$ to $C_0$. Note that we have $U_i\edge U_{i+1}$ for all $i$, and that $V'$ is a successor of $V$ in the sequence $Q:=(U_0,\ldots,U_{q-1})$.

    We will successively construct $\mu$ and $\mu'$ by appending realisable $\Omega$-orbit-labelled paths (i.e. walks in $\urgraph / \Omega$). Let us first consider $\mu \merge \mu'$ defined by
    \begin{equation*}
      \binom{U_{n-1}}{U'}\ledgeo \binom{U_{n-2}}{U_{n-1}}\ledgeo \dots\ledgeo\binom{U_0=C_0}{U_1}\ledgeo \binom{C_{k-1}}{C_0}.
    \end{equation*}
    Let us observe that the pair $(\mu,\mu')$ is properly separated as $\urgraph / \Omega$ does not contain any loop.

    Let us now append to $\mu$ the $\edge$-forward walk  which first follows the walk $(C_{k-1},C_0=U_0,\dots,U_{q-1}=C_0)$ (taking $q$ steps), and then walks $k$ steps along $C$, thus finishing at $C_0$. Next, append to $\mu'$ the $\edge$-forward walk which starts at $C_0$, takes $k$ steps on $C$ until $C_0$, then walks $q-1$ steps  along $Q$ to $U_{q-1}=C_0$, and then takes a final step from $C_0$ to $C_1$. Note that the two walks have the same number of steps, namely $q+k$. Moreover, after appending the walks to $\mu$ and $\mu'$, respectively, $(\mu, \mu')$ remains properly separated. Indeed, suppose otherwise, i.e. there exists an edge from the $i$-th label $K$ of $\mu$ to the  $(i+1)$-st label $L$ of $\mu'$ for some $i$. But then $K,L$ would have to lie in a cycle of length $k-1$ on $P\cup C\subseteq \comp(U)$, contradicting the minimality of $k$. 
    
    Finally, we append to $\mu \merge \mu'$ the walks \begin{equation*}
        \binom{C_0}{C_1} \ledgeo \binom{C_{k-1} }{C_0} \ledgeo \dots \ledgeo \binom{C_0=U_{q-1}}{C_1} \ledgeo \binom{U_{q-2} }{U_{q-1}}\ledgeo\dots \ledgeo\binom{U_{n}=V}{U_{n+1}=V'}.
    \end{equation*}  Clearly, $(\mu, \mu')$ remains properly separated, and by construction, $\mu$ extends $U \edgeo V$.
    
\bigskip

\textbf{Case 2.} $\comp(U) \neq \comp(V):$ Let $C=(C_0, \dots, C_{k-1}) $ be the nearest cycle reachable by a $\edgeo$-backward walk from $U$. The condition on reachability of $C$ yields a $\edgeo$-backward walk  $(U=U_{n-1},\ldots,U_0)$ such that $U_i\ledge U_{i+1}$ for all $i$, $U_0$ appears in $C$, and $U_1,\dots,U_{n-1}$ do not appear in any cycle; by relabelling $C$, we may assume  $U_0=C_0$. 
Let $D=(D_0,\ldots, D_{\ell-1})$ be the nearest cycle reachable by a $\edgeo$-forward walk from $V$ via $V'$. Note that such cycles exist since $\urgraph / \Omega$ is smooth and finite. We will from now on consider all indices in $C$ modulo $k$, and in $D$ modulo $\ell$.

As in the previous case, we set $U_{n}:=V, U_{n+1}:=V'$, and let $(U_{n+1},\ldots, U_{q-1}=D_0)$ be a $\edge$-forward walk from $U_{n+1}$ to $D_0$ such that $U_{n+1},\dots,U_{q-2}$ is contained in no cycle; if $V'$ is contained in $D$, we can relabel $D$ so as  $V'=D_0$. Note that we have $U_i\ledge U_{i+1}$ for all $i$, and that $V'$ is a successor of $V$ in the sequence $Q:=(U_0,\ldots,U_{q-1})$. As above, we will successively construct $\mu$ and $\mu'$ by appending walks in $\urgraph / \Omega$.

To this end, consider the induced subgraph $\mathbb C$ of $\urgraph/\omega$ on the vertices  $C_0, \dots, C_{k-1}$. For technical reasons, we distinguish the following two subcases based on whether or not $\mathbb C$ includes a particular edge. In either case, we shall assign values to $p,j,i,m, m'$ that will thereafter allow for similar constructions in both cases. \bigskip

\textbf{Case 2.1.} Let us first assume that $\mathbb C$ does not include any of the edges $C_{k-1}\edgeo C_1$ and $C_{k-2}\edgeo C_0$. In this case, we set $p:=1$, $j:=0$, $i:=1$, and $m,m'$ to be any numbers such that $m'>q-1$ (which is the length of $Q$), and $m'k=m\ell$ (where $k$ is the length of $C$, and $\ell$ is the length of $D$). \bigskip

\textbf{Case 2.2.} If $\mathbb C$ includes the edge $C_{k-1}\edgeo C_1$, then this means that $\mathbb C$ includes both the cycle $C$ of length $k$ and the cycle $C_1,\dots,C_{k-1}$ of length $k-1$. By~\Cref{lem:euklid}, we can find $s,t> 0$ such that $k+sk+t(k-1)$ is coprime with $\ell$. Indeed, apply the Lemma to $k-1$ and $k$, and choose $n$ coprime with $\ell$ such that $n>(k-1)k+k$ is coprime with $k(k-1)$.   It follows that $\mathbb C$ contains a cycle containing $C_0$ of length $k+sk+t(k-1)$. We will replace $C$ by this cycle and update $k$, so that $C=C_0,\dots,C_{k-1}$ is a cycle whose length is coprime with $\ell$.

To ease notation, let us introduce the following terminology: A \textit{segment} of the cycle $C$ is any edge of the form $C_{j} \edgeo C_{i}$. The \textit{length} of a segment $C_{j} \edgeo C_{i}$ is $i-j \ \textnormal{mod}\ k$. In particular, every edge $C_{j} \edgeo C_{j+1}$ is a segment of length $1$, and there are no segments of length $0\ \textnormal{mod}\ k$ as the existence of such a segment would imply a loop in $\urgraph / \Omega$. Now, fix a segment $C_{j} \edgeo C_{i}$ of maximal length $p$; it follows that $C_i,C_j$ are contained in the cycle $C'=(C_i,C_{i+1},\dots,C_j)$ of length $k-p+1$. Set $m$ to be any number such that $m\ell = m'k-p+1$ for some $m'$ satisfying $m'-p>q-1$. Note that we can easily find such numbers by first taking $m'' > q-1+p$ such that $m''$ is a multiplicative inverse of $k$ modulo $\ell$, and setting $m':=(r\ell+p-1)m''$ for some $r$ such that $r\ell+p-1>0$,  $m:=(m'k-p+1) / \ell$. 
\bigskip 

The rest of the proof works similarly for both cases.
Let us first consider $\mu \merge \mu'$ defined by
    \begin{equation*}
        \binom{U=U_{n-1}}{U'}\ledgeo \binom{U_{n-2}}{U=U_{n-1}}\ledgeo \dots \ledgeo \binom{U_0=C_0}{U_1} \ledgeo \binom{C_{k-1}}{U_0=C_0}\ledgeo \dots\ledgeo\binom{C_{j-1}}{C_j},
    \end{equation*}
and  observe that the pair $(\mu,\mu')$ is properly separated. Append to $\mu$  the $\edge$-forward walk  which first follows the walk $(C_{j-1},\dots,C_{k-1},C_0=U_0,\dots,U_{q-1}=D_0)$ (taking $k-j+q$ steps), and then walks $m\ell$ steps along $D$, thus finishing at $D_0$; To $\mu'$ we append  the $\edge$-forward walk which starts at $C_j$, takes a step to $C_i$, takes $m'k-p$ steps on $C$ until $C_j$, then walks $k-j$ steps on the walk $(C_j,C_{j-1},\dots,C_{k-1},C_0=U_0)$, then takes $q-1$ steps along $Q$ to $U_{q-1}=D_0$, and then takes a final step from $D_0$ to $D_1$;  Note that the two walks have the same number of steps, namely $m'k-p+1+k-j+q=m\ell+k-j+q$. 

We now show that after appending the respective walk, $(\mu, \mu')$ remains properly separated. Supposing otherwise, there exists an edge from the $i$-th label $K$ of $\mu$ to the  $(i+1)$-st label $L$ of $\mu'$ for some $i$. If $i\leq k-j+1$, then this means that we get a contradiction to the choice of the segment $C_j\edgeo C_i$ by obtaining a segment of greater length $p+1$. If $i>k-j+1$, it means that $K,L$ lie in a cycle of length $m'k-p$ on $P\cup C\cup D$.
By the choice of $m'$, $m'k-p>q-1$, whence $K$ is an element of $D$, or $L$ is an element of $C$. Striving for contradiction, assume first that $K$ is an element of $D$. It follows that $L\notin C\cup \{U_1,\dots, U_{n-1}=U\}$ since these elements do not lie in the strongly connected component $\comp(V)$. Moreover, $L\notin Q\backslash (\{U_1,\dots, U_{n-1}\}\cup D)$ since $D$ was the nearest cycle reachable from $V$ via $V'$. Finally, $L\notin D$ by the choice of $m'$, whence the case $K\in D$ is not possible. On the other hand, we can also exclude the case when $K\in D$ by a similar argument.  Whence, $(\mu, \mu')$ is indeed properly separated.

 Finally, we append to $\mu$  the $\edgeo$-backward walk starting in $D_0$, walking $m\ell$ steps on $D$ back to $D_0$, and then following $Q$ backwards to $V=U_n$. Simultaneously, we add to $\mu'$  the $\edgeo$-backward walk starting in $D_1$, walking $m\ell$ steps on $D$ back to $D_1$, taking a step to $D_0$, and then following $Q$ backwards to $V'=U_{n+1}$. Clearly, $(\mu, \mu')$ remains  properly separated. By construction,  $\mu$ extends $U\edgeo V$, thus finishing the proof of Case 2.
\bigskip

\textbf{Case 3.} $\comp(U)=\comp(V)$ and $\comp(U) \neq \comp(V')$: This case can be reduced to the previous cases as follows. Take $W\in \comp(U)$ with $V \edgeo W$, and $W'$ arbitrary with $V' \edgeo W'$ (again, such $W'$ exists by  smoothness of $\urgraph / \omega$, whereas $W$ can be chosen even within $\comp(U)$  as $\comp(U)$ is smooth and of size greater than one). Since $U,V,W$ belong to the same strongly connected component, Case 1. yields a pair $(\mu_1, \mu'_1)$ for $U\edgeo V$ and $V\edgeo W$ as in the claim.  By Case 2., there exists a pair $(\mu_2, \mu'_2)$ for $V \edgeo W$ and $V' \edgeo W'$. Finally, labelling $\ledgeo$ by the labels $(V', V)$ and $(W', W)$, we get a pair $(\mu_3, \mu'_3)$ for $V' \ledgeo V$ and $W' \ledgeo W$. Now, $(\mu_1 + \mu_2 + \mu_3,\mu_1' + \mu_2' + \mu_3') $ has the required properties.
\end{proof}
    By the remark after the claim, this concludes the proof of the lemma.
\end{proof}

\alphaonpairs*

\begin{proof}
 Recall the notion $\gamma_{\xi} $ for the relation induced by an $\Omega$-orbit-labelled path $\xi$. Let us set $\alpha:=\alpha(\urgraph,\Omega)$. By  definition of $\alpha$, there exists a symmetric $\Omega$-orbit-labelled path  $\pi$ 
 such that for all $a,b\in O$ we have  $\alpha(a,b)$ if and only if there is a realisation of $\pi$ from $a$ to $b$, i.e. $(a,b) \in \gamma_{\pi} $ . Similarly, there exists such a path $\mu$ defining $\alpha|_P$  in the same way.
 \Cref{lemma:labellings} applied to $\pi$ and the orbits $O,P$ now provides us with an extension $\pi'$ of $\pi$ from $O$ to $O$,  and a relabelling $\rho'$ thereof from $P$ to $P$ such that the pair $(\pi', \rho')$ is  properly separated; similarly, applying the lemma to $\mu$ and the orbits $P,O$ we obtain an extension $\mu'$ of $\mu$ and a relabelling $\nu'$ thereof such that $(\mu',\nu')$ is properly separated, whence so is $(-\nu',-\mu')$ by the remark below \Cref{defi:properSeparation}. Now consider $\kappa:=\pi'-\nu'$ (an $\Omega$-orbit-labelled path from $O$ to $O$) and $\lambda:=\rho'-\mu'$ (an $\Omega$-orbit-labelled path from $P$ to $P$ which is a relabelling of $\kappa$). 
 
We show that the relations $\gamma_{\kappa}$ and $\gamma_\lambda$ are both $\alpha$-stable and induce a bijection on the set of $\alpha$-classes contained in $O$ and $P$, respectively.
To this end, note that for every $\alpha$-class $A\subseteq O$ and for every $a\in A$, we have $a + \gamma_{\pi'}=A$ (since $\pi'$ extends $\pi$ and by the remark below~\Cref{defi:alpha}). Furthermore, $A-\gamma_{\nu'} $ is non-empty (since $\nu'$ is realisable), and $\alpha$-stable (by~\Cref{aitm:walk} in~\Cref{lem:alphasymmetry}). In fact,  as implied by~\Cref{aitm:bijection} of~\Cref{lem:alphaFinitises}, $A-\gamma_{\nu'} $ only consists of one $\alpha$-class $B$. Hence, we have $a+ \gamma_{\kappa} =B$ for every $a\in A$. Likewise, one can show $b- \gamma_{\kappa} =A $ for every $b\in B$. Since $A\subseteq O$ was arbitrary, and $B$ was uniquely determined by $A$, it now follows that $\gamma_{\kappa}$ is $\alpha$-stable inducing a bijection on the set of $\alpha$-classes contained in $O$.  Similarly, for every $\alpha$-class $A\subseteq P$, there exists a uniquely determined $\alpha$-class $B\subseteq P$ such that $a+\gamma_{\lambda}= B $ for every $a\in A$ and $b- \gamma_{\lambda} = A$ for every $b\in B$, and the statement follows. Subsequently, we may choose  a number $k\geq 1$ such that for $\kappa^k:=\kappa + \dots + \kappa$ and $\lambda^k:=\lambda + \dots + \lambda$ (where the summands appear exactly $k$ times), both $\gamma_{\kappa^k} $ and $\gamma_{\lambda^k} $ are $\alpha$-stable relations inducing the identity bijection on the set of $\alpha$-classes contained in $O$ and $P$, respectively.

Set $R:=\gamma_{\kappa^k \merge \lambda^k}$, and observe that $\alpha\subseteq R$.  We claim that $R$ satisfies $A+R=A$ for every $\alpha$-class $A$ contained in $O$. To see this, observe first that $(\kappa,\lambda)$ is properly separated as  both $(\pi', \rho')$ and $(-\nu',-\mu')$ are. Thus, any realisation of $\kappa\merge \lambda$ starting at an element in $O$ must be a realisation of $\kappa$ (and hence end in $O$), i.e.  $(a, b)\notin \gamma_{\kappa \merge \lambda}$ for all $(a,b)\in O \times P$. By $\alpha$-stability of $\gamma_{\kappa^k}$, we now get $A+R=A$.

Switching the roles of $O$ and $P$ in the preceding arguments,  we obtain a relation $S$ with $\alpha\subseteq S$ and so that $B+S=B$ for all $\alpha$-classes $ B\subseteq P$. Finally, $R\cap S$ is the restriction of $\alpha$ to $O\cup P$ as desired.

It remains to show that $R\cap S$ is pp-definable from $\edge$ and unions of pairs of $\omega$-classes. Being extensions of the symmetric $\pi$ and $\mu$, respectively, $\pi'$ and $\mu'$ have algebraic length $0$, and hence so do $\kappa$ and $\lambda$. It follows that the induced relations $\gamma_{\kappa^k} $ and $\gamma_{\lambda^k} $ are pp-definable from $\edge$ and $\omega$-classes. Consequently,  $R = \gamma_{\kappa^k\merge\lambda^k}$ induced by  $\kappa^k\merge\lambda^k$ is pp-definable from $\edge$ and 
unions of pairs of $\omega$-classes. Arguing the same for $S$, we conclude the proof. 
\end{proof}

\subsection{Missing proofs from \texorpdfstring{\Cref{sect:reductionistic}}{Section~\ref{sect:reductionistic}}}\label{sect:missingreductionistic}

\redchar*
\begin{proof}
Clearly, any $\Omega$-reductionistic set $H$ has the stated property. For the converse, let $\Omega$ act naturally on subsets of $G$, and consider the orbit $M:=\{f(H) \mid f\in\Omega\}$ of $H$ in this action. We claim that $M$ is a partition of the underlying set $\bigcup M$. To see this, suppose $f(H)\cap g(H)\neq\emptyset$ for some $f,g\in\Omega$. Pick $x\in f(H)\cap g(H)$, and set $y:=f^{-1}(x)\in H$. Then we have $(g^{-1}\circ f)(y)\in H$, hence $(g^{-1}\circ f)(H)\subseteq H$. It follows that $g(H)\supseteq f(H)$, and similarly $f(H)\supseteq g(H)$, whence $f(H)=g(H)$. Now, the equivalence relation corresponding to the partition $M$ is $\Omega$-invariant, and $H\in M$. Adding one class $G\setminus \bigcup M$, we obtain an $\Omega$-invariant equivalence relation on $G$ which has $H$ as a class.
\end{proof}

\redfinitises*

\begin{proof}
    We will prove that $\alpha|_{\subdomain\times\subdomain}$ satisfies all items of~\Cref{lem:alphaFinitises}. To prove \Cref{aitm:finite}, observe that any weakly connected component $D$ of $\urgraph|_{\subdomain}$ is contained in a weakly connected component $D'$ of $\urgraph$ and the set $D' / \alpha$ is finite by assumption.

    To see \Cref{aitm:bijection}, observe that if $O,P\subseteq\subdomain$ are $\Omega|_\subdomain$-orbits such that $O\edgeo P$ in $\urgraph|_{\subdomain}$, then for every $\alpha|_{\subdomain\times\subdomain}$-class $A\subseteq O$, there exists an $\alpha|_{\subdomain\times\subdomain}$-class $B\subseteq P$ with $A\edgea B$ which immediately yields the desired statement. Indeed, since $O\edgeo P$, there exist $\alpha|_{\subdomain\times\subdomain}$-classes $A'\subseteq O,B'\subseteq P$ with $A'\edgea B'$. Take any $f\in\Omega$ with $f(A')=A$. Since $A\in f(\subdomain)\cap \subdomain$, this intersection is non-empty, whence $f|_\subdomain\in \Omega|_\subdomain$. As $\alpha$ finitises $(\urgraph, \Omega)$, it follows that $A=f(A')\edgea f(B')\in \subdomain$ as desired.

    Finally, \Cref{aitm:invariant} follows since $\alpha|_{\subdomain\times\subdomain}$ is invariant under $\Omega|_\subdomain$.
\end{proof}

\smoothpartreduc*

\begin{proof}
Let $H\subseteq G$ be $\alpha$-stable and $\Omega$-reductionistic, let $S\subseteq H$ be its smooth part, and let $f\in\Omega$ be so that $f(S)\cap S\neq \emptyset$. Then in particular $f(H)\cap H\neq \emptyset$, hence $f(S)\subseteq f(H)\subseteq  H$. Since $S\cup f(S)\subseteq H$ is smooth and contains $S$, it equals $S$, implying  $f(S)\subseteq S$. To see that $S$ is also $\alpha$-stable, let $a\in S$ and let $b\in G$ satisfy $\alpha(a,b)$. Then $b\in H$ by $\alpha$-stability. Picking any $f\in\Omega$ with $f(a)=b$,  the above argument yields again  $f(S)\subseteq S$, hence $b\in S$. Whence, $S$ is $\Omega$-reductionistic by~\Cref{lem:reductionistic_characterization}.
\end{proof}

\linkedcompreduc*

\begin{proof}
    Let $\fingraph := \urgraph /\alpha$. By~\Cref{lem:liftingTreeDefs}, it suffices to show that $K$ is tree pp-definable from $\fingraph$ with parameters. To this end, let $\beta$ be the $k$-linkedness-equivalence  on $J$, i.e.~relating two elements  if and only if they lie in the relation induced by any  fence of height $k$. Since $\beta$ is tree pp-definable from $\edge$, the $\beta$-class $K$ is tree pp-definable from $\edge$ together with any of its elements. Clearly, $\beta$ is  $\Omega$-invariant, and hence $K^\alpha$ is $\Omega$-reductionistic.
\end{proof}

\subsection{Missing proofs from \texorpdfstring{\Cref{sect:full}}{Section~\ref{sect:full}}}\label{sect:proofoffull}
For completeness, we now provide the full proof of~\Cref{prop:claimBK12}.
\claimBK*
\begin{proof}
Let $\fingraph := \urgraph /\alpha$. As the $\alpha$-blow-up of any subset of $J$ that is tree pp-definable from $\fingraph$ with parameters is pp-definable from $\urgraph$ with $\alpha$-classes by~\Cref{lem:liftingTreeDefs}, and $\alpha$-stable by definition, it suffices to work in the finite digraph $\fingraph$. Note that the $\alpha$-blow-up of a unary relation on $J$ is given by the union of its elements. By abuse of notation, we will identify subsets of $J$ with their $\alpha$-blow-ups.

Let $\beta$ be the $(k-1)$-linkedness  relation on $J$, and  
fix an arbitrary $\beta$-class $K'$. Recall that $K'$ is $\Omega$-reductionistic and pp-definable from $\urgraph$ and $\alpha$-classes by \Cref{l:linkedcompreduc}. We employ the following claim, which is taken from~\cite[Claim 3.11 and Proposition 3.2]{cyclicterms}:

\fullcaseoldargument*
\begin{proof}[Proof of~\Cref{claim:fullcaseoldargument}]
Let $S$ denote the smooth part of $K'$, and let $\phi(x)$ be the pp-formula asserting the existence of both an $\edge$-forward walk as well as that of an $\edge$-backward walk from $x$ in $K'$. Then $\phi$ is a tree pp-formula that defines $S$  from $\edge$ and $K'$: any element of the smooth part  clearly satisfies $\phi(x)$; on the other hand, the set of all elements satisfying $\phi(x)$ induces a smooth digraph since whenever $\phi(A)$ holds for some $A$, then it also holds for all witnesses of the existentially quantified variables.

We next observe that for any element $A\in K'$ there exists $B\in K'$ with $A\edge B$. It then follows that $K'$ contains a cycle, and hence its smooth part is non-empty. To make this observation, note there exists an $\edge$-forward walk from $A$ in $\fingraph$ of length $k-1$ by smoothness, ending at some element $C$. There also exists an $\edge$-forward walk from $A$  to $C$ of length $k$ since $\edgek$ is full. The second element $B$ of this walk then lies in the $\beta$-class of $A$ and we have $A\edge B$ as required.

Finally, pick any $A_0\edge A_1$ in $S$ and extend it to a $\edge$-forward walk $(A_0,\ldots,A_{k-1})$ in $S$; it exists by smoothness.  Since $\edgek$ is full on $\fingraph/\alpha$, there is also an $\edge$-forward walk $(B_0,B_1,\ldots,B_{k-1}, B_k)$ of length $k$ from $A_0=B_0$ to $A_{k-1}=B_k$. 
For all $i$, we have that $B_{i+1}$  is  $\beta$-equivalent to $A_i$ (this is obvious for $i=0$ from the two walks, and  can be seen for all other $i$ by extending both walks  beyond $A_{k-1}=B_k$); hence they lie in  $K'$. Since $B_0,B_k$ satisfy $\phi(x)$, so do all $B_i$; hence they lie in $S$. Together, we get that $A_0$ and $A_1$ are connected by a fence of height $k-1$ in $S$. As $A_0,A_1$ were arbitrary, it follows that every weakly connected component of $S$ is $(k-1)$-linked as desired.
\end{proof}

We may now take $K$ to be any weakly connected component of the smooth part of $K'$. Indeed, by~\Cref{l:smoothpartreduc}, the smooth part $S$ of the $\Omega$-reductionistic set $K'$ is itself $\Omega$-reductionistic. Thus, since $K$ is $(k-1)$-linked and weakly connected,~\Cref{l:linkedcompreduc} applied to $\urgraph|_S$ yield that   $K$ is $\Omega$-reductionistic. 

\end{proof}

\subsection{Missing proofs from \texorpdfstring{\Cref{sect:thm7}}{Section~\ref{sect:thm7}}}\label{sect:proofofproper}

We now fill in the details on the proof of \Cref{thm:7}.

\begin{center}
\noindent\fbox{%
    \parbox{0.95\columnwidth}{\textbf{In the rest of this section, we fix a smooth digraph $\urgraph =(\urdomain; \edge)$, an oligomorphic subgroup $\Omega$ of $\Aut(\urgraph)$, and $\alpha\subseteq\omega$ which finitises $(\urgraph, \Omega)$. Additionally, we assume that $\urgraph /\alpha$ is $k$-linked for some $k\geq 1$, that $k$ is the smallest number with this property,  and  that $\edgek\neq (\urdomain / \alpha)^2$.
    }
}}
\end{center}

We will make use of the following relation multiple times:

\begin{definition}
    Fix a tuple $\tuple t$ containing one representative of every $\alpha$-class, and let $\tilde \inj$ be the  $\Omega$-orbit of $\tuple t$. Then by $\inj$ we denote smallest $\alpha$-stable relation containing $\tilde \inj$. 
\end{definition}
For our purposes, the precise tuple $\mathbf t$ will not matter, (although different tuples yield different relations). Note that as $\Omega$ acts on the $\alpha$-classes, any tuple in $\tilde \inj$ contains  precisely one representative of every $\alpha$-class. Moreover,  $\inj$  consists of all tuples which are componentwise $\alpha$-equivalent to some tuple in $\tilde\inj$. Although this pp-definition of $\inj$ from $\tilde\inj$ uses $\alpha$, $\inj$ can be pp-defined from $\Omega$-orbits and the restrictions of $\alpha$ to $\omega$-classes: Indeed, $\tilde{\inj}$ is given by an $\Omega$-orbit, and since the projection of $\tilde\inj$ to any coordinate is contained in a single $\omega$-class $O$, componentwise $\alpha$-equivalence to a tuple in $\tilde\inj$ only requires access to $\alpha_O$ for all $\omega$-classes $O$. If $M \subseteq \urdomain$ is 
$\omega$-stable, then the \emph{projection $\prinj{M}$ of $\inj$ to $M$} is the relation obtained by projecting $I_G$ to those  coordinates which take values in $M$.

\subsubsection{Central or OR (proof of \texorpdfstring{\Cref{prop:Marcinsmagic}}{Proposition~\ref{prop:Marcinsmagic}} )}\label{sect:proofofmarcinsmagic}


\Marcinsmagic*

For the proof of \Cref{prop:Marcinsmagic}, we need to introduce the following notion: for a binary relation $R$, and a set $\subdomain\subseteq \urdomain$ that is a union of $\alpha$-classes $A_1,\dots,A_\ell$, we write $\subdomain\rplus R$ for the unary relation $$\subdomain\rplus R(x)\equiv \exists x_1,\dots,x_\ell\bigwedge\limits_{i\in[\ell]} (x_i\in A_i\wedge R(x_i,x)).$$ Note that if $H=\{A_1\}$, then $H\rplus R=H+ R$. Furthermore, observe that if $R$ is pp-definable without parameters, and $\subdomain$ is $\omega$-stable, then $\subdomain\rplus R$ is pp-definable from $\urgraph$, $\Omega$-orbits, and the restrictions of $\alpha$ to $\omega$-classes using the relation $\prinj{H}$. 

\Cref{prop:Marcinsmagic} follows immediately from the combination of \Cref{lemma:trickM,lemma:trickM2} below. 

\begin{lemma}\label{lemma:trickM}
    $\urgraph$ together with $\Omega$-orbits and the restrictions of $\alpha$ to $\omega$-classes pp-define one of the following
    \begin{enumerate}
        \item a relation that is central modulo $\alpha$, or
        \item a binary relation $R$ and a unary relation $C$ such that there exists an $\omega$-class $O$ with the property that for every $\alpha$-class $A\subseteq O$, $((A\rplus R)\cap C)-R=\urdomain$ but $((O\rplus R)\cap C)-R\neq \urdomain$. \label{itm:marcinmagicnoncentral}
    \end{enumerate}
\end{lemma}

\begin{proof}
    If $\edgek$ is central modulo $\alpha$, then we have already found a relation as in the first item. Let us therefore assume that this is not the case.
    Since $\urgraph/\alpha$ is $k$-linked and $\edgek$ is not central modulo $\alpha$, we can find $\ell\geq 1$ such that $\fence{k}{\ell}$ is full, but $\fence{k}{\ell-1}$ is not. If $\ell=1$, we set $R:=\edgek$; otherwise, we set $R:=\fence{k}{\ell-1}$, and observe that $R$     satisfies $(R-R)/\alpha=(\urdomain/\alpha)^2$, but $R/\alpha\neq (\urdomain / \alpha)^2$. For a an $\alpha$-stable set $S$, we will denote by $S^{\rplus}$ the set $(S\rplus R)\cap C$. 
    We first set $C:=\urdomain$. Now, we gradually shrink the unary relation $C$ by intersecting it with sets of the form $O^{\rplus}$, where $O$ is an $\omega$-class, keeping the property 
    that the relation $$R_C(x,y)\equiv \exists z R(x,z)\wedge R(y,z)\wedge C(z)$$ satisfies $R_C/\alpha=(\urdomain/\alpha)^2$. 
    Observe that $R_{G}/\alpha=(\urdomain/\alpha)^2$ holds indeed true, as we have assumed that $(R-R)/\alpha=(\urdomain/\alpha)^2$.  

    Let $O$ be an $\omega$-class. For every $\alpha$-class $A\subseteq O$, it holds that $A^{\rplus}-R=\urdomain$ by~\Cref{aitm:walk} in~\Cref{lem:alphasymmetry}.
    If $O^{\rplus}-R\neq \urdomain$, we are in case (2). Let us therefore assume that $O^{\rplus}-R = \urdomain$, replace $C$ by $C\cap O^{\rplus}$, and observe that $C$ is pp-definable without parameters using the relation $\prinj{O}$.

    Let us first assume that with the updated $C$, $R_C/\alpha\neq (\urdomain/\alpha)^2$. Then it follows that $O$ is contained in the $\alpha$-centre of $R_C$, and $R_C$ satisfies item (1).

    If $R_C/\alpha = (\urdomain/\alpha)^2$ and there is an $\omega$-class $O'$ for which we have not updated $C$ yet, we proceed with $O'$.

    Finally, if $R_C/\alpha = (\urdomain/\alpha)^2$ and we have already updated $C$ for all $\omega$-classes, we conclude that $R$ is itself central modulo $\alpha$. Indeed, any $\alpha$-class $A\subseteq C$ is contained in the $\alpha$-centre of $R$.
\end{proof}

\begin{lemma}\label{lemma:trickM2}
    If \Cref{itm:marcinmagicnoncentral} of~\Cref{lemma:trickM} applies, then $\urgraph$ together with   $\Omega$-orbits and the restrictions of $\alpha$ to $\omega$-classes pp-define a  relation $\OR(T,T)$ for an $\alpha$-stable TSR-relation $T$.
\end{lemma}

\begin{proof}
    Let $R,C$ and $O$ be as in item (2) of~\Cref{lemma:trickM}. As in the proof of~\Cref{lemma:trickM}, for an $\alpha$-stable relation $S$, we will denote by $S^{\rplus}$ the set $(S\rplus R)\cap C$. Observe that since $R$ and $C$ are pp-definable from $\urgraph$, $\Omega$-orbits, and the restrictions of $\alpha$ to $\omega$-classes, for all $g \in \Omega$ we have \begin{equation}\label{eq:(g)plus=(gplus)}
        g(S)^{\rplus}-R = g(S^{\rplus} -R).
    \end{equation} Let $S\subseteq O$ be a union of $\alpha$-classes which is maximal with respect to inclusion with the property that $S^{\rplus}-R=\urdomain$. In particular, $S \subsetneq O$ by the properties of $O$. Using~\eqref{eq:(g)plus=(gplus)}, for all $g \in \Omega$ we also get \begin{equation}\label{eq:g(S)+-R}
        g(S)^{\rplus}-R=\urdomain.
    \end{equation}  Let $\ell\geq 1$ be maximal such that there exists $f_{\ell} \in\Omega$ satisfying $f_{\ell} (S)\neq S$ and $|(S/\alpha)\cap (f_{\ell} (S)/\alpha)|=\ell-1$. Clearly, $\ell \leq |S/\alpha|.$
    
    Among all $f\in \Omega$ with $f(S)\neq S$ we choose one with the property that -- for a suitable choice of $\ell$-many pairwise distinct $\alpha$-classes $S_1, \dots, S_{\ell}$ contained in $S$ -- the amount $i_{f}$ of $\alpha$-classes contained in the set $$(S\cup \{f(S_1),\dots,f(S_\ell)\})^{\rplus}-R$$ is maximal.
    Note that we do not require  $|(S/\alpha)\cap (f (S)/\alpha)|=\ell-1$, and that $i_f \geq 1$. Indeed, it is enough to show that there exists a choice of $f$, and $S_1,\dots,S_\ell$ satisfying that $i_f\geq 1$. Choosing $f_\ell$ as above, and $\alpha$-classes $S_1,\dots, S_{\ell-1}\subseteq S\cap f_{\ell}(S)$, for all $\alpha$-classes $S_{\ell} \subseteq S\setminus\{S_{1}, \dots, S_{\ell-1}\}$  we get $f_{\ell}(S_\ell)\in (S \cup \{f_{\ell}(S_1), \dots, f_{\ell}(S_{\ell-1}), f_{\ell}(S_{\ell})\})  ^{\rplus}-R$ as $S^{\rplus}-R=\urdomain$. Whence, $f_\ell(S_\ell)$ witnesses that $i_{f_\ell}\geq 1$. Fix $S_1, \dots, S_{\ell} $ witnessing the maximality of $i_f$, and set $U:=S\cup \{f(S_1),\dots,f(S_\ell)\}$.  It follows that $\emptyset \neq U^{\rplus}-R\neq \urdomain$ by the choice of $S$ and $\ell$.
    Finally, let us remark that $U$ is $\alpha$-stable by definition, as is $U^{\rplus}-R$ by~\Cref{aitm:walk} in~\Cref{lem:alphasymmetry}. However, $U^{\rplus}-R$ might only be pp-definable from $\urgraph$ using $\alpha$-classes from $S\cup f(S)$. Below, we will find a suitable $k \geq 1$ such that the closure $T$ of $(U^{\rplus}-R)^k$ under $\Omega$ is not full, and we will prove that the relation $\OR(T,T)$ is pp-definable from $\urgraph$, $\Omega$-orbits, and the restrictions of $\alpha$ to $\omega$-classes without parameters.

    To this end, let us choose $k\geq 1$ so that every $(k-1)$-element subset of $\urdomain$ is included in $g(U^{\rplus}-R)$ for some $g\in\Omega$, and there exists a $k$-element subset of $\urdomain$ which is not included in $g(U^{\rplus}-R)$ for any $g\in\Omega$. Such $k$ exists as $U^{\rplus}-R\subsetneq \urdomain$ is $\alpha$-stable. In particular, if $N:= \urdomain /\alpha$, we have $k < N$, as any tuple from $\inj$ containing the representatives of all $\alpha$-classes is not contained in $U^{\rplus}-R$. Set our desired $\alpha$-stable TSR-relation $T$ to consist of all tuples $(t_1,\dots t_k)$ such that $\{t_1,\dots,t_k\}\subseteq g(U^{\rplus}-R)$ for some $g\in\Omega$. By choice of $k$, $T$ is proper.

    Assume that $(S_1, \dots,  S_n) $ is an enumeration of all $\alpha$-classes contained in $S$.   In order to construct  the $OR(T,T)$-relation, we first prove the following claim.

    \begin{claim}\label{claim:OR}
        There exists a pp-formula $\phi(\tuple x, \tuple y, \tuple v)$, where $\tuple x,\tuple y$ are of length $n$, and $\tuple v$ is of length $k$, such that all of the following items hold.
        \begin{enumerate}
            \item For all $g\in\Omega$, and for every $\tuple s,\bar{\tuple s}\in (g(S))^n$, the relation defined by $\phi(\tuple s,\bar{\tuple s},\tuple v)$ is $\urdomain^k$.
            \item For all $g, \bar g\in\Omega$ with $g(S) \neq \bar g(S)$, and for every $\tuple s \in g(S_1)\times \dots \times g(S_n) $ and $ \bar{\tuple s} \in \bar g(S_1)\times \dots \times \bar g(S_n)$, the formula $\phi(\tuple s,\bar{\tuple s},\tuple v)$ defines a subset of $T$.
            \item For every $\tuple t\in T$, there exist $g\in\Omega$ and $\tuple s \in g(S_1)\times \dots \times g(S_n) $, $ \bar{\tuple s} \in  g f(S_1)\times \dots \times gf(S_n)$ such that $\phi(\tuple s,\bar{\tuple s},\tuple t)$ holds true.
        \end{enumerate}
    \end{claim}

    \begin{proof}
        Let $u_1,\dots,u_m$ be representatives of all $\alpha$-classes in $U^{\rplus}-R$, and let $O_U$ be the $m$-orbit of $(u_1,\dots,u_m)$ under $\Omega$. Note that $m\geq k$ as $U\subseteq U^{\rplus}-R$.  We define $\phi(\tuple x, \tuple y, \tuple v)$ as follows:
\begin{equation*}
    \begin{gathered}
        \exists w_1,\dots,w_m,w'_1,\dots,w'_m,v'_1,\dots,v'_k, \\
        x_1^{w_1},\dots,x_1^{w_m},x_2^{w_1},\dots,x_n^{w_m},
        x_1^{v_1},\dots,x_1^{v_m},x_2^{v_1},\dots,x_n^{v_m}, \\
        y_1^{w_1},\dots,y_1^{w_m},y_2^{w_1},\dots,y_\ell^{w_m},
        y_1^{v_1},\dots,y_1^{v_m},y_2^{v_1},\dots,y_\ell^{v_m} \\
        O_U(w_1,\dots,w_m)\wedge O(\tuple x)\wedge O(\tuple y)\quad \wedge \\
        \left.
        \begin{array}{l}
            \bigwedge\limits_{i \in [m]} \quad \big(  
            R(w_i, w'_i) \wedge C(w'_i) \quad  \wedge \\ 
            \bigwedge\limits_{j \in [n]} \quad R(x^{w_i}_j, w'_i) \wedge \alpha|_O(x_j, x^{w_i}_j) \quad \wedge \\
            \bigwedge\limits_{j \in [\ell]} \quad R(y^{w_i}_j, w'_i)  \wedge \alpha|_O(y_j, y^{w_i}_j) 
            \big) \quad \wedge
        \end{array}
        \right\} (i) \\ 
        \left.
        \begin{array}{l}
            \bigwedge\limits_{i \in [k]} \quad \big( 
            R(v_i, v'_i) \wedge C(v'_i) \quad \wedge  \\ 
            \bigwedge\limits_{j \in [n]} \quad R(x^{v_i}_j, v'_i) \wedge \alpha|_O(x_j, x^{v_i}_j) \quad \wedge \hspace{0.4cm} \\
            \bigwedge\limits_{j \in [\ell]} \quad R(y^{v_i}_j, v'_i) \wedge \alpha|_O(y_j, y^{v_i}_j) 
            \big)
        \end{array}
        \right\} (ii)
    \end{gathered}
\end{equation*}

Observe that tuples $\tuple x $ and $\tuple y$ can take values only in $O$ by the conjunct $O(\tuple x) \wedge O(\tuple y)$ in $\phi(\tuple x, \tuple y, \tuple v)$. Moreover, the relation defined by $\phi$ is $\alpha$-stable: Indeed, as the variables  from $\tuple x$ appear only in conjuncts using the $\alpha$-stable relation $\alpha|_O$, the truth of $\phi$ depends only on the $\alpha$-classes appearing in the values of $\tuple x$. Similarly, the truth of $\phi$ depends only on the $\alpha$-classes appearing in the values of $\tuple y$. To see that $\phi$ also depends only on the $\alpha$-classes appearing in the values of $\tuple v$,   we apply ~\Cref{aitm:walk} in~\Cref{lem:alphasymmetry} to $C$ and $R$, which are pp-definable from $\urgraph$, $\Omega$-orbits, and the restrictions of $\alpha$ to $\omega$-classes.

Let us now prove that $\phi$ satisfies all items from the statement of the claim.
To prove the first item, let $g$ and $\tuple s,\bar{\tuple s}$ be as in the statement, and take $\tuple t=(t_1,\dots,t_k)\in \urdomain^k$ arbitrarily. We will define an evaluation $\val$ of the existentially quantified variables of $\phi$ witnessing $\phi(\tuple s, \bar{\tuple s}, \tuple t) $. Recall that as observed in~\eqref{eq:g(S)+-R}, $g$  satisfies $g(S)^{\rplus}-R=\urdomain$.   Since $\tuple s,\bar{\tuple s}\in (g(S))^n$, by definition of $g(S)^{\rplus}$, every element from $\urdomain$ has an $R$-neighbour contained in $C$ in common with all the elements of a set containing representatives of every $\alpha$-class appearing in $\tuple s$ and $\bar{\tuple s}$. In other words, for every $i\in[k]$ there exist tuples $(p^i_1, \dots, p^i_n)  , (\bar p^i_1, \dots, \bar p^i_n)  \in g(S)^n  $, and $a_i \in C$ satisfying \begin{itemize}
    \item for every $j\in[n]$, $p^i_j$ lies in the $\alpha$-class of $s_j$, and $(p^i_j, a_i) \in R$,
    \item for every $j \in \br{n}$, $\bar p^i_j$ lies in the $\alpha$-class of $\bar{s}_j$, and  $(\bar p^i_j, a_i) \in R$,
    \item $(t_i, a_i) \in R$.
\end{itemize}
Setting $\val(v'_i):=a_i$ for every $i\in[k]$, $\val(x_j^{v_i}):=p^i_j$ for every $(i,j)\in[k]\times \br n$, and $\val(y_j^{v_i}):=p^i_j$ for every $(i,j)\in[k] \times \br{\ell}$ , we see that $\val$ satisfies $(ii)$ in the definition of $\phi$. Similarly, we proceed for the tuple $(u_1, \dots, u_m)$: Take $(q^i_1, \dots, q^i_\ell), (\bar q^i_1, \dots, \bar q^i_n)  \in g(S)^n, $ $b_i \in C$ satisfying the corresponding bullet points from above, and evaluate $\val(w_i):=u_i,  \val(w'_i):=b_i, \val(x_j^{w_i}):=q_j^i $, and  $\val(y_j^{w_i}):=\bar q_j^i$ for all appropriate $i,j$. Now, $\val$ also satisfies $(i)$ in the definition of $\phi$, thus witnessing $\phi(\tuple s, \bar{\tuple s}, \tuple t) $. 

In order to see that $\phi$ satisfies the second item of the claim, recall that $S$ was chosen to be maximal with the property that $S^{\rplus}-R=\urdomain$, and that $g(S)$ inherits this property by~\eqref{eq:g(S)+-R} for every $g \in \Omega$. Thus, if $g,\bar g\in\Omega$ satisfy $g(S)\neq \bar g(S)$, $\tuple x$ takes a value $\tuple s \in g(S_1)\times \dots \times g(S_n), $ and $\tuple y$ takes a value $ \bar{\tuple s} \in \bar g(S_1)\times \dots \times \bar g(S_n)$, then there needs to exist $i\in[\ell]$ such that $\bar g(\bar{s}_i)\notin g(S)$ by the choice of $\ell$. Indeed, otherwise, $g^{-1} \bar g(\bar{s}_i)\in S$ for every $i\in[\ell]$, whence $|S\cap g^{-1} \bar g(S)|\geq \ell$, and $g^{-1} \bar g(S)\neq S$ in contradiction to the choice of $\ell$. Hence, $g(S)\subsetneq g(S)\cup \bar g(S_1)\cup\dots\cup\bar g(S_\ell)$, whence $\tuple v$ asserting $\phi(\tuple s, \bar{\tuple s}, \tuple v)$ cannot take all values in $\urdomain^k$ by the maximality of $S$.
Let us show that, in fact, $\phi(\tuple s,\bar{\tuple s},\tuple v)$ can take values only in $T$. To this end, let $V\subseteq O$ be the set containing precisely the $\alpha$-classes of $s_1,\dots,s_n,\bar{s}_1,\dots,\bar{s}_\ell$; we have already derived that $V^{\rplus}-R\neq \urdomain$. Let now $\tuple t\in \urdomain^k$ be such that $\phi(\tuple s,\bar{\tuple s},\tuple t)$ holds true, and let $\val'$ be an evaluation of the existentially quantified variables witnessing this. By the first conjunct of $\phi$, we get that there exists $h\in\Omega$ such that $\val'(w_i)=h(u_i)$ for every $i\in[m]$. Let us define a new evaluation $\val(x):=h^{-1}(\val'(x))$. It follows that 
$h^{-1}(\tuple s)\in (h^{-1}g(S))^n$, and $h^{-1}(\bar{\tuple s})\in (h^{-1}\bar{g}(S))^n$.
Moreover, we have $U^{\rplus}-R\subseteq h^{-1}(V^{\rplus}-R)$. Indeed, by the conjuncts $(i)$ in the definition of $\phi$, $\val'(w_i)\in V^\rplus -R$ for every $i\in[m]$. Whence, we obtain $\val(w_i)=u_i\in h^{-1}(V^\rplus - R)$, and since $u_1,\dots,u_m$ are representatives of all $\alpha$-classes in $U^\rplus - R$, ~\Cref{aitm:walk} in~\Cref{lem:alphasymmetry} yields that $U^{\rplus}-R\subseteq h^{-1}(V^{\rplus}-R)$. 
    
By the maximality of $U^{\rplus}-R$, it now follows that 
\begin{equation}\label{eq:h(V)+-R}
   U^{\rplus}-R= h^{-1}(V^{\rplus}-R) = h^{-1}(V)^{\rplus} -R,
\end{equation} where the last equation is justified by~\eqref{eq:(g)plus=(gplus)}.
 Looking at the conjuncts $(ii)$ in $\phi$, we get  for every $i\in[k]$ that $h^{-1}(t_i)\in h^{-1}(V)^{\rplus} -R. $ By~\eqref{eq:h(V)+-R}, therefore   $ h^{-1}(t_i) \in U^{\rplus}-R$, whence $t_i  \in h(U^{\rplus}-R)$ for all $i \in \br k$, and we conclude $\tuple t\in T$ as desired.

Finally, let us prove the third item of the claim. To this end, let $\tuple t\in T$, which means that there exists $g\in \Omega$ such that $\tuple t\in (g(U^\rplus -R))^k$. Let $\tuple s \in g(S_1)\times \dots \times g(S_n) $, $ \bar{\tuple s} \in  gf(S_1)\times \dots \times gf(S_n)$ be arbitrary; we claim that $\phi(\tuple s,\bar{\tuple s},\tuple t)$ holds true. We will define a satisfying evaluation $\val$ of the existentially quantified variables of $\phi$ as follows. Similar to the proof of the first item, let us observe that for every $i\in[k]$, the fact that $t_i\in g(U^\rplus -R)=g(U)^{\rplus}-R $ is witnessed by the existence of $(p^i_1, \dots, p^i_{n}) \in g(S_1) \times \dots \times g(S_n) ,  (\bar p^i_1, \dots, \bar p^i_{\ell}) \in gf(S_1) \times \dots \times gf(S_{\ell})$, and $a_i \in C$ such that all of the following hold:
    \begin{itemize}
        \item for every $j\in[n]$, $p^i_j$ lies in the $\alpha$-class of $s_j$, and $(p^i_j, a_i) \in R$,
        \item for every $j \in \br{\ell}$, $\bar p^i_j$ lies in the $\alpha$-class of $\bar{s}_j$, and  $(\bar p^i_j, a_i) \in R$,
        \item $(t_i, a_i) \in R$.
    \end{itemize}
Setting $\val(v'_i):=a_i$ for every $i\in[m]$, $\val(x_j^{v_i}):=p^i_j$ for every $(i,j)\in[m]\times \br n$, and $\val(y_j^{v_i}):=\bar p^i_j$ for every $(i,j)\in[m] \times \br{\ell}$ , we see that $\val$ witnesses that $(\tuple s,\bar{\tuple s},\tuple t)$ satisfies the conjunct $(ii)$ in the definition of $\phi$.

In order to finish the construction of $\val$, it remains to find values for the variables appearing in the conjunct $(i)$. To this end, let us first set $\val(w_i):=g(u_i)$ for every $i\in[m]$, and take, similarly as above, $(q^i_1, \dots, q^i_n) \in S_1 \times \dots \times S_n  , (\bar q^i_1, \dots, \bar q^i_{\ell} ) \in f(S_1)\times \dots \times f(S_{\ell}) $ and $b_i \in C$ witnessing that $(u_1,\dots,u_m)\in U^\rplus - R$, i.e. such that all of the following hold:
    \begin{itemize}
        \item for every $j\in[n]$, $q^i_j$ lies in the $\alpha$-class of $g^{-1}(s_j)$, and $(q^i_j, b_i) \in R$,
        \item for every $j \in \br{\ell}$, $\bar q^i_j$ lies in the $\alpha$-class of $g^{-1}(\bar{s}_j)$, and   $(\bar q^i_j, b_i) \in R$,
        \item $(u_i, b_i) \in R$.
    \end{itemize}
Let us set $\val(w'_i):=g(b_i)$ for every $i\in[m]$, $\val(x_j^{w_i}):=g(q^i_j)$ for every $(i,j)\in[m]\times \br n$, and $\val(y_j^{v_i}):=g(\bar q^i_j)$ for every $(i,j)\in[m] \times \br{\ell}$. As $\alpha$ is invariant under $\Omega$ by assumption, $(q^i_j,g^{-1}(s_j))\in \alpha$ implies $(g(q^i_j),s_j)\in \alpha$, and similarly for $\bar q^i_j$ and $\bar s_j$. Moreover, $R-R$ is invariant under $\Omega$ since $R$ is pp-definable from $\urgraph$, $\Omega$-orbits, and the restrictions of $\alpha$ to $\omega$-classes. Hence, we see that $\val$ witnesses that $(\tuple s,\bar{\tuple s},\tuple t)$ satisfies also the conjunct $(i)$ in the definition of $\phi$ which finishes the proof. 
\end{proof}

    In order to finish the proof of~\Cref{lemma:trickM2}, let $\phi$ be the formula provided by~\Cref{claim:OR}, let for every $i\in[n]$, $s_i\in S_i$ be arbitrary, let $O_f$ be the $2n-$orbit of $(s_1,\dots,s_n,f(s_1),\dots,f(s_n))$ under $\Omega$ , and let $O_S$ be the $n$-orbit of $(s_1,\dots,s_n)$. We claim that the relation $W(\tuple v, \tuple v')$ defined  by the  following pp-formula defines our desired $\OR(T,T)$-relation:
    \begin{equation*}
        \begin{aligned}
            \exists x_1,\dots,x_n,y_1,\dots,y_n,z_1,\dots,z_n\\
            O_f(x_1,\dots,x_n,z_1,\dots,z_n)\wedge O_S(y_1,\dots,y_n)\wedge\\
            \phi(x_1,\dots,x_n,y_1,\dots,y_n,v_1,\dots,v_k)
            \wedge \\\phi(y_1,\dots,y_n,z_1,\dots,z_n,v'_1,\dots,v'_k)
        \end{aligned}
    \end{equation*}

    Let us first prove that $T\times \urdomain^k, \urdomain^k\times T\subseteq W$. Indeed, to see that $T\times \urdomain^k\subseteq W$, take $(\tuple v,\tuple v')\in T\times \urdomain^k$ arbitrary. As $\tuple v\in T$, by the third item of~\Cref{claim:OR} there exist $g \in \Omega$, $\tuple r \in g(S_1) \times \dots \times g(S_n), \bar{\tuple r} \in gf(S_1) \times \dots \times gf(S_n)$ such that $\phi(\tuple r, \bar{\tuple r}, \tuple v)$ holds true. Evaluating $x_i\mapsto g(r_i)$, $y_i,z_i\mapsto gf(r_i)$, and applying item (1) from~\Cref{claim:OR}, we get that $(\tuple v,\tuple v')\in W$.  
    The proof that $A^k\times T\subseteq W$ is symmetric.

    On the other hand, if $(\tuple v,\tuple v')\in W$, take an evaluation $\val$ witnessing this fact. It follows that all of the tuples $(\val(x_1),\dots,\val(x_n))$, $(\val(y_1),\dots,\val(y_n))$, and $(\val(z_1),\dots,\val(z_n))$ lie in $O_S$, whence there exist $g_1,g_2,g_3\in\Omega$ such that $(\val(x_1),\dots,\val(x_n))\in (g_1(S))^n$, $(\val(y_1),\dots,\val(y_n))\in (g_2(S))^n$, and $(\val(z_1),\dots,\val(z_n))\in (g_3(S))^n$. Moreover, since $(\val(x_1),\dots,\val(x_n),\val(y_1),\dots,\val(y_n))$ lies in $O_f$, and $f(S)\neq S$, it follows that either $g_1(S)\neq g_2(S)$, or $g_2(S)\neq g_3(S)$. Applying item (2) from~\Cref{claim:OR}, we get that in the former case, $\tuple v\in T$, and in the latter one, $\tuple v'\in T$. 
\end{proof}

\subsubsection{From central to OR (proof of \texorpdfstring{\Cref{prop:trickT2}}{Proposition~\ref{prop:trickT2}})}\label{sect:proofoftrickT2}

In the case of \Cref{itm:marcinsmagiccentralcase} of \Cref{prop:Marcinsmagic}, we proceed by constructing a relation $\OR(D_L, D_R)$ for some unary, $\omega$-stable relations $D_L$ and $D_R$.

 \trickT*

\begin{proof}

    Let $C$ denote the $\alpha$-centre of $R$. Note that by $\Omega$-invariance, $C$ is $\omega$-stable. Since $\urgraph$ is weakly connected, there exist $\omega$-classes $O_{\oin}, O_{\out}$ such that $O_{\out} \cap C = \emptyset$, $O_{\oin} \subseteq C$, and $O_{\oin} \edgeo O_{\out}$ or $O_{\out} \edgeo O_{\oin}$. 

    \begin{claim}\label{claim:centralmagic} Then there exist
    $\omega$-classes $O'_{\oin}\subseteq C, O'_{\out}\subseteq \urdomain\setminus C $, a realisable $\Omega$-orbit-labelled path $\pi$ from $O_{\oin} $ to $O'_{\out}$, and a realisable relabelling $\pi'$ of $\pi$ from $O_{\out} $ to $O'_{\oin}$ such that $(\pi, \pi')$ are properly separated. 
\end{claim}
\begin{proof}
    We assume that $O_{\out} \edgeo O_{\oin}$ as the other case works symmetrically. 
    We consider two cases. If there is an $\edgeo$-walk of the form $O_{\oin}\edgeo O_1 \edgeo \dots \edgeo O_{n-1} \edgeo O_n $ such  that for all $ 1 \leq i \leq n-1$, $O_i$ is contained in $C$, but $O_n \cap C = \emptyset$, then we set $O'_{\oin}:=O_{n-1} $ and $O'_{\out}:= O_n$, and the orbits-labelled paths $\pi, \pi'$ given by \begin{equation*}
        \begin{pmatrix}
            O_{\oin} \\ O_{\out}
        \end{pmatrix} \edgeo \begin{pmatrix}
            O_1 \\ O_{\oin}
        \end{pmatrix} \edgeo \begin{pmatrix}
            O_2 \\ O_1
        \end{pmatrix} \edgeo \dots \edgeo \begin{pmatrix}
            O_n \\ O_{n-1} 
        \end{pmatrix} = \begin{pmatrix}
            O'_{\out} \\ O'_{\oin}
        \end{pmatrix}
    \end{equation*} satisfy the required properties. In the second case, if there does not exist such a walk, then the idea of the proof is the same as in the proof of \Cref{lemma:labellings}. Since no walk as above exists,
    there must be an a $\edgeo$-cycle within $C$. Take the first such cycle $P_0 \edgeo P_1 \edgeo \dots \edgeo P_{k}=P_0$ with $P_i \subseteq C$ for all $i$ that is $\edgeo$-reachable from $O_{\oin}$, i.e. 
    \begin{equation*}
        O_{\out} \edgeo O_{\oin} \edgeo O_1 \edgeo \dots \edgeo O_{n} \edgeo P_0 \edgeo \dots \edgeo P_{k}=P_0,
    \end{equation*} 
    $O_i \subseteq C$ for all $i \geq 1$,  and $O_i$ is not contained in any $\edgeo$-cycle for all $i <n$. Consider the induced subgraph of $\urgraph/\omega$ on the vertices  $P_0, \dots, P_{k-1}$. As in the proof of \Cref{lemma:labellings}, a \textit{segment} of the cycle $P_0 \edgeo \dots \edgeo P_{k}=P_0$ is any edge of the form $P_{j} \edgeo P_{i}$. The \textit{length} of a segment $P_{j} \edgeo P_{i}$ is $i-j \ \textnormal{mod}\ k$. In particular, every edge $P_{j} \edgeo P_{j+1}$ is a segment of length $1$ (where all indices $i$ of orbits labelled with $P_i$ are from now on considered mod $k$). Take a segment $P_{j} \edgeo P_{i}$ of maximal length, and consider the $\Omega$-orbit-labelled paths $\pi, \pi'$ given by
    \begin{equation*}
        \begin{pmatrix}
            O_{\oin} \\ O_{\out}
        \end{pmatrix} \edgeo  \begin{pmatrix}
            O_1 \\ O_{\oin}
        \end{pmatrix} \edgeo \dots \edgeo  \begin{pmatrix}
            P_{i+1} \\ P_{i}
        \end{pmatrix} \ledgeo \begin{pmatrix}
            P_{i}  \\ P_{j}
        \end{pmatrix} 
    \end{equation*}

    Note that $(\pi, \pi')$ is indeed properly separated, as $P_{j} \edgeo P_{i+1} $ would give either a segment of greater length, or a loop in $\urgraph /\omega$. With the same reasoning, we can append to $\pi$ and $\pi'$, respectively, the $\Omega$-orbit-labelled paths
    \begin{equation*}
    \begin{pmatrix}
            P_{i}  \\ P_{j}
        \end{pmatrix} \ledgeo \begin{pmatrix}
            P_{i-1}  \\ P_{j-1}
        \end{pmatrix} \ledgeo \dots \ledgeo \begin{pmatrix}
            P_{0}  \\ P_{j-i}
        \end{pmatrix} \ledgeo \begin{pmatrix}
            O_{n}  \\ P_{j-i-1}
        \end{pmatrix}.
    \end{equation*} 
 Finally, we append
\begin{equation*}
    \begin{pmatrix}
            O_{n}  \\ P_{j-i-1}
        \end{pmatrix} \ledgeo \begin{pmatrix}
            O_{n-1}  \\ P_{j-i-2}
        \end{pmatrix} \ledgeo \dots \ledgeo \begin{pmatrix}
            O_{\out}  \\ P_{j-i-(n-1)}
        \end{pmatrix}.
    \end{equation*}
Indeed, $(\pi, \pi')$ remains properly separated as $O_i$ is not contained in any $\edgeo$-cycle  for all $i <n$. We set $O'_{\oin}:= P_{j-i-(n-1)} \subseteq C $, and  $O'_{\out}:= O_1 \subseteq \urdomain\setminus C $. Now, $(\pi, \pi')$ has the required properties.
\end{proof}

Let $\gamma_{\pi \merge \pi'} \subseteq (O_{\out} \cup O_{\oin}) \times (O'_{\oin} \cup O'_{\out})$ denote the  relation induced by $\pi \merge \pi'$ on $(O_{\out} \cup O_{\oin}) \times (O'_{\oin} \cup O'_{\out})$. To construct a relation $\OR(D_L, D_R)$ for $\alpha$-stable sets $D_L, D_R \subseteq \urdomain$, we symmetrically built an $\alpha$-stable left-hand-side $S_L\subseteq \urdomain \times (O_{\out} \cup O_{\oin})$ and an $\alpha$-stable right-hand-side $S_R\subseteq (O_{\out} \cup O_{\oin}) \times \urdomain$ by letting  
 \begin{align*}
        S_L:= R^{-1} + \alpha|_{O_{\out} \cup O_{\oin}} \\
        S_R:  \alpha|_{O_{\out} \cup O_{\oin}} + \gamma_{\pi \merge \pi'} + R.
     \end{align*}
  
Observe that both $S_L$ and $S_R $  are indeed $\alpha$-stable, and  satisfy \begin{gather}   
 \begin{aligned}\label{eq:S_LR}
    \urdomain \times O_{\oin} \subseteq S_L \\
    O_{\out} \times \urdomain \subseteq S_R.    
\end{aligned} \end{gather}  
Observe that now  $(A,B)\in S_L$ -- where $S_L$ is considered as a relation on $\urdomain/\alpha$ -- holds true if and only if $A \times B \subseteq S_L$. The same is valid for $S_R$ by $\alpha$-stability.

Let $D_L \subseteq \urdomain$ be the union of all $\omega$-classes $O$ such that $S_L|_{O\times O_{\out}}=O\times O_{\out}$. Similarly, $D_R \subseteq \urdomain$ is the union of all $\omega$-classes $P$ such that $S_R|_{O_{\oin} \times P}=O_{\oin} \times P$. By definition, $D_L$ and $D_R$ are indeed $\alpha$-stable. Note that by $\alpha$-stability of $S_L$ and~\Cref{aitm:neighbours} in~\Cref{lem:alphasymmetry}, $A \subseteq D_L$ if and only if for all $B \subseteq O_{\out} $ it holds that $(A,B) \in S_L$. Analogously, $B \subseteq D_R$ if and only if $(A,B) \in S_R$ for all $A \subseteq O_{\oin} $. In particular, \begin{gather}
    \begin{aligned}\label{eq:centresofSL}
    D_L \times O_{\out} \subseteq S_L \\ 
O_{\oin} \times D_R \subseteq S_R .
\end{aligned}\end{gather}  Clearly, by definition of $C$, we have $C \subseteq D_L, D_R$.  Furthermore, we claim that $D_L,D_R\neq \urdomain$ i.e. both $S_L$ and $S_R$ are proper. In fact, we can even show: 
\begin{claim}
For every $B \subseteq O_{\out}  $ there exists $A_B$ such that $(A_B,B) \notin S_L$. For every $A \subseteq O_{\oin} $ there exists $B_A \subseteq \urdomain$ with $(A, B_A) \notin S_R$.
\end{claim}
\begin{proof}
    As $O_{\out} \cap C = \emptyset$, for every $b \in O_{\out}$ there exists  $a_b \in \urdomain$ such that $(b, a_b) \notin (R/\alpha)^{\alpha} $. By $\alpha$-stability of $S_L$, this gives us  $(A_B,B) \notin S_L$, where $A_B$ and $B$ denote the $\alpha$-classes of $a_b$ and $b$, respectively.  
    
    To prove the similar statement for $S_R$, recall that $O'_{\out} \cap C = \emptyset$. Thus,  for every $a' \in O'_{\out} $ we can take $b_{a'} \in \urdomain$ such that $(a',b_{a'}) \notin (R/\alpha)^{\alpha} $. For every $a' \in O'_{\out}$,  choose $a \in O_{\oin} $ with $\gamma_{\pi \merge \pi'}(a, a')$. Observe that as $\alpha$ finitises $(\urgraph, \Omega)$ by 
    assumption, if $(a'_1, a'_2) \notin \alpha$, then also $(a_1, a_2) \notin \alpha$. Thus, by picking $a \in O_{\oin} $ for every $a' \in O'_{\out} $ as above, we eventually pick a representative of every $\alpha$-class in  $O_{\oin}$. Fix elements $a, a', b_{a'}$, accordingly.  Since $(\pi, \pi')$ is properly separated, we have $a + \gamma_{\pi \merge \pi'} \subseteq O'_{\out}$. In particular,  again by ~\Cref{lem:alphaFinitises}, we get $b_{a'} \notin a + \gamma_{\pi \merge \pi'} + (R/\alpha)^{\alpha}$. By $\alpha$-stability of $S_R$, henceforth $(A, B_{A'}) \notin S_R$, where $A$ and $B_{A'}$ denote the $\alpha$-classes of $a$ and $b_{a'}$, respectively.
\end{proof}

For $(A_B,B) \notin S_L$ as above, we now want to thin out $S_L$ to a relation $S_L^B$ with the property that the only right-$S_L^B$-neighbours of elements from $A$ are all elements of $O_{\oin}$. Similarly, we proceed for $S_R$ and tuples $(A, B_A) \notin S_R$. 

\begin{claim}\label{cl:S_L^B}
    For every $\alpha$-class $B\subseteq O_{\out}$, $S_L$,
    $\Omega$-orbits, restrictions of $\alpha$ to $\omega$-classes, and $B$ pp-define a relation $S_L^{B}$ with $D_L\times B, \ \urdomain\times O_{\oin}\subseteq S_L^{B}$, and such that for every $A\notin B-S_L$, it holds that $A+S_L^{B}=O_{\oin}$.  
    
    For every $\alpha$-class $A\subseteq O_{\oin}$, $S_R$,
    $\Omega$-orbits, restrictions of $\alpha$ to $\omega$-classes, and $A$ pp-define a relation $S_R^{A}$ with $A \times D_R,\ O_{\out} \times \urdomain\subseteq S_R^{A}$, and such that for every $B\notin A+ S_R$, it holds that $B-S_R^{A}=O_{\out}$. 
\end{claim}
\begin{proof}
We will only prove the first statement of the claim. The proof of the second statement works dually and will be omitted. Fix an $\alpha$-class $B \subseteq O_{\out} $. Clearly, we have $D_L \subseteq B - S_L$. If for some (equivalently: for all) $A \notin B -  S_L$ there exists $C \in O_{\out} $ with $(A,C) \in S_L$, then  $D_L \neq B-S_L$.  Indeed, in this case we apply~\Cref{aitm:neighbours} from~\Cref{lem:alphasymmetry} twice to get $A'$ with $(A, A') \in \omega$ such that $(A', B) \in S_L$, and $D \in O_{\out}$ with $(A', D) \notin S_L$.  Hence,  $A'\subseteq (B-S_L) \setminus D_L$. Recall that $\urdomain\times O_{\oin}\subseteq S_L$ by~\Cref{eq:S_LR}. Hence, if $B-S_L=D_L$, then $S_L^B:=S_L$ already satisfies the required properties. 

Suppose that $D_L \neq B-S_L$, and let \begin{equation*}
    Z_L^B:=\bigcap_{A \in B-S_L}A+S_L.
\end{equation*} Note that $Z_L^B$ is pp-definable from $S_L$, $B$, and the projection of $\inj$ to $B-S_L$. Clearly, $B\cup O_{\oin} \subseteq Z_L^B $. Furthermore, since $D_L \neq B-S_L$, we have $Z_L^B \neq O_{\oin} \cup O_{\out}$. Indeed, if $A$ is such that $(A,B) \in S_L$, and $(A,C) \notin S_L$ for some $C \in O_{\out} $, then $C \cap Z_L^B = \emptyset$. Define \begin{equation*}
    S_L^B:=S_L|_{\urdomain \times Z_L^B},
\end{equation*}  and let us verify that $S_L^B$ satisfies the required properties. Firstly, $D_L \times B \subseteq S_L^B$ holds true since $B \subseteq Z_L^B$, and~\Cref{eq:centresofSL}.
Secondly, $\urdomain\times O_{\oin}\subseteq S_L^{B}$ by~\Cref{eq:S_LR} and $O_{\oin} \subseteq Z_L^B$. Lastly, take $A\notin B- S_L$. We need to show that $A+S_L^{B}=O_{\oin}$. Assume that $A_1, \dots, A_m$ is an enumeration of all $\alpha$-classes from the $\omega$-class of $A$ contained in $B-S_L$. Suppose towards a contradiction that there exists an $\alpha$-class $B' \subseteq O_{\out}  $ with $(A , B') \in  S_L^{B}.$ Since $B' \subseteq Z_L^B$, we also have $(A_{j}, B') \in S_L $ for every $j\in\{1,\dots,m\}$. As $A$ was chosen such that $(A , B) \notin B- S_L $, this contradicts~\Cref{aitm:neighbours} in~\Cref{lem:alphasymmetry}. Hence, $(A, C) \in S_L^B$ holds true if and only if $C \in O_{\oin}, $ which concludes the proof of the claim.
\end{proof}

For every pair $(B,A)$ of $\alpha$-classes with $B\subseteq O_{\out} $ and $ A \subseteq  O_{\oin}$, we now let \begin{equation*}
    S^{(B,A)}:=S^B_L+S^A_R.
\end{equation*}  Observe that since $D_L\times B\subseteq S_L^{B}$ and $O_{\out} \times \urdomain\subseteq S_R^{A}$, we have  $D_L\times \urdomain \subseteq S^{(A,B)}$. Similarly, because of  $\urdomain\times O_{\oin}\subseteq S_L^{B}$ and $A \times D_R\subseteq S_R^{A}$ we get $ \urdomain\times D_R\subseteq S^{(A,B)}$. In other words, for every $B\subseteq O_{\out} $ and  $ A \subseteq  O_{\oin}$, \begin{equation*}
    \OR(D_L, D_R) \subseteq S^{(B,A)}.
\end{equation*}  Furthermore, $S^{(B,A)} \neq \urdomain^2$, as for every $A'\notin B- S_L $ and for every $ B'\notin A+ S_R$ we have $(A'+S_L^B) \cap (B'-S_R^A) = \emptyset $, i.e. $(A',B')\notin S^B_L+S^A_R = S^{(A,B)}$. Finally, we let \begin{equation}\label{eq:S}
    S:= \bigcap_{B \times A \subseteq O_{\out} \times O_{\oin}} S^{(B,A)},
\end{equation} and claim that $S = \OR(D_L, D_R)$. It is only left to show that $S \subseteq \OR(D_L, D_R)$. Thus, take $(x,y) \in S$, and suppose first that $x \notin D_L$, i.e.  there exists an $\alpha$-class $B \subseteq O_{\out} $ such that $x \notin B -  S_L$. Thus, by~\Cref{cl:S_L^B}, $x+ S_L^B = O_{\oin}$. Since in particular, $(x,y) \in \bigcap_{A \subseteq O_{\oin}} S_L^B + S_R^A$, for every  $\alpha$-class $A \subseteq O_{\oin} $ there exists $z_A$ with $(x, z_A) \in S_L^B$, and $(z_A, y) \in S_R^A$. Since $x+ S_L^B = O_{\oin}$, we must have $z_A \in O_{\oin} $ for every $A$, giving us $y - S_R^A \neq O_{\out}$. Again by~\Cref{cl:S_L^B}, we obtain $y \in A + S_R$ for every $A\subseteq O_{\oin}$, i.e. $y \in D_R$.  Similarly, for any $y\notin D_R$, we get $y-S=D_L$, and it follows that $S=\OR(D_L,D_R)$ as desired. 

 We remark that $S$ as defined in~\Cref{eq:S} is in fact pp-definable from $R,  \inj$, and $\alpha|_{O_{\out} \cup O_{\oin}}$. Indeed, let $(t_1,\dots,t_{2n})\in \prinj{O_{\out} \cup O_{\oin}}$ be arbitrary. Without loss of generality, we assume that $t_1, \dots, t_n \in O_{\out} $, and $t_{n+1}, \dots, t_{2n} \in O_{\oin}$.  
 Then $S$ is defined  by the formula \begin{multline*}
     \exists x_1, \dots, x_n, y_1, \dots, y_n \\  \prinj{O_{\out} \cup O_{\oin}}(x_1, \dots, x_n, y_1, \dots, y_n) \wedge 
     \bigwedge_{(i,j)\in \br{n} \times \br{n} } S_L^{x_i} + S_R^{y_j},
 \end{multline*}
where $S_L^{x_i}$ is obtained by replacing every occurrence of $A$ in the pp-definition of $S_L^A$ by the $\alpha$-class of $x_i$ -- more formally, we replace every occurrence of $A$ by a fresh existentially quantified variable $x'$ and for every variable introduced in this way, we add the conjunct $\alpha|_{O_{\oin} \cup O_{\out}} (x_i,x')$, and $ S_R^{y_j}$ by replacing every occurrence of $B$ in the definition of $S_R^B$ by the $\alpha$-class of $y_j$.
\end{proof}

\subsubsection{Improving OR}\label{sect:symmetries} 

In this section, we will prove the remaining lemmata needed for the proof of \Cref{thm:7}. The following statement allows us to replicate the corresponding proofs from~\cite{symmetries}.

\begin{lemma}\label{l:ppdefalphastable} Let $\findomain:= \urdomain /\alpha$, and let $R_1, \dots, R_m$ be relations on $\findomain$. If $S \subseteq \findomain^n$ is pp-definable from $R_1, \dots R_m$, then its $\alpha$-blow-up $S^{\alpha}\subseteq \urdomain^n $ is pp-definable from $R_1^{\alpha} , \dots ,R_m^{\alpha}.$ 
\end{lemma}
\begin{proof}
Let us first observe that if $S$ has a quantifier-free pp-definition, then the statement is trivial. Assume now
that $S(\tuple x)$ is defined by the formula $\phi(\tuple x)\equiv\exists \tuple y T(\tuple x, \tuple y), $ where $T$ is quantifier-free pp-definable from $R_1^{\alpha} , \dots ,R_m^{\alpha}.$ Then $S^{\alpha}$ is defined by $\phi^{\alpha}(\tuple x) \equiv \exists \tuple y T^{\alpha}(\tuple x, \tuple y)$. Indeed, by definition of $T^{\alpha}$, a tuple $(B_1, \dots, B_k)$ witnesses $\phi(A_1, \dots, A_n)$ if and only if every tuple $(b_1, \dots, b_k) \in B_1 \times \dots \times B_k$ witnesses $\phi^{\alpha}(a_1, \dots, a_n) $ for every $(a_1, \dots, a_n) \in A_1 \times \dots \times A_n$. 
\end{proof}

\paragraph{Decreasing the arity of \texorpdfstring{$\OR(T,T)$}{OR(T,T)}}

In this section, we show that for an arbitrary $\alpha$-stable TSR-relation $T$, the relation $\OR(T,T)$ pp-defines either the relation $\OR(U,U)$ for some unary, $\alpha$-stable $U$, or the $\alpha$-blow-up of a proper equivalence on a finite set. As in \cite{symmetries}, the following definitions play a major role:

\begin{definition}
A subdirect relation $R\subseteq \urdomain^n$ for $n \geq 2$ is \emph{P-central} if its \emph{P-centre}, the set given by $$\left\{a \in \urdomain: (a, a_2, \dots, a_3) \in R \ \textnormal{for all} \ a_2, \dots, a_n\right\},$$ is non-empty and proper. \end{definition} Note that this definition coincides with the one of a central relation if $n=2$.

\begin{definition}
A proper relation $R\subseteq \urdomain^n$ for $n > 2$ is \emph{PQ-central} if its projection to any pair of coordinates is full, and the formula $$\phi(x,y)\equiv \forall z_3, \dots, z_n \ R(x,y,z_3, \dots, z_n)$$ defines an equivalence $\sigma$ on $\urdomain$. In this case, $\sigma$ is called the \emph{central equivalence} of $R$.\end{definition}

We will first show that an $\alpha$-stable TSR-relation pp-defines an $\alpha$-stable TSR-relation that is, additionally, P-central or PQ-central. The statement is proved in \cite[Lemma 14]{symmetries}, nevertheless,  we provide the proof for the convenience of the reader.

\begin{lemma}\label{l:lemma14}
Let $R \subseteq \urdomain^n$ be a proper $\alpha$-stable TSR-relation.
    \begin{enumerate}
        \item If $n\geq 3$, then $R$ equality-free pp-defines a proper $\alpha$-stable TSR-relation $T$ which is P-central or PQ-central.
        \item If $n=2$ and $R$ is additionally linked, then $R$ equality-free pp-defines a proper $\alpha$-stable TSR-relation $T$ which is P-central or PQ-central.
    \end{enumerate} 
\end{lemma}
\begin{proof} Let $\fingraph:= \urgraph/\alpha$. By~\Cref{l:ppdefalphastable} and $\alpha$-stability of $R$, it suffices to pp-define from $R$ --  considered as a relation on $J$ -- a suitable TSR-relation on $\findomain$ as its $\alpha$-blow-up will then satisfy the required properties. By finiteness of $\fingraph:= \urgraph/\alpha$, we may follow exactly the proof of Lemma~14 of \cite{symmetries}. For convenience of the reader, we restate the whole proof.

We repeatedly use the subsequent construction on $J$: For $\ell \geq n-1$, we define  
\begin{equation}\label{eq:R^k} R^{[\ell]}(x_1, \dots, x_{\ell}) = \exists y \bigwedge _{1\leq i_1<\dots <i_{n-1}\leq \ell} R(y, x_{i_1},\dots,x_{i_{n-1}}).\end{equation}

Observe that $R^{[\ell]}$ is defined by an equality-free pp-formula. 
We proceed in a series of steps. In each step, if $R^{[n]}=R$, we stop. Otherwise, we consider two cases: If for all $\ell \geq n$, we have $R^{[\ell]}= \findomain^{\ell}$, we stop. In the remaining case, we take the smallest $\ell\geq n$ such that $R^{[\ell]}\neq \findomain^{\ell}$. Note that in this case, $R^{[\ell]}$ remains a proper TSR-relation. Indeed, $R^{[\ell]}$ is totally symmetric as $R$ is.  By choice of $\ell$ and total reflexivity of $R$, we have $R^{[\ell-1]} = J^{\ell-1}$. By total reflexivity and total symmetry of $R$, this implies that $R^{[\ell]}$ is totally reflexive too.   We substitute $R^{[\ell]}$ for $R$, thus   changing $n$ to $\ell$, and continue by repeating the process.    

In order to see that this process must stop eventually, let $N:=|\findomain|$. We claim that the arity of the relation we consider in the last step is at most $N$. First, note that by the pigeon-hole principle, any TSR-relation on $J$ of arity greater than $N$ is full. In particular, $n \leq N$. Assume that we have not stopped with $R$ of  arity $n$, i.e. $R^{[n]}\neq R $, and there exists a minimal $ n \leq \ell \leq N$ with $R^{[\ell]}\neq \findomain^{\ell}$. Let us first consider the case $\ell=n$. Since $R$ is totally reflexive, we have $R \subseteq R^{[n]}$ as witnessed for every tuple by its first entry, whence $R\subsetneq R^{[n]} \subsetneq \findomain^n$. By substituting $R^{[n]}$ for $R$, we have thus added at least one element to $R\subseteq \findomain^n$. If $\ell >n$, substituting $R^{[\ell]}$ for $R$ increases the arity of $R$. Both steps can only be done finitely many times maintaining $R\subsetneq \findomain^n$ and $n \leq \ell\leq N$. 

    Assume that the process stops with $R$ of arity $n \leq N$, i.e. either $R^{[n]}= R $, or $R^{[\ell]}= \findomain^{\ell}$ for all $\ell \geq n$. Let us first assume the latter. In particular, $R^{[N]}$ contains a tuple $(A_1, \dots, A_N) $ enumerating all elements of $\findomain$. Let  $A$ be the value of $y$ in~\Cref{eq:R^k} witnessing $(A_1, \dots, A_N) \in R^{[N]}$. \begin{claim}
        $R$ is P-central, and $A$ belongs to the P-centre of $R$.
    \end{claim} 
    \begin{proof}
        We need to show that $(A, B_1, \dots, B_{n-1}) \in R $ for all $(B_1, \dots, B_{n-1}) \in \findomain^{n-1}$. Indeed, if $(B_1, \dots, B_{n-1}) \in \findomain^{n-1} $ is a tuple such that $B_i=B_j $ for some $i \neq j$, then $(A, B_1, \dots, B_{n-1}) \in R $ by total reflexivity of $R$.  
    Otherwise, if $(B_1, \dots, B_{n-1}) \in \findomain^{n-1} $ is an injective tuple, then we can extend it to a tuple $(B_1, \dots, B_{n-1}, B_n, \dots, B_N )$  enumerating all elements of $\findomain$. Since $R$ is totally symmetric and $A$ witnesses $(A_1, \dots, A_N) \in R^{[N]}$, we must also have $(A, B_1, \dots, B_{n-1}) \in R$ by looking at the appropriate conjunct in~\Cref{eq:R^k}. 
    \end{proof} 

    Let us now assume $R^{[n]}= R $. We first show  that in this case, $n >2$. Indeed, if $n=2$, then by assumption, $R$ would be linked. Note that if $(A_1, A_2), (A_2, A_3) \in R$, then also $(A_1, A_3) \in R$ since $ R$ is symmetric and $A_2$ witnesses $(A_1, A_3) \in R^{[2]} =R$. In other words, $R$ is an equivalence on $\findomain$. By linkedness, this contradicts $R \neq \findomain^2$. Hence, we must have $n >2$. 
    \begin{claim}
        $R$ is PQ-central.
    \end{claim}
    \begin{proof}
      The projection to any two coordinates of $R$ is full as $R^{[n-1]}=\findomain^{n-1} $.  Let $\sigma$ be defined by
      \begin{equation*}
        \{ (A, A') \in \findomain^2 \mid \forall A_1, \dots, A_{n-2} (A, A', A_1, \dots, A_{n-2}) \in R\}.
    \end{equation*} We want to show that $\sigma$ is an equivalence. Indeed, $\sigma$ is reflexive and symmetric since $R$ is totally symmetric and totally reflexive. It remains to show that it is transitive. Take $(A_1, A_2), (A_2, A_3) \in \sigma.$ As before, $A_2$ witnesses that $(A_1, A_3, B_1, \dots, B_{n-2})\in R^{[n]}= R $ for all $(B_1, \dots, B_{n-2})$, and $\sigma$ is indeed an equivalence. By definition, $R$ is thus PQ-central.
\end{proof}  
Applying~\Cref{l:ppdefalphastable} to the relation obtained, we get a pp-definition of its $\alpha$-blow-up, which satisfies the desired properties. 
\end{proof}

The following auxiliary result is proved in the same way as Lemmata 15 and 16 in~\cite{symmetries}. 

\begin{lemma}\label{lemma:15+16}
    An $\alpha$-stable TSR P-central (PQ-central) relation $R$, $\Omega$-orbits, and the restrictions of $\alpha$ to $\omega$-classes equality-free pp-define the  centre (the  central equivalence) of $R$.
\end{lemma}

\begin{proof}
    We prove the lemma for the case where $R$ is PQ-central. The proof for the P-central case is almost identical, and will be omitted. The proof of this case might however be found in~\cite[Lemma 15]{symmetries}. Let $n$ be the arity of $R$, and let $N$ be the arity of $\inj$. Then the formula
    \begin{multline*}
        \sigma(x_1,x_2)\equiv \exists y_1,\dots,y_N \inj(y_1,\dots,y_N)\wedge \\\bigwedge\limits_{1\leq i_1,\dots,i_{n-2}\leq N} R(x_1,x_2,y_{i_1},\dots,y_{i_{n-2}})
    \end{multline*}
    pp-defines the central equivalence $\sigma$ of $R$.
\end{proof}

Moreover, We will use the following  statement from~\cite{symmetries}.

\begin{lemma}\cite[Lemma 42]{symmetries}\label{lemma:42}
    Let $R, S, R'$ be relations such that $R'$ has an equality-free pp-definition in $R$. Then $\OR(R,S)$ pp-defines $\OR(R', S)$.
\end{lemma}

Finally, we are in the position to  prove the main result of this section. We follow the proofs of Lemmata 21 and 22 in~\cite{symmetries}, and provide the arguments  for the convenience of the reader. Observe that any $\alpha$-stable equivalence $\sigma$ on $\urdomain$ is the $\alpha$-blow-up of $\sigma/\alpha$ on $\urdomain / \alpha$, which is a finite set by assumption.

\ltwentyonetwo*

\begin{proof}
    Let us denote the arity of $T$ by $n$. If $T$ is unary, we are done.

    If $n=2$ and $T$ is not linked, the transitive closure $\sigma$ of $T$ is pp-definable. Indeed, $T$ is $\alpha$-stable, whence for $\sigma$ to be the transitive closure of $T$, it is enough to verify that $\sigma / \alpha$ is the transitive closure of $T / \alpha$, which is a relation on a finite set. Observe that the transitive closure of a symmetric and reflexive binary relation $T/\alpha$ on a finite set is equality-free pp-definable -- if $T$ is itself not transitive, then replacing $T$ with the relation defined by $\exists z T(x,z)\wedge T(z,y)$ yields a relation which contains strictly more pairs than $T$, and is contained in its transitive closure. As the base set of $T/\alpha$ is finite, the process must stop eventually. Moreover, $\sigma$ is itself $\alpha$-stable, and can be defined without equality from $T$. We apply~\Cref{lemma:42} to $R=S:=T, R':=\sigma$, and get a pp-definition of $\OR(T,\sigma)$ from $\OR(T,T)$. Permuting the coordinates and applying~\Cref{lemma:42} again, we get a pp-definition of $\OR(\sigma,\sigma)$.

    If $n=2$ and $T$ is linked, or $n\geq 3$, then we apply~\Cref{l:lemma14}, and get an $\alpha$-stable TSR-relation $R$ which is $P$-central or $PQ$-central, and that has an equality-free pp-definition in $T$. Applying~\Cref{lemma:15+16}, $R$ equality-free pp-defines its centre $C$ which is itself $\alpha$-stable (if $R$ is P-central) or its $\alpha$-stable central equivalence $\sigma$ (if $R$ is $PQ$-central). Again, applying~\Cref{lemma:42} twice, we get a pp-definition of $\OR(C,C)$, or of $\OR(\sigma,\sigma)$, respectively, and we are done.
\end{proof}

\paragraph{From \texorpdfstring{$\OR(D_L, D_R)$}{OR(DL, DR)} to \texorpdfstring{$\OR(U,U)$}{OR(U,U)}}

We will show that from a relation $\OR(D_R,D_L)$, where $D_L,D_R$ are unary, $\alpha$-stable relations, we can get a relation $\OR(U,U)$ for some unary, $\omega$-stable $U$. We  use the following reformulation of~\cite[Lemma 19]{symmetries}.

\begin{lemma}\label{lemma:19}
    A non-empty  $C\subsetneq \urdomain$ and $\edge$ tree pp-define a non-empty  $\subdomain\subsetneq \urdomain$ such that $\urgraph|_\subdomain$ is smooth. If $C$ is $\omega$-stable, then so is $\subdomain$.
\end{lemma}

\begin{proof}
    Apply~\cite[Lemma 19]{symmetries} to $\urgraph/\omega$ and $C/\omega$, and observe that the formula $\psi$ defining $\subdomain$ is a tree pp-formula.
\end{proof}

The following is an adjusted version of~\cite[Lemma 20]{symmetries} to our setting.

\ltwenty*

\begin{proof}
    We proceed as in~\cite{symmetries}, the only difference being the precise formula defining our desired $\OR(U,U)$-relation. We provide the complete proof for the convenience of the reader.

    Let us assume that the first item does not apply. Recall that $k$ was chosen to be the smallest number such that $\urgraph/\alpha$ is $k$-linked, and $\edgek\neq (\urdomain / \alpha)^2$.
    Let $\psi$ be the tree pp-formula defining $\subdomain\subsetneq \urdomain$ from~\Cref{lemma:19} applied to $D_L$. Since $D_L$ is easily seen to be pp-definable from $\OR(D_L,D_R)$ and $\inj$, it in particular follows that $\urgraph|_{\subdomain}$ is not $k$-linked. Let $\beta$ be the $k$-linkedness equivalence relation on $\subdomain/\alpha$. Let $\phi$ be the  tree pp-formula defining the $k$-linkedness relation on $\urgraph/\alpha$. Note that the same formula $\phi$ pp-defines a  relation on $\urdomain$ which is reflexive and symmetric. Moreover, this relation is full modulo $\alpha$ since $\urgraph/\alpha$ is $k$-linked. 
    Let now $\phi'$ be the formula obtained from $\phi$ by restricting all variables to $\subdomain$, i.e. $\phi'$ is a  pp-definition of $\beta$ in $(\urgraph|_{\subdomain})/\alpha$.

    Still following the proof of~\cite[Lemma 20]{symmetries}, we define $\phi''$ to be the formula obtained from $\phi'$ by replacing every restriction $\subdomain(x)$ on an existentially quantified variable by the formula $\psi(x)$. Note that if we remove in $\phi''$ all the restrictions of the quantified variables to $D_L$, we obtain a formula which defines a full relation in the factor $\subdomain/\alpha$.

    Now, we successively remove the conjuncts $D_L(x)$ in $\phi''$, and we stop when we arrive at a formula with a quantified variable $x$ such that if we remove the conjunct $D_L(x)$, $\phi''$ defines in $(\urgraph|_{D_L})/\alpha$ a subset of $(D_L / \alpha)^2$ strictly larger than $\beta$, but with the conjunct $D_L(x)$, it defines $\beta$. Making $x$ free, we get a tree pp-definition of a ternary relation $S$ satisfying
    \begin{itemize}
        \item $\exists x S(y,z,x)$ defines a subset of $(\subdomain/\alpha)^2$ strictly larger than $\beta$, and
        \item $\exists x S(y,z,x)\wedge D_L(x)$ defines $\beta$ in $(\subdomain/\alpha)^2$.
    \end{itemize}
     Let us now find $a,b\in\urdomain$ arbitrary such that for their $\alpha$-classes $a/\alpha,b/\alpha$ we have $(a/\alpha,b/\alpha)\notin \beta$, and such that there exists $c\in\urdomain$ with $S(a,b,c)$. Note that $c\notin D_L$, whence $c+\OR(D_L, D_R)=D_R$.

    Assume without loss of generality that $\inj$ contains a tuple $(a,b,a_3,\dots,a_N)$. We set $R:=\OR(D_L, D_R)$, and define a pp-formula with free variables $x_1,x_2$ as follows
    \begin{equation*}
        \begin{aligned}
            \exists y_0,y_1,y_2,X_1,X_2,z_3,\dots,z_N
            \\
            \inj(y_0,y_2,z_3,\dots,z_N)\wedge
            S(y_0,y_1,X_1)\wedge \\
            R(X_1,x_1)\wedge
            S(y_1,y_2,X_2)\wedge R(X_2,x_2)
        \end{aligned}
    \end{equation*}
     We now  claim that $T=\OR(U,U)$ for $U:=D_R$.

    Let us first take $(d_1,d_2)\in \urdomain^2$ such that $d_1\in U$ (the case $d_2\in U$ is symmetric), and let us find an evaluation witnessing that $(d_1,d_2)\in T$.
    We evaluate $y_0\mapsto a, y_1,y_2\mapsto b, X_1\mapsto c$, $X_2$ to an element $d$ of $D_L$ such that $S(b,b,d)$ holds (note that this is possible as the relation defined by $\phi''$ on $H$ is reflexive). Finally, we evaluate $z_1,z'_1\mapsto a, z_2,z'_2\mapsto b$, and $z_i\mapsto a_i$ for $3\leq i\leq n$. It is easy to check that this is an evaluation witnessing that $(d_1,d_2)\in T$. Indeed, the first three conjuncts as well as the conjuncts containing $S$ hold by the choice of $a,b,c,d$; moreover, the other conjuncts hold since $R$ is an $\OR(D_L,D_R)$-relation.

    For the opposite direction, let $(d_1,d_2)\in T$ be such that $d_1\notin U$; we aim to show that $d_2\in U$ (the case $d_2\notin U$ being completely symmetric again), and let $\val$ be an evaluation of the quantified variables witnessing that $(d_1,d_2)\in T$. As $(\val(X_1),d_1)\in R$, it follows that $\val(X_1)\in D_L$, whence $(\val(y_0)/\alpha,\val(y_1)/\alpha)\in \beta$. Note that $(\val(y_0)/\alpha,\val(y_2)/\alpha)\notin\beta$ -- indeed, $(\val(y_0),\val(y_2),\val(z_3),\dots,\val(z_N))\in \inj$ by the third conjunct, whence there exists $g\in\Omega$ such that $(g(\val(y_0)),a)\in \alpha$ and $(g(\val(y_2)),b)\in \alpha$, and we know that $(a/\alpha,b/\alpha)\notin \beta$. As $(\val(y_0)/\alpha,\val(y_1)/\alpha)\in \beta$, it follows that $(\val(y_1)/\alpha,\val(y_2)/\alpha)\notin \beta$ by transitivity of $\beta$. Now, $\val(X_2)\notin D_L$, and since $(\val(X_2),d_2)\in R$, it holds that $d_2\in U$ as desired.
\end{proof}

\paragraph{The final step}

In this section, we derive one of the two items of~\Cref{thm:7} from the fact that $\urgraph$ defines the relation $\OR(U,U)$ for a unary, $\omega$-stable $U$. We  prove the following modification of~\cite[Lemma 23]{symmetries}.

\ltwentythree*

\begin{proof}
    The  structure of the proof is the same as in~\cite[Lemma 23]{symmetries}, but we provide the full proof for the convenience of the reader. Let us first observe that $R:=\OR(U,U)$ pp-defines $U$ by the formula $R(x,x)$. Moreover, for any $\ell\geq 1$, the relation $\OR(U^\ell,U^\ell)(x_1,\dots,x_\ell,y_1,\dots,y_\ell)$ is pp-definable from $R$ by the formula 
    $$
    \bigwedge\limits_{i\in[\ell]}\bigwedge\limits_{j\in[\ell]} \OR(U,U)(x_i,y_j)$$
    
    As in the proof of~\Cref{lemma:20}, we take the tree pp-formula $\psi$ from~\Cref{lemma:19} which uses $U$ to define a non-empty $\omega$-stable set $\subdomain\subseteq \urdomain$ such that $\urgraph|_\subdomain$ is smooth. We assume that $(\urgraph|_H)/\alpha$ is not $k$-linked; let $\beta$ be the $k$-linkedness equivalence on $\subdomain/\alpha$.

    We take the tree pp-formula $\phi''$ as in the proof of~\Cref{lemma:20}, and we stop removing the conjuncts $U(x)$ when the binary relation defined by the formula on $\subdomain/\alpha$ becomes linked. By making a variable free, we get a tree pp-definition of a ternary relation $S$ which satisfies
    \begin{itemize}
        \item $\exists x S(y,z,x)\wedge U(x)$ is contained in a non-trivial equivalence on $\subdomain/\alpha$ (namely, in the $k$-linkedness equivalence), and
        \item $\exists x S(y,z,x)$ is not contained in any proper equivalence on $B/\alpha$.
    \end{itemize}

    Observe moreover that both binary relations above are reflexive on $\subdomain/\alpha$ as both of them contain $\beta$. Moreover, we can find a formula over the signature consisting of a single binary symbol $Q$ which, depending on the interpretation of the symbol $Q$, pp-defines the symmetric and transitive closure of both of these relations.

    Let us modify this formula as follows. We replace every occurrence of the conjunct $Q(y,z)$ by $S(y,z,x_i)$ for a new free variable $x_i$, and we obtain a tree pp-definition of a relation $T$ such that
    \begin{itemize}
        \item $\exists x_1,\dots,x_n T(y,z,x_1,\dots,x_n)$ defines $(\subdomain/\alpha)^2$, and
        \item $\exists x_1,\dots,x_n T(y,z,x_1,\dots,x_n)\wedge \bigwedge\limits_{i\in[n]}U(x_i)$ defines a proper equivalence $\sigma$ on $\subdomain/\alpha$.
    \end{itemize}
   
    Consider the equivalence $\sigma$ on $\subdomain/\alpha$. Note that the pp-definition of $\sigma$ does not necessarily lift to a pp-definition of its $\alpha$-blow-up $\sigma^\alpha$ on $\subdomain$ since $\sigma$ is a binary relation (\Cref{lem:liftingTreeDefs} enables us to lift only definitions of unary relations). However, it suffices to find $P\subseteq\subdomain$ such that $\alpha|_P$ is pp-definable, and such that $\sigma^{\alpha}|_{P\times P}$ is not full. To this end, let us distinguish two cases: if there exists an $\omega$-class $O$ such that $\sigma^{\alpha}|_{O\times O}$ is not full (i.e. $O/\alpha$ intersects at least two $\sigma$-classes), then we set $P:=O$; otherwise, every $\sigma^{\alpha}$-class is $\omega$-stable. Since $\sigma^{\alpha}$ is a non-trivial equivalence, and $\urgraph$ is weakly connected, there must exist two $\omega$-classes $O_1,O_2$ with $O_1\edgeo O_2$, and such that $O_1/\alpha$ and $O_2/\alpha$ are contained in two different $\sigma$-classes. In this case, we set $P:= O_1\cup O_2$.

    Let us now observe that by existentially quantifying $x_1,\dots,x_n,y',z'$ in the formula
    $$\alpha|_P(y',y)\wedge \alpha|_P(z',z)\wedge T(y',z',x_1,\dots,x_n),$$ we obtain a formula defining $P$, while by existentially quantifying the same variables in the formula $$\alpha|_P(y',y)\wedge \alpha|_P(z',z)\wedge T(y',z',x_1,\dots,x_n)\wedge \bigwedge\limits_{i\in[n]}U(x_i),$$ we obtain a formula defining $\sigma^{\alpha}$ on $P$.

    As we have observed above, $\OR(U,U)$ pp-defines $\OR(U^n,U^n)$. Hence, the following formula $\chi(u,v,y,z)$ pp-defines the $\alpha$-blow-up of the restriction of $\OR(\sigma, \sigma)$ to $P/\alpha$: \begin{equation*}
        \begin{aligned}
            \exists x_1,\dots,x_n,y_1,\dots,y_n,u',v',y',z'\\
            \OR(U^n,U^n)(x_1,\dots,x_n,y_1,\dots,y_n)\wedge\\
            \alpha|_P(u',u)\wedge \alpha|_P(v',v)\wedge \alpha|_P(y',y)\wedge \alpha|_P(z',z)\wedge\\
            T(u',v',x_1,\dots,x_n)\wedge T(y',z',y_1,\dots,y_n).
        \end{aligned}
    \end{equation*} 
    Finally, observe  that  $P/\alpha$ is finite  as $\alpha$ finitises $(\urgraph, \Omega)$ by assumption.
\end{proof}

\subsection{Proof of \texorpdfstring{\Cref{prop:Siggerscrap}}{Proposition~\ref{prop:Siggerscrap}}}\label{sect:proofsiggerscrap}

The proof of~\Cref{prop:Siggerscrap} below again uses the  finitising equivalence $\alpha(\urgraph, \Omega)$.   In this case, knowledge of the particular shape of $\urgraph/\Omega$ allows us to pp-define $\alpha(\urgraph, \Omega)$ on the entirety of the domain of $\urgraph$. 
\Siggerscrap*

\begin{proof}
    Note that all unions of pairs of $\omega$-classes are pp-definable from $\urgraph$ and $\omega$-classes. Indeed $0\cup 1$ is equal to $(0+\edgeo)+\edgeo$, $0\cup 2=1+\ledgeo$, and $1\cup 2=0+\edgeo$. In other words, the $2$-conservative expansion of $\urgraph$ with respect to $\Omega$ is pp-definable from $\urgraph$. \Cref{l:alphaonpairs} yields that $\urgraph$ pp-defines $\alpha:=\alpha(\urgraph,\Omega)$ on all pairs of $\omega$-classes. Since $\alpha$ finitises $(\urgraph, \Omega)$ by \Cref{lem:alphaFinitises}, the edge relation $\edgeo$ induces a bijection between $\edgeo$-adjacent $\omega$-classes, whence the following formulae define the same relation in $\urgraph / \alpha$:
    \begin{itemize}
        \item $\phi(x,y)\equiv \exists z (x\edge z)\wedge (z\edge y)$
        \item $\phi(x,y)\equiv \exists z (x\edge z)\wedge (z\edge y)\wedge z\in 0\cup 1$
    \end{itemize}

    From the definition of $\alpha$ using symmetric $\Omega$-orbit-labelled paths, it now easily follows that the following formula pp-defines  $\alpha(x,y)$ on $G$:
    \begin{multline*}
        \exists u,u',v,v' ((x\edge u)\wedge \alpha|_{0\cup 1}(u,u')\wedge \\(u'\ledge y))\wedge ((x\ledge v)\wedge \alpha|_{0\cup 1}(v,v')\wedge (v'\edge y)).
    \end{multline*}

    By~\cite[Corollary 3.10]{wonderland}, $\urgraph$ pp-constructs the finite digraph $\subgraph:=\urgraph / \alpha$. Moreover, as $\alpha$ finitises $(\urgraph, \Omega)$, it is $\Omega$-invariant, whence  $\Omega$ acts on $\alpha$-classes. Since $\subgraph$ is a finite digraph such that $\subgraph / \Omega=\urgraph / \Omega$ has algebraic length $1$, \cite[Theorem 2]{symmetries} finishes the proof.
\end{proof}

\end{document}